%% file: da.tex
\documentclass[11pt,a4paper,austrian,titlepage,
chapterprefix,headsepline,parskip,pdftex,
,pointlessnumbers,bibtotoc]{scrreprt}

\makeatletter 
\renewcommand{\paragraph}{\@startsection
   {paragraph} 
   {4} 
   {0mm} 
   {-\baselineskip} 
   {0.1\baselineskip} 
   {\normalfont\normalsize\bfseries}} 
\makeatother 

\makeatletter 
\renewcommand{\subparagraph}{\@startsection
   {subparagraph} 
   {5} 
   {0mm} 
   {-\baselineskip} 
   {0.1\baselineskip} 
   {\normalfont\normalsize\bfseries}} 
\makeatother 

\usepackage{setspace}
\usepackage[pdftex]{graphicx}

\usepackage[pdftex]{color}

\usepackage[colorlinks=true,
    linkcolor=black,
    citecolor=black,
    pagecolor=black,
    urlcolor=black,
    breaklinks=true,
    bookmarksnumbered=true,
    hypertexnames=false,
    pdfpagemode=UseOutlines,
    pdfview=FitH,
    plainpages=false,
    pdfpagelabels,
    bookmarks=true,
    linktocpage=true]{hyperref}

\hypersetup{pdfauthor={Vorname Nachname},
    pdftitle={Testtitel},
    pdfsubject={Diplomarbeit},
    pdfkeywords={},
    pdfcreator={pdfLaTeX with hyperref (\today})}

\newcommand{\internerLink}[1]{\hyperref[#1]
{Siehe \ref*{#1}~\nameref{#1} auf S.~\pageref{#1}}}

\newcommand{\ffinternerLink}[1]{\hyperref[#1]
{Siehe S.~\pageref{#1}ff}}

\newcommand{\xinternerLink}[1]{\hyperref[#1]
{\ref*{#1}~\nameref{#1} auf S.~\pageref{#1}}}

\newcounter{cfootnotecounter}

\renewcommand*{\othersectionlevelsformat}[1]{%
\llap{\csname the#1\endcsname\autodot\enskip}}

\setkomafont{chapter}{\Huge}

\usepackage[automark]{scrpage2}

\clearscrheadings \clearscrplain \clearscrheadfoot
\ohead{\pagemark}
\ihead{\headmark}
\usepackage{array}

\addtokomafont{chapter}{\sffamily}
\addtokomafont{sectioning}{\rmfamily}

\usepackage[english]{babel}
\usepackage[latin1]{inputenc}
\usepackage[T1]{fontenc}
\usepackage{ae}

\usepackage{url}

\usepackage{amsmath, amsthm, amssymb}
\newtheorem{thm}{Theorem}[section]

\newtheorem{lemma}[thm]{Lemma}
\newtheorem{prop}[thm]{Proposition}

\newtheorem{definition}[thm]{Definition}
\newtheorem{example}[thm]{Example}

\newtheorem{remark}[thm]{Remark}
\newtheorem{notation}[thm]{Notation}

\newcommand{\ve}[1]{\textit{\textbf{#1}}}
\newcommand{\abs}[1]{\left|{#1}\right|}

\newcommand{\sign}[1]{\textnormal{sign}(#1)}
\newcommand{\HH}[1]{\textnormal{H}(#1)}
\renewcommand{\H}[2]{\textnormal{H}_{#1}(#2)}
\newcommand{\M}[2]{M\left(#1,#2\right)}
\newcommand{\Mp}[2]{M^+\left(#1,#2\right)}

\newcommand{\R}{\mathbb R}
\newcommand{\Q}{\mathbb Q}

\newcommand{\Z}{\mathbb Z}
\newcommand{\N}{\mathbb N}

\newcommand{\ie}{i.e.,\ }

\usepackage{rotating}
\newcommand{\shuffle}{\, \raisebox{1.2ex}[0mm][0mm]{\rotatebox{270}{$\exists$}} \,}
\usepackage{algorithm}
\usepackage{algorithmicx}
\usepackage{algpseudocode}
\usepackage{graphicx}
\usepackage{pdfpages}

\newcounter{mmacnt}
\def\restartmma{\setcounter{mmacnt}{0}}
\restartmma \catcode`|=\active
\def|#1|{\mathrm{#1}}
\catcode`|=12
\newenvironment{mma}{
 \par\smallskip
 \catcode`|=\active
 \parskip=0pt\parindent=0pt 
 \small
 \def\In##1\\{%
   \def\linebreak{\hfill\break\null\qquad}%
   \refstepcounter{mmacnt}
   \hangindent=2.5em\hangafter=0
   \leavevmode
   \llap{\tiny\sffamily In[\arabic{mmacnt}]:=\kern.5em}%
   \mathversion{bold}\footnotesize$\displaystyle##1$\normalsize
   \mathversion{normal}\par
 }%
 \def\Print##1\\{%
   \def\linebreak{\hfill\break}%
   \hangindent=2.5em\hangafter=0
   \leavevmode ##1\par}%
 \def\Out##1\\{%
   \def\linebreak{$\hfill\break\null\hfill$}%
   \kern\abovedisplayskip\par
   \hangindent=2.5em\hangafter=0
   \leavevmode
   \llap{\tiny\sffamily Out[\arabic{mmacnt}]=\kern.5em}
   \footnotesize$\displaystyle##1$\normalsize\hfill\null\par
   \kern\belowdisplayskip
 }%
 \def\Warning##1##2\\{%
   \def\linebreak{\hfill\break}%
   \hangindent=2.5em\hangafter=0
   \leavevmode
   {\scriptsize##1 : ##2}\par}%
}{%
 \par\smallskip
}

\usepackage{qtree}

\usepackage{color}
\usepackage{framed}

\newenvironment{fshaded}{%
\MakeFramed {\FrameRestore}
}%
{\endMakeFramed}

\newenvironment{fmma}[1]{\definecolor{shadecolor}{rgb}{1,1,1}%
\definecolor{framecolor}{rgb}{0,0,0}%
\begin{fshaded}\text{\bf HarmonicSums session.}\begin{mma}#1}{\end{mma}\end{fshaded}}

\begin{document}

\pagenumbering{roman}
\includepdf{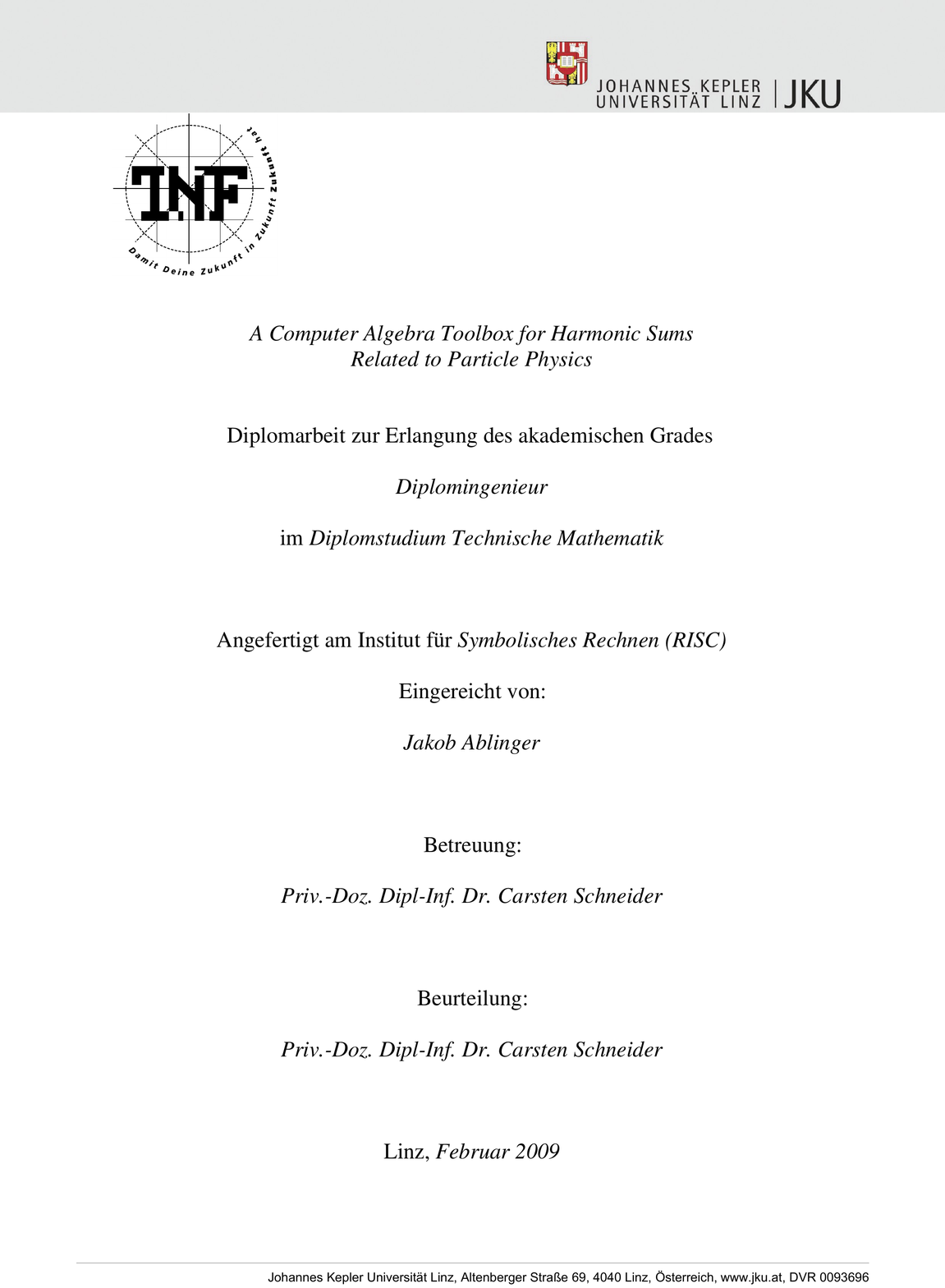}

\pagenumbering{Roman}

\chapter*{Eidesstattliche Erklärung}
\markright{Eidesstattliche Erklärung}

Ich erkläre hiermit an Eides statt, dass ich die vorliegende
Arbeit selbstständig und ohne fremde Hilfe verfasst, andere als
die angegebenen Quellen nicht benützt und die den benutzten
Quellen wörtlich oder inhaltlich entnommenen Stellen als solche
kenntlich gemacht habe.

Linz, am 26. Februar 2009,

\hfill Jakob Ablinger

\chapter*{Kurzfassung}
\markright{Kurzfassung}

In dieser Arbeit präsentieren wir das Computeralgebra Paket \ttfamily HarmonicSums \rmfamily und dessen theoretischen Hintergrund für die Verarbeitung von harmonischen Summen und damit verbundenen Objekten wie zum Beispiel Euler-Zagier Summen und harmonische Polylogarthmen. Harmonische Summen und verallgemeinerte harmonische Summen treten als Spezialfälle der sogenannten d'Alembert Lösungen von Rekurrenzen auf. Wir zeigen, dass die harmonischen Summen eine quasi-shuffle Algebra bilden und präsentieren eine Methode zum Auffinden algebraisch unabhängiger harmonischer Summen. Außerdem definieren wir eine Differentiation auf den harmonischen Summen über eine Erweiterung der Mellin Transformation. Dabei treten neue Relationen zwischen den harmonischen Summen auf. Zusätzlich stellen wir einen Algorithmus vor, der be- stimmte verschachtelte Summen in Ausdrücke bestehend aus harmonischen Summen umschreibt.
Durch nicht triviale Beispiele illustrieren wir, wie diese Algorithmen in Zusammenarbeit mit dem Summationspaket \ttfamily Sigma \rmfamily die Berechnung von Feynmanintegralen unterstützen.

\chapter*{Abstract}
\markright{Abstract}

In this work we present the computer algebra package \ttfamily HarmonicSums \rmfamily and its theoretical background for the manipulation of harmonic sums and some related quantities as for example Euler-Zagier sums and harmonic polylogarithms. Harmonic sums and generalized harmonic sums emerge as special cases of so-called d'Alembertian solutions of recurrence relations. We show that harmonic sums form a quasi-shuffle algebra and describe a method how we can find algebraically independent harmonic sums. In addition, we define a differentiation on harmonic sums via an extended version of the Mellin transform. Along with that, new relations between harmonic sums will arise. Furthermore, we present an algorithm which rewrites certain types of nested sums into expressions in terms of harmonic sums. We illustrate by nontrivial examples how these algorithms in cooperation with the summation package \ttfamily Sigma \rmfamily support the evaluation of Feynman integrals.

\tableofcontents

\chapter*{Notations}
\markright{Notations}

  \begin{tabbing}
    \hspace*{2.5cm}\= \kill
    $\N$ \> $\N=\{1,2,\ldots\},$ natural numbers\\
    $\Z$ \> integers \\
    $\Q$ \> rational numbers\\
    $\R$ \> real numbers\\
    $S_{a_1,a_2,\ldots}(n)$ \> harmonic sum; see page \pageref{defHsum}\\
    $Z_{a_1,a_2,\ldots}(n)$ \> Euler-Zagier  sum; see page \pageref{defEZsum}\\
    $\H{m_1,m_2,\ldots}{x}$ \> harmonic polylogarithm; see page \pageref{abpoly2}\\
    $\mathcal{S}_p(n)$ \>the set of polynomials in the harmonic sums with integer indices; see page \pageref{sp}\\
    $\mathcal{S}(n)$ \>the set of polynomials in the harmonic sums; see page \pageref{s}\\
		$\mathcal{I}_p$ \> an ideal of $\mathcal{S}_p(n)$; see page \pageref{idealp}\\
    $\mathcal{I}$ \> an ideal of $\mathcal{S}(n)$; see page \pageref{ideal}\\
    $\shuffle$ \> the shuffle product; see pages \pageref{abshuffle1} and \pageref{hshuffpro}\\
    $*$ \> the quasi-shuffle product; see page \pageref{quasishuffprodef}\\
    $A^*$ \> the free monoid over $A$; see page \pageref{noncomalg}\\
    $A^+$ \> the set of non-empty words over $A$; see page \pageref{abnonemp}\\
    $R\left\langle A\right\rangle$ \> the set of non-commutative polynomials over $R$; see page \pageref{noncomalg}\\
    Lyndon($A$) \> the \textit{Lyndon} words over $A$; see page \pageref{ablyn}\\
    $\ell()$ \> $\ell(\ve w)$ gives the length of a word $\ve w$\\
    $l_n$ \> the number of \textit{Lyndon} words; see page \pageref{wittform}\\
    $\sign$ \> $\sign{a}$ gives the sign of the number $a$\\
    $\abs{\ }$ \> $\abs{\ve w}$ gives the degree of a word $\ve e$; $\abs{a}$ is the absolute value of the number $a$\\ 
    $\wedge$ \> $a \wedge b = \sign{a}\sign{b}(\abs{a}+\abs{b});$ see page \pageref{abwedge}\\
    $\mu(n)$ \>the Möbius function; see \pageref{abmue}\\
    $\zeta_n$\> $\zeta_n$ is the value of the zeta-function at $n$\\
    $\textnormal{part}(\ve a)$ \> see page \pageref{abdefpart}\\
    $\textnormal{rev}(\ve a)$ \> see page \pageref{abdefpart}\\
    $1_n$ \> a vector of length $n$ where all components are 1\\
    $0_n$ \> a vector of length $n$ where all components are 0\\
    $\M{f(x)}{n}$ \> the Mellin transform of f(x); see page \pageref{abmell} \\
    $\Mp{f(x)}{n}$ \> the extended Mellin transform of f(x); see page \pageref{abmellplus} \\
  \end{tabbing}

\newpage
\pagenumbering{arabic}
\input{introduction}
\input{algebraic}
\input{harmonicpolylogs}
\input{halfint}

\input{summation}
\input{example}

\appendix

\bibliographystyle{plain}

\end{document}

%% file: introduction.tex
\chapter{Introduction}
\label{Introduction}
Multiple harmonic sums, see e.g., \cite{Bluemlein2004,Vermaseren1998} and Nielsen-type integrals, associated to them by a Mellin transform, emerge in perturbative calculations of massless or massive single scale problems in quantum field theory; for more details how these sums emerge we refer, \ie to \cite{Bluemlein2009}. Due to the complexity of higher order calculations the knowledge of as many as possible relations between the finite harmonic sums is of importance to simplify the calculations and to obtain as compact as possible analytic results. The multiple finite harmonic sums $S_{a_1,a_2,\ldots,a_k}(n)$ defined as
\begin{equation}
	S_{a_1,\ldots ,a_k}(n)= \sum_{n\geq i_1 \geq i_2 \geq \cdots \geq i_k \geq 1} \frac{\sign{a_1}^{i_1}}{i_1^{\abs {a_1}}}\cdots
	\frac{\sign{a_k}^{i_k}}{i_k^{\abs {a_k}}}
\label{erstedef}	
\end{equation}
yield a very convenient description for these field theoretic quantities \cite{Bluemlein2008}.\\ 
Harmonic sums form a quasi-shuffle algebra \cite{Hoffman1992,Hoffman1997,Hoffman} and the algebraic relations of the harmonic sums are well known \cite{Bluemlein2004}. The number of the respective basis elements are counted by the Lyndon words. All further relations between the nested harmonic sums are called structural relations. They result from the mathematical structure of these objects beyond that given by their indices \cite{Bluemlein2008,Bluemlein2009,Bluemlein2009a}.\\
In physical applications \cite{Bluemlein1999,Bluemlein2000,Bluemlein2004,Bluemlein2005,Bluemlein2008} an analytic continuation of the nested harmonic sums is required and eventually one has to derive the complex analysis for nested harmonic sums. In passing, $n$ takes values $n\in \Q$ and $n\in \R,$ which leads to new relations \cite{Bluemlein2008}. Considering $n\in\Q$ leads, e.g., to half-integer relations, while considering $n\in\R$ will allow us, e.g., to differentiate with respect to $n$ and so new structural relations will arise.\\
Multiple harmonic sums can be represented in terms of Mellin integrals \cite{Vermaseren1998,Bluemlein1999}, in fact they are the Mellin transforms of harmonic polylogarithms \cite{Remiddi2000}, which belong to the Poincar\'{e} iterated integrals see \cite{Poincare1884,Lappo-Danielevsky1953}. This fact allows us to establish a differentiation on the harmonic sums.\\
The package \ttfamily HarmonicSums \rmfamily tries to combine all these features and generalizes some aspects to Euler-Zagier sums and S-sums (see \cite{Moch2002}). Besides basic operations such as multiplication and evaluation of harmonic sums, harmonic polylogarithms, Euler-Zagier sums and S-sums it provides procedures to compute the Mellin and the inverse Mellin transform, to differentiate harmonic sums and to find algebraic and structural relations. Tables of these relations are provided and procedures to apply these relations are included in the package \ttfamily HarmonicSums.\rmfamily\\
With the package \ttfamily Sigma \rmfamily \cite{Schneider2007} we can simplify nested sums such that the nested depth is optimal and the degree of the denominator is minimal. This is possible due to a refined summation theory \cite{Schneider2007a,Schneider2008a} of $\Pi \Sigma$-fields \cite{Karr1981}. The package \ttfamily HarmonicSums \rmfamily finds sum representations of such simplified nested sums in terms of harmonic sums as much as it is possible.\\
The remainder of the thesis is structured as follows: In Chapter \ref{Algebraic Relations between Multiple Harmonic Sums} harmonic sums and similar structures are introduced. We show that harmonic sums form a quasi-shuffle algebra and we show the connection to Lyndon words. Finally, we outline a way how we can find algebraic relations between harmonic sums and how we can apply these relations.
While Chapter \ref{Algebraic Relations between Multiple Harmonic Sums} only deals with algebraic relations we introduce harmonic polylogarithms in Chapter \ref{Harmonic Polylogarithms}, which leads to differentiation of harmonic sums via Mellin-transform. Along with that new relations, structural relations will pop up.
Chapter \ref{Half-Integer Relations} deals with half-integer relations between harmonic sums and in Chapter \ref{Summation of Multiple Harmonic Sums} we show how a certain kind of nested sums can be rewritten in terms of harmonic sums. In the last chapter we give a nontrivial example to illustrate how these algorithms in cooperation with the summation package \ttfamily Sigma \rmfamily support the evaluation of Feynman integrals.\\

\text{\bf Acknowledgments.} This diploma thesis was supported by the SFB grant F1305 and the grant P20347-N18 of the Austrian FWF.\\
I would like to thank C. Schneider for supervising this work and his support. In addition I would like to thank J. Blümlein for his helpful suggestions and comments. In particular it was an exciting experience to be invited at DESY, Zeuthen.

%% file: algebraic.tex
\chapter{Algebraic Relations between Multiple Harmonic Sums}
\label{Algebraic Relations between Multiple Harmonic Sums}
In this section we will define harmonic sums and relate to them the Euler-Zagier sums. In the following we will rely on the fact that harmonic sums form a quasi-shuffle algebra. Quasi-shuffle algebras are connected to the free polynomial algebra on the \textit{Lyndon} words. It will turn out that the number of \textit{Lyndon} words allows one to determine the number of algebraically independent harmonic sums (in the context of quasi-shuffle algebras) and  gives at least an upper bound for the number of algebraically independent harmonic sums (considered as sequences). In the end of this chapter we will present a way how we can derive these algebraically independent harmonic sums following the ideas from \cite{Bluemlein2004}.  
\section{Definition and Product of Harmonic Sums}

We start by defining multiple harmonic sums; see, e.g.,\cite{Bluemlein2004,Vermaseren1998}.
\begin{definition} For $k \in \N$, $n \in \Z$ and $a_i\in \Z/\{0\}$ with $1\leq i\leq k,$
\begin{eqnarray}
	S(n)&=&\left\{ 
		  	\begin{array}{ll}
						1,\  n > 0 & \\
						0,\  n\leq 0, & 
					\end{array} \right. \nonumber\\
	S_{a_1,\ldots ,a_k}(n)&=& \sum_{i=1}^n \frac{\sign{a_1}^i}{i^{\abs{a_1}}}S_{a_2,\ldots ,a_k}(i); \nonumber
	\label{defHsum}
\end{eqnarray}
\end{definition}
$k$ is called the depth and $w=\abs{a_1}+\cdots+\abs{a_k}$ is called the weight of the harmonic sum $S_{a_1,\ldots ,a_k}(n)$. Note that an equivalent definition is given in (\ref{erstedef}).

\begin{notation}
Sometimes we will write $S_{a_1a_2a_3\cdots a_k}(n)$ instead of $S_{a_1,a_2,a_3,\ldots ,a_k}(n)$.
\end{notation}

As we will see later in Section \ref{shufflequasialg} harmonic sums obey an algebra. In this algebra the following property is crucial: a product of two harmonic sums with the same upper summation limit can be written in terms of single harmonic sums: for $n\in \N,$
\begin{eqnarray}
	S_{a_1,\ldots ,a_k}(n)*S_{b_1,\ldots ,b_l}(n)&=&
	\sum_{i=1}^n \frac{\sign{a_1}^i}{i^{\abs {a_1}}}S_{a_2,\ldots ,a_k}(i)*S_{b_1,\ldots ,b_l}(i) \nonumber\\
	&&+\sum_{i=1}^n \frac{\sign{b_1}^i}{i^{\abs {b_1}}}S_{a_1,\ldots ,a_k}(i)*S_{b_2,\ldots ,b_l}(i) \nonumber\\
	&&-\sum_{i=1}^n \frac{\sign{a_1*b_1}^i}{i^{\abs{a_1}+\abs{b_1}}}S_{a_2,\ldots ,a_k}(i)*S_{b_2,\ldots ,b_l}(i);
\label{hsumproduct}
\end{eqnarray}
here $*$ means the usual multiplication.
The proof of equation (\ref{hsumproduct}) follows immediately from 
$$
\sum_{i=1}^n \sum_{j=1}^n a_{ij}=\sum_{i=1}^n \sum_{j=1}^i a_{ij}+\sum_{j=1}^n \sum_{i=1}^j a_{ij}-\sum_{i=1}^n a_{ii};
$$
for a graphical illustration of the proof see Figure \ref{fig1}. As a consequence, any product of harmonic sums can be written as a linear combination of harmonic sums by iterative application of (\ref{hsumproduct}).
\begin{figure}
\centering
\parbox{.15\textwidth}{\includegraphics[width=2cm]{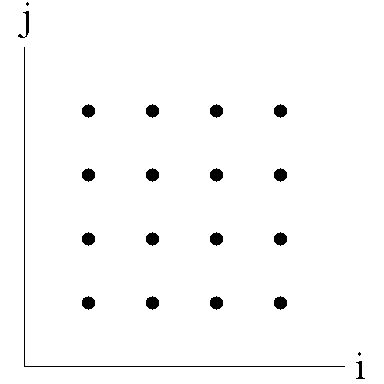}}
\centering
\parbox{.05\textwidth}{$=$}
\centering
\parbox{.15\textwidth}{\includegraphics[width=2cm]{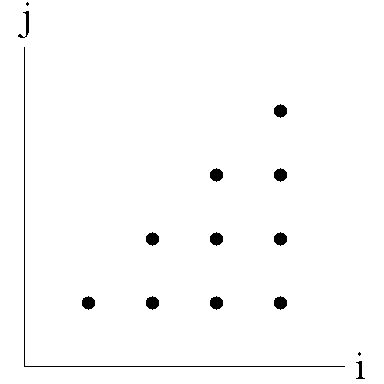}}
\centering
\parbox{.05\textwidth}{$+$}
\centering
\parbox{.15\textwidth}{\includegraphics[width=2cm]{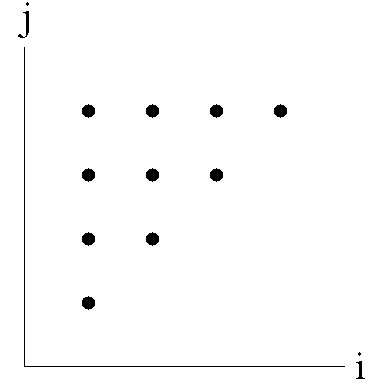}}
\centering
\parbox{.05\textwidth}{$-$}
\centering
\parbox{.15\textwidth}{\includegraphics[width=2cm]{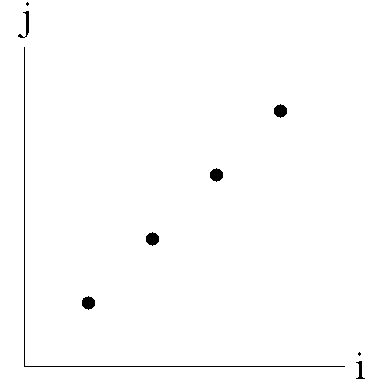}}
\caption{Sketch of the proof of the multiplication of harmonic sums (compare \cite{Moch2002}).}
\label{fig1}
\end{figure}
We give an example for the product of two harmonic sums.
\begin{example} For $n \geq 0,$
\begin{eqnarray}
S_{a_1}(n)*S_{b_1,\ldots ,b_l}(n)&=&S_{a_1,b_1,\ldots ,b_l}(n)+S_{b_1,a_1,b_2,\ldots ,b_l}(n)+\cdots+S_{b_1,b_2,\ldots ,b_l,a_1}(n)\nonumber\\
&&-S_{a_1\wedge b_1,b_2,\ldots,b_l}(n)-\cdots-S_{b_1,b_2,\ldots,a_1\wedge b_m}(n).\\ \nonumber
\label{pruductex}
\end{eqnarray}
Here the symbol $\wedge$ is defined as
\begin{equation}
a \wedge b = \sign{a}\sign{b}(\abs{a}+\abs{b}).\nonumber\\
\label{abwedge}
\end{equation}
\end{example}
Using the package \ttfamily HarmonicSums \rmfamily in \sffamily Mathematica \rmfamily we can do this automatically:
\begin{fmma}
{
\In \text{\bf \text{S}[1,4,n]\,\text{S}[2,-3,n]//\text{LinearExpand}}\\
\Out {\text{S}[3,-7,n]-\text{S}[1,2,-7,n]-\text{S}[1,6,-3,n]-\text{S}[2,-4,4,n]-\text{S}[2,1,-7,n]- \text{S}[3,-3,4,n]-\text{S}[3,4,-3,n]+\text{S}[1,2,-3,4,n]+\text{S}[1,2,4,-3,n]+\text{S}[1,4,2,-3,n]+\text{S}[2,-3,1,4,n]+\text{S}[2,1,-3,4,n]+\text{S}[2,1,4,-3,n]}\\
}
\end{fmma}
 
\begin{notation}
As it turns out later, the following modified notation (also used in \cite{Vermaseren1998}) for multiple harmonic sums will be useful. In this notation we will only allow the indices $1,\ 0,$ and $-1$. An index zero which is followed by an index that is nonzero means that 1 should be added to the absolute value of the nonzero index. So for example:
$$
S_{0,0,1,0,-1,0,0,0,1}(n)=S_{3,-2,4}(n)
$$
We will refer to this modified notation as the expanded version.
\label{not1}
\end{notation}

Using the modified notation it is quite easy to calculate the number $N(w)$ of multiple harmonic sums for a given weight $w$. Since the length of the index set in the expanded version is exactly the weight of the harmonic sum, we have to consider all possible index sets of length $w.$ Each index except the last can take the three values $1,\ 0,$ and $-1.$ The last index can just be 1 or -1. Hence we get the following formula, compare \cite{Vermaseren1998} or \cite{Bluemlein2004}:
$$
N(w)=2\cdot3^{w-1}.
$$

\section{Euler-Zagier Sums}
Now we define in a similar way the so called Euler-Zagier sums, see \cite{Zagier1994,Euler1775} or compare \cite{Moch2002}:
\begin{definition}For $k \in \N$, $n \in \Z$ and $a_i\in \Z/\{0\}$ with $1\leq i\leq k,$
\begin{eqnarray}
	Z(n)&=&\left\{ 
		  	\begin{array}{ll}
						1,\  n \geq 0 & \\
						0,\  n < 0, & 
					\end{array} \right. \nonumber\\
	Z_{a_1,\ldots ,a_k}(n)&=& \sum_{i=1}^n \frac{\sign{a_1}^i}{i^{\abs {a_1}}}Z_{a_2,\ldots ,a_k}(i-1). \nonumber
\end{eqnarray}
\label{defEZsum}
\end{definition}
Again $k$ is called the depth and $w=\abs{a_1}+\cdots+\abs{a_k}$ is called the weight of the Euler-Zagier sum $Z_{a_1,\ldots ,a_k}(n)$. An equivalent definition is given by 
\begin{equation}
	Z_{a_1,\ldots ,a_k}(n)= \sum_{n \geq  i_1 > i_2 > \cdots > i_k \geq 1} \frac{\sign{a_1}^{i_1}}{i_1^{\abs {a_1}}}\cdots
	\frac{\sign{a_k}^{i_k}}{i_k^{\abs {a_k}}}.
\end{equation}
The product for these sums can be expanded with the formula
\begin{eqnarray}
	Z_{a_1,\ldots ,a_k}(n)*Z_{b_1,\ldots ,b_l}(n)&=&
	\sum_{i=1}^n \frac{\sign{a_1}^i}{i^{\abs {a_1}}}Z_{a_2,\ldots ,a_k}(i)*Z_{b_1,\ldots ,b_l}(i) \nonumber\\
	&&+\sum_{i=1}^n \frac{\sign{b_1}^i}{i^{\abs {b_1}}}Z_{a_1,\ldots ,a_k}(i)*Z_{b_2,\ldots ,b_l}(i) \nonumber\\
	&&+\sum_{i=1}^n \frac{\sign{a_1*b_1}^i}{i^{\abs{a_1}+\abs{b_1}}}Z_{a_2,\ldots ,a_k}(i)*Z_{b_2,\ldots ,b_l}(i).
\label{zpro}
\end{eqnarray}
 
Again we can perform this expansion using the package \ttfamily HarmonicSums\rmfamily:
\begin{fmma}
\In \text{\bf \text{Z}[1,4,n]\,\text{Z}[2,-3,n]}\\
\Out {\text{Z}[3,-7,n]+\text{Z}[1,2,-7,n]+\text{Z}[1,6,-3,n]+\text{Z}[2,-4,4,n]+\text{Z}[2,1,-7,n]+ \text{Z}[3,-3,4,n]+\text{Z}[3,4,-3,n]+\text{Z}[1,2,-3,4,n]+\text{Z}[1,2,4,-3,n]+\text{Z}[1,4,2,-3,n]+\text{Z}[2,-3,1,4,n]+\text{Z}[2,1,-3,4,n]+\text{Z}[2,1,4,-3,n]}\\
\end{fmma} 

Multiple harmonic sums and Euler-Zagier sums are closely related. In particular, we can easily convert them from one representation to the other one, compare \cite{Moch2002}. Namely by using 
\begin{eqnarray}
	S_{a_1,\ldots}(n)&=&
	\sum_{i_1=1}^n \frac{\sign{a_1}^{i_1}}{i_1^{\abs {a_1}}}\sum_{i_2=1}^{i_1-1} 																																\frac{\sign{a_2}^{i_2}}{i_2^{\abs{a_2}}}S_{a_3,\ldots}(i_2)\nonumber\\
	&&+S_{a_1\wedge a_2,a_3,\ldots}(n),\nonumber
\end{eqnarray}
we can recursively convert multiple harmonic sums to Euler-Zagier sums. Similarly one can exploit
\begin{eqnarray}
	Z_{a_1,\ldots}(n)&=&
	\sum_{i_1=1}^n \frac{\sign{a_1}^{i_1}}{i_1^{\abs {a_1}}}\sum_{i_2=1}^{i_1} \frac{\sign{a_2}^{i_2}}{i_2^{\abs{a_2}}}Z_{a_3,\ldots, 					}(i_2-1)\nonumber\\
	&&-Z_{a_1 \wedge a_2 ,a_3,\ldots}(n),\nonumber
\end{eqnarray}
in order to convert Euler-Zagier sums to multiple harmonic sums.
\begin{example}For $n \in \N,$
\begin{eqnarray*}
S_{1,3,4}(n)&=&Z_{8}(n)+Z_{1,7}(n)+Z_{4,4}(n)+Z_{1,3,4}(n),\\
Z_{1,3,4}(n)&=&S_{8}(n)-S_{1,7}(n)-S_{4,4}(n)+S_{1,3,4}(n).
\end{eqnarray*}
\end{example}
The package \ttfamily HarmonicSums \rmfamily provides procedures for these conversions:
\begin{fmma}
\begin{mma}
\In \text{\bf ZToS[Z[1, 3, 4, n]]}\\
\Out {\text{S}[8,n]-\text{S}[1,7,n]-\text{S}[4,4,n]+\text{S}[1,3,4,n]}\\
\end{mma}
\begin{mma}
\In \text{\bf SToZ[S[1, 3, 4, n]]}\\
\Out {\text{Z}[8,n]+\text{Z}[1,7,n]+\text{Z}[4,4,n]+\text{Z}[1,3,4,n]}\\
\end{mma}
\end{fmma} 

\section{The Shuffle and the Quasi-shuffle Algebra}
\label{shufflequasialg}
In Section \ref{harasshuf} we present the algebra of multiple harmonic sums, but first we have to give some basic definitions; in particular, we will introduce the shuffle and the quasi-shuffle algebra.
\begin{definition}[Graded algebra]
A graded algebra over a ring\footnote{Throughout this thesis we assume that all rings and fields contain $\Q$. The ground field, denoted by $K$ and the ground ring, denoted by $R$, are always commutative.} $R$ is an $R$-algebra $E$ that has a direct sum decomposition
$$
E=\bigoplus_i E_i = E_0 \oplus E_1 \oplus E_2 \oplus \cdots
$$
such that for any $x\in E_i$, and $y\in E_j$ we have $xy \in E_{i+j},$ $\ie$ $E_iE_j \subseteq  E_{i+j}$ and such that for any $r\in R$ and $x \in E_i$ we have $rx\in E_i,$ $\ie$ $RE_i \subseteq  E_{i}.$
\end{definition}
\begin{definition}[Non-commutative Polynomial Algebra]
\label{noncomalg}
Let $A$ be a totally ordered, graded set. The degree of $a\in A$ is denoted by $\abs{a}.$
Let $A^*$ denote the free monoid over $A$, i.e.,
		$$A^*=\left\{a_1 \cdots a_n | a_i \in A, n\geq 1\right\}\cup\left\{\epsilon \right\}.$$
We extend the degree function to $A^*$ by $\abs{a_1 a_2 \cdots a_n}=\abs{a_1}+\abs{a_2}+\cdots+\abs{a_n}$ for $a_i \in A$ and $\abs{\epsilon}=0.$
Let $R\supseteq \mathbb{Q}$ be a commutative ring. The set of non-commutative polynomials over $R$ is defined as
		$$
			R\left\langle A\right\rangle:=\left\{\left.\sum_{\ve w\in A^*} r_\ve w \ve w\right|r_\ve w \in R, r_\ve w=0 \textnormal{ for almost all } \ve w \right\}.
		$$
Addition in $R\left\langle A\right\rangle$ is defined component wise and multiplication is defined by
		$$
			\sum_\ve w a_\ve w \ve w \cdot \sum_\ve w b_\ve w \ve w :=  \sum_\ve w (\sum_{uv=\ve w} a_u b_v) \ve w. 
		$$
\end{definition}
\begin{remark}
We refer to elements of $A$ as \textit{letters} and to elements of $A^*$ as \textit{words}. We call $\epsilon$ the \textit{empty} word. The length of a word $\ve w=a_1a_2\ldots a_k$ is denoted by $\ell(\ve w)=k.$ The set of non-empty words is denoted by $A^+.$
\label{abnonemp}
\end{remark}
\begin{definition}[Shuffle algebra]
On $A^*$, the free monoid over $A,$ we define the shuffle product as follows. If 
\begin{equation}
\ve u=a_1 \underbrace{a_2 \ldots a_k}_{\ve{u}_1} \textnormal{ and } \ve v=b_1 \underbrace{b_2 \ldots b_l}_{\ve{v}_1} \nonumber
\end{equation}
then
\begin{eqnarray}
\ve u\shuffle \epsilon&=&u, \nonumber \\ 
\ve u\shuffle \ve v&=&a_1\cdot(\ve{u}_1\shuffle \ve v)+b_1\cdot(\ve u\shuffle \ve v_1).\nonumber \\
\nonumber
\label{abshuffle1}
\end{eqnarray}
This product is extended to $R\left\langle A\right\rangle$ by linearity.
\label{shuffdef}
\end{definition}

\begin{remark}
It is an easy exercise to verify that $(R\left\langle A\right\rangle, \shuffle)$ is an associative, commutative and graded $K$-algebra, which is in the following called the \textit{shuffle algebra}.
\end{remark}
\subsection{Lydon words}
Suppose now that the set $A$ of letters is totally ordered by $<$.
We extend this order to words with the following lexicographic order:
$$
\left\{ 
		  	\begin{array}{ll}
						\ve u <\ve uv \ \textnormal{for}  \ v \in A^+\ \textnormal{and } \ve u \in A^*,& \\
						\ve w_1a_1\ve w_2<\ve w_1a_2\ve w_3, \ \textnormal{if} \, a_1<a_2 \ \textnormal{for}\,  a_1,a_2 \in A \textnormal{ and } w_1,w_2,w_3 \in A^*. & 
					\end{array} \right.
$$
\begin{definition}
For a word $\ve w=pxs$ with $p,x,s \in A^*,$ $p$ is called a \textit{prefix} of $w$ if $p\neq \epsilon$ and any $s$ is called a \textit{suffix} of $\ve w$ if $s\neq \epsilon$.
A word $\ve w$ is called a \textit{Lyndon} word if it is smaller than any of its suffixes. The set of all \textit{Lyndon} words on $A$ is denoted by Lyndon($A$).
\label{ablyn}
\end{definition}
Note that any word $\ve w \in A^*$ can be factorized uniquely in a lexicographically decreasing product of \textit{Lyndon} words \cite{Bigotte2002,Lothaire1983,Chen1958}. Therefore it can be written as $\ve w=\ve l_1^{i_1}\cdots \ve l_k^{i_k}$, where $\ve l_j \in$ Lyndon($A$) and $\ve l_i>\ve l_j$ if $i<j$.
The following theorem was first obtained by Radford \cite{Radford1979} a proof can be found in \cite{Reutenauer1969}:
\begin{thm}
The shuffle algebra $(R\left\langle A\right\rangle, \shuffle)$ is the free polynomial algebra on the \textit{Lyndon} words, i.e., $(R\left\langle A\right\rangle, \shuffle) \simeq (R[\textnormal{Lyndon}(A)], \shuffle).$
\label{shuffalgpolyalg}
\end{thm}
In other words, any non-commutative polynomial in $R\left\langle A\right\rangle$, can be expressed uniquely as a commutative polynomial on the \textit{Lyndon} words.
\begin{notation}
For $w\in A^*$ and $i\in N$ we use the abbreviation
$$
w^{\shuffle i}=\underbrace{w \shuffle w \shuffle \ldots \shuffle w}_{i\ \textnormal{times}}.
$$
\end{notation}
We can use the following lemma to rewrite each word as a commutative polynomial of \textit{Lyndon} words, see \cite{Radford1979} or \cite{Bigotte2002}.
\begin{lemma}Let $\ve w=\ve l_1^{i_1}\cdots \ve l_k^{i_k}$ be a word, where $\ve l_1,\ldots,\ve l_k \in$ Lyndon($A$) with $\ve l_i>\ve l_j$ if $i<j$ and $i_1,\ldots,i_k \in \N$. The polynomial $Q_\ve w=1/(i_1!\cdots i_k!)\ve l_1^{\shuffle i_1}\shuffle \cdots \shuffle \ve l_k^{\shuffle i_k}$ can be written as $\ve w+R_\ve w$, where $R_\ve w$ contains only words whose weights are equal to the weight of $\ve w$ and which are smaller than $\ve w$ (with respect to our lexicographic order).
\end{lemma}
\begin{example}
Let $\ve w=babab$, where $a<b$. Obviously, $\ve w$ is not a \textit{Lyndon} word. Using the lemma above recursively, we get
$$
\ve w = \frac{1}{2} b \shuffle ab^{\shuffle 2}+12 aabbb+4 ababb-2 abb \shuffle ab-2 b\shuffle aabb.
$$
\end{example}
\subsection{The Quasi-Shuffle Algebra}
We define a new multiplication $*$ on $R\left\langle A\right\rangle$ which is a generalisation of the \textit{shuffle product}, by requiring that $*$ distributes with the addition. We will see that this product can be used to describe properties of harmonic sums. In \cite{Hoffman} (see also \cite{Hoffman1992,Hoffman1997}) there is a similar construction which for example gives rise to properties of Euler-Zagier sums.  
\begin{definition}[Quasi-shuffle product]
\label{quasishuffprodef}
$*: R\left\langle A\right\rangle \times R\left\langle A\right\rangle \longrightarrow R\left\langle A\right\rangle$ is called Quasi-shuffle product, if it distributes with the addition and 
\begin{eqnarray}
\epsilon *\ve w &=& \ve w* \epsilon = \ve w, \ \textnormal{for all}\ w \in A^*,\nonumber\\
a\ve u*b\ve v &=& a(\ve u*\ve v)+b(\ve u*\ve v)-[a,b](\ve u*\ve v), \ \textnormal{for all}\ a,b \in A;\ve u,\ve v \in A^*,\label{eq:quasipro}
\end{eqnarray}
where $[\cdot,\cdot]:\overline{A}\times \overline{A}\rightarrow \overline{A}$, $(\overline{A} = A \cup \left\{0\right\})$ is a function satisfying
\begin{eqnarray}
&\textnormal{S}0.& [a,0]=0 \ \textnormal{for all } a \in \overline{A}; \nonumber\\
&\textnormal{S}1.& [a,b]=[b,a] \ \textnormal{for all } a,b \in \overline{A}; \nonumber\\
&\textnormal{S}2.& [[a,b],c]=[a,[b,c]] \ \textnormal{for all } a, b, c \in \overline{A}; \textnormal{ and} \nonumber\\
&\textnormal{S}3.& \textnormal{Either } [a,b]=0 \textnormal{ or } \abs{[a,b]}=\abs a +\abs b \ \textnormal{for all } a, b\in \overline{A}.\nonumber
\end{eqnarray}
\end{definition}

Note that the quasi-shuffle product gives rise to various applications by defining explicit functions $[\cdot,\cdot]$ accordingly. In Section \ref{harasshuf} we exploit this construction to handle harmonic sums properties.

\begin{thm}(compare \cite{Hoffman})
$(R\left\langle A\right\rangle,*)$ is a commutative graded algebra. 
\end{thm}
\begin{proof}
It is enough to show that the operation $*$ is commutative, associative and adds degrees. For commutativity we have to show $\ve w_1*\ve w_2=\ve w_2*\ve w_1$ for any words $\ve w_1$ and $\ve w_2.$ We proceed by induction on $\ell(\ve w_1)+\ell(\ve w_2)$. Since there is nothing to prove if either $\ve w_1$ or $\ve w_2$ is empty, we can assume that there are letters $a,b$ so that $\ve w_1=a\ve u$ and $\ve w_2=b\ve v$. Then (\ref{eq:quasipro}) together with the induction hypothesis gives
\begin{equation}
\ve w_1*\ve w_2-\ve w_2*\ve w_1=-[a,b](\ve u*\ve v)+[b,a](\ve v*\ve u),\nonumber
\end{equation}
and the right-hand side is zero by the induction hypothesis and (S1). For associativity we have to prove $(\ve w_1*\ve w_2)*\ve w_3=\ve w_1*(\ve w_2*\ve w_3)$ for any words $\ve w_1, \ve w_2$ and $\ve w_3.$ We proceed by induction on $\ell(\ve w_1)+\ell(\ve w_2)+\ell(\ve w_3).$ Since there is nothing to prove if $\ell(\ve w_1)+\ell(\ve w_2)+\ell(\ve w_3)<3$ or either $\ve w_1$, $\ve w_2$ or $\ve w_3$ is empty, we can assume that there are letters $a,b,c$ so that $\ve w_1=a\ve v_1$, $\ve w_2=b\ve v_2$ and $\ve w_3=b\ve v_3$. Then the following holds by using (S2), distributivity, the induction hypothesis and (\ref{eq:quasipro}):
\footnotesize
\begin{eqnarray}
			(\ve w_1*\ve w_2)*\ve w_3&=&(a(\ve v_1*b\ve v_2)+b(a\ve v_1*\ve v_2)-[a,b](\ve v_1*\ve v_2))*c\ve v_3\nonumber \\  
						 &=&a((\ve v_1*b\ve v_2)*c\ve v_3)+c(a(\ve v_1*b\ve v_2)*\ve v_3)-[a,c]((\ve v_1*b\ve v_2)*\ve v_3)\nonumber\\
						 &&+b((a\ve v_1*\ve v_2)*c\ve v_3)+c(b(a\ve v_1*\ve v_2)*\ve v_3)-[b,c]((a\ve v_1*\ve v_2)*\ve v_3)\nonumber\\
						 &&-[a,b]((\ve v_1*\ve v_2)*c\ve v_3)-c([a,b](\ve v_1*\ve v_2)*\ve v_3)+[[a,b],c]((\ve v_1*\ve v_2)*\ve v_3)\nonumber\\
						 &=&a(\ve v_1*(b\ve v_2*c\ve v_3))+c((a\ve v_1*b\ve v_2)*\ve v_3)-[a,c](\ve v_1*(b\ve v_2*\ve v_3))\nonumber\\
						 &&+b(a\ve v_1*(\ve v_2*c\ve v_3))-[b,c](a\ve v_1*(\ve v_2*\ve v_3)) -[a,b](\ve v_1*(\ve v_2*c\ve v_3))\nonumber\\
						 &&+[a,[b,c]](\ve v_1*(\ve v_2*\ve v_3))\nonumber\\ 
						 &=&a(\ve v_1*(b(\ve v_2*c\ve v_3)))+a(\ve v_1*c(b\ve v_2*\ve v_3))-a(\ve v_1*[b,c](\ve v_2*\ve v_3))\nonumber\\ 
						 &&+c(a\ve v_1*(b\ve v_2*\ve v_3))-[a,c](\ve v_1*(b\ve v_2*\ve v_3))+b(a\ve v_1*(\ve v_2*c\ve v_3))\nonumber\\ 
						 &&-[b,c](a\ve v_1*(\ve v_2*\ve v_3)) -[a,b](\ve v_1*(\ve v_2*c\ve v_3))+[a,[b,c]](\ve v_1*(\ve v_2*\ve v_3))\nonumber\\
						 &=&a\ve v_1*(b(\ve v_2*c\ve v_3)+c(b\ve v_2*\ve v_3)-[b,c](\ve v_2*\ve v_3))=\ve w_1*(\ve w_2*\ve w_3).	\nonumber				 
	\end{eqnarray}
\normalsize
To show that $*$ adds degrees it suffices to show $\abs{\ve w_1*\ve w_2}=\abs{\ve w_1}+\abs{\ve w_2}.$ We apply again induction on $\ell(\ve w_1)+\ell(\ve w_2).$ Since there is nothing to prove if either $\ve w_1$ or $\ve w_2$ is empty, we can assume there are letters $a,b$ so that $\ve w_1=a\ve u$ and $\ve w_2=b\ve v$. Recall that $\abs{a\ve u}=\abs{a}+\abs{\ve u},$ see Definition \ref{noncomalg}. By (\ref{eq:quasipro}) we get 
\begin{equation}
\abs{\ve w_1*\ve w_2}=\abs{a(\ve u*b\ve v)+b(a\ve u*\ve v)-[a,b](\ve u*\ve v)}.\nonumber
\end{equation}
By (S3) and the induction hypothesis the right-hand side is smaller or equal to $\abs{\ve w_1}+\abs{\ve w_2}$.
\end{proof}
Similar to \cite{Hoffman} we proceed by defining an isomorphism $\varphi(\ve w): R\left\langle A\right\rangle \rightarrow R\left\langle A\right\rangle$, see Theorem \ref{shuffalgquasishuffalg}. In order to accomplish this task, we will use the following notation which extends the operation $[\cdot,\cdot]$ as follows. Define inductively $[S] \in \overline{A}$ for any finite sequence $S$ of elements of $A$ by setting $[a]=a$ for $a \in A$ and setting $[a,S]=[a,[S]]$ for any $a \in A$ and sequence $S$ of elements of $A$. We get the following proposition \cite{Hoffman}:
\begin{prop}
\begin{itemize}
	\item[\textit{(i)}] If $[S]=0$, $[S']=0$ whenever $S$ is a subsequence of $S'$.
	\item[\textit{(ii)}] $[S]$ does not depend on the order of the elements of $S$.
	\item[\textit{(iii)}] For any sequences $S_1$ and $S_2$, $[S_1 \sqcup S_2]=[[S_1],[S_2]]$ where $S_1 \sqcup S_2$ denotes the concatenation of sequences $S_1$ and $S_2$.
	\item[\textit{(iv)}] If $[S]\neq 0$, then the degree of $S$ is the sum of the degrees of the elements of $S$.
\end{itemize}
\label{prop1}
\end{prop}
\begin{proof}
\textit{(i)}, \textit{(ii)}, \textit{(iii)}, \textit{(iv)} follow from (S0), (S1), (S2), (S3) respectively.
\end{proof}
A composition of a positive integer $n$ is a sequence $I=(i_1,i_2,\ldots,i_k)$ of positive integers such that $i_1+i_2+\ldots+i_k=n$. We call $k=\ell(I)$ the length of $I$ and write $C(n)$ for the set of compositions of $n$. As in \cite{Hoffman}, compositions act on words via $[\cdot,\cdot]$ as follows. For any word $w=a_1a_2\cdots a_n$ and composition $I=(i_1,i_2,\ldots,i_k)\in C(n)$, set
$$
I[w]=[a_1,\ldots,a_{i_1}][a_{i_1+1},\ldots,a_{i_1+i_2}]\cdots[a_{i_1+\cdots+i_{k-1}+1},\ldots,a_n].
$$
Now let $\varphi(w): R\left\langle A\right\rangle \rightarrow R\left\langle A\right\rangle$ be the linear map with $\varphi(\epsilon)=\epsilon$ and
\begin{equation}
\varphi(\ve w)=\sum_{(i_1,\ldots,i_k)\in C(\ell(\ve w))}{\frac{(-1)^{\ell(\ve w)-k}}{i_1! \ldots i_k!}}(i_1 \ldots i_k)[\ve w]
\label{phi1}
\end{equation}
for any nonempty word $\ve w$. There is an inverse $\psi$ of $\varphi$ given by
\begin{equation}
\psi(\ve w)=\sum_{(i_1,\ldots,i_k)\in C(\ell(\ve w))}{\frac{1}{i_1 \ldots i_k}}(i_1 \ldots i_k)[\ve w]
\label{psi1}
\end{equation}
for any word $\ve w$, and extended to $R\left\langle A\right\rangle$ by linearity; this follows by taking $f(t)=\frac{e^t-1}{e^t}=1-\frac{1}{e^t}=1-\sum_{n=0}^\infty{\frac{(-t)^n}{n!}}=\sum_{n=1}^\infty{\frac{-(-t)^n}{n!}}$ in the following Lemma. ($f^{-1}(t)=\log(\frac{1}{1-t})=\sum_{n=1}^\infty{\frac{t^n}{n}}$)
\begin{lemma}(see \cite{Hoffman})
Let $f(t)=a_1t+a_2t^2+a_3t^3\ldots$ be a function analytic at the origin, with $a_1\neq0$ and $a_i \in \R$ for all $i$, and let $f^{-1}(t)=b_1t+b_2t^2+b_3t^3\ldots$ be the inverse of $f$. Then the map $\Psi_f: R\left\langle A\right\rangle\rightarrow R\left\langle A\right\rangle$ given by
\begin{equation}
\Psi_f(\ve w)=\sum_{I=(i_1,\ldots,i_k)\in C(\ell(\ve w))}{a_{i_1}a_{i_2} \cdots a_{i_k}}I[\ve w]
\end{equation}
for words $\ve w$, and extended by linearity, has the inverse $\Psi_f^{-1}=\Psi_{f^{-1}}$ given by
\begin{equation}
\Psi_{f^{-1}}(\ve w)=\sum_{I=(i_1,\ldots,i_k)\in C(\ell(\ve w))}{b_{i_1}b_{i_2} \cdots b_{i_k}}I[\ve w].
\end{equation}
\label{isolem}
\end{lemma}
\begin{thm}
$\varphi$ is an isomorphism of $(R\left\langle A\right\rangle, \shuffle)$ onto $(R\left\langle A\right\rangle,*)$ (as graded k-algebras). 
\label{shuffalgquasishuffalg}
\end{thm}
\begin{proof} We follow the proof of \cite[Thm. 2.5]{Hoffman}. From the Lemma \ref{isolem} $\varphi$ is invertible. From Proposition \ref{prop1}\textit{(iv)} and (\ref{phi1}) it follows that $\abs{\varphi(\ve w)}\leq\abs{\ve w}$. Similar, for the inverse we have $\abs{\psi(\ve w)}\leq\abs{\ve w}$ and thus $\abs{\varphi(\ve w)}=\abs{\ve w}=\abs{\psi(\ve w)}$. To show that $\varphi$ is a homomorphism it suffices to verify that $\varphi(\ve v\shuffle \ve w)=\varphi(\ve v)*\varphi(\ve w)$ for any words $\ve w,\ve v$. Let $\ve v=a_1 \cdots a_m$ and $\ve w=b_1 \cdots b_n$. Evidently both, $\varphi(\ve v\shuffle \ve w)$ and $\varphi(\ve v)*\varphi(\ve w)$ are sums of rational multiples of terms generated by
\begin{equation}
[S_1 \sqcup T_1][S_2 \sqcup T_2]\cdots[S_l \sqcup T_l]
\label{eq:term}
\end{equation}
where the $S_i$ and $T_i$ are sequences of $a_1,\ldots,a_m$ and $b_1,\ldots,b_n,$ respectively, such that
\begin{enumerate}
	\item for each $i,$ at most one of the $S_i$ and $T_i$ is empty; and
	\item the concatenation $S_1 S_2 \ldots S_l$ is the sequence $a_1,\ldots,a_m$ and $T_1 T_2 \ldots T_l$ is the sequence $b_1,\ldots,b_n$.
\end{enumerate}
Now the term generated by (\ref{eq:term}) arises in $\varphi(\ve v)*\varphi(\ve w)$ in only one way. The absolute value of the coefficient of the term given by (\ref{eq:term}) in $\varphi(\ve v)*\varphi(\ve w)$ is
$$
\frac{1}{(\textnormal{length } S_1)!(\textnormal{length } S_2)!\ldots(\textnormal{length } S_l)!(\textnormal{length } T_1)!(\textnormal{length } T_2)!\ldots(\textnormal{length } T_l)!}.
$$
To determine the sign of the coefficient, let $\ve v_1$ be the word in $\varphi(\ve v)$ and let $\ve w_1$ be the word in $\varphi(\ve w)$ such that (\ref{eq:term}) is a word in $\ve v_1*\ve w_1$. Let $x$, $y$ be the lengths of $\ve w_1$ and $\ve v_1$ respectively. The sign is given by the following ingredients:
\begin{itemize}
	\item the sign of $\ve v_1$ in $\varphi(\ve v)$ which is $(-1)^{m-x};$
	\item the sign of $\ve w_1$ in $\varphi(\ve w)$ which is $(-1)^{n-y};$
	\item the sign of (\ref{eq:term}) in $\ve v_1*\ve w_1$ which is $(-1)^{(x+y)-l}$.
\end{itemize}
Hence the coefficient of the term produced by all possibilities of (\ref{eq:term}) in $\varphi(\ve v)*\varphi(\ve w)$ is:
$$
\frac{(-1)^{m+n-l}}{(\textnormal{length } S_1)!(\textnormal{length } S_2)!\ldots(\textnormal{length } S_l)!(\textnormal{length } T_1)!(\textnormal{length } T_2)!\ldots(\textnormal{length } T_l)!}.
$$
On the other hand, the generation of a term as in (\ref{eq:term}) arises in $\varphi(\ve v \shuffle \ve w)$ from
\begin{eqnarray}
{\textnormal{length } S_1 \sqcup T_1 \choose \textnormal{length } S_1}{\textnormal{length } S_2 \sqcup T_2 \choose \textnormal{length } S_2} \cdots {\textnormal{length } S_l \sqcup T_l \choose \textnormal{length } S_l }= \nonumber\\
\frac{(\textnormal{length } S_1 \sqcup T_1)! \cdots (\textnormal{length } S_l \sqcup T_l)!}{(\textnormal{length } S_1)!\ldots(\textnormal{length } S_l)!(\textnormal{length } T_1)!\ldots(\textnormal{length } T_l)!}\nonumber
\end{eqnarray}
distinct terms of the shuffle product $\ve v \shuffle \ve w$, and after application of $\varphi$ each such term acquires a coefficient of
$$
\frac{(-1)^{m+n-l}}{(\textnormal{length } S_1 \sqcup T_1)! \cdots (\textnormal{length } S_l \sqcup T_l)!}.
$$
This proves that $\varphi(\ve v\shuffle \ve w)=\varphi(\ve v)*\varphi(\ve w)$.
\end{proof}
It follows from Theorems \ref{shuffalgpolyalg} and \ref{shuffalgquasishuffalg} that $(R\left\langle A\right\rangle,*)$ is the free polynomial algebra on the elements $\{\varphi(w)| w \textnormal{ is a } Lyndon \textnormal{ word}\}$. In fact the following is true:
\begin{thm}
\label{haupt}
The quasi-shuffle algebra $(R\left\langle A\right\rangle,*)$ is the free polynomial algebra on the \textit{Lyndon} words, i.e., $(R\left\langle A\right\rangle, *) \simeq (R[\textnormal{Lyndon}(A)], \shuffle).$
\end{thm}

\section{Harmonic Sums as Quasi-Shuffle Algebra and Unique Representations}
\label{harasshuf}
Subsequently, we specialize the Quasi-shuffle algebra from Definition \ref{quasishuffprodef} in order to model the harmonic sums accordingly. We consider the alphabet 
$$A_p=\{1,2,3,4,\ldots\},$$ 
and define the degree of a letter $n$ by $\abs{n}=n$. We order the letters by $n \prec m$ for all $n, m \in \mathbb N,$ with the usual order $n < m$ and extend this order lexicographically. Now we define $[a,b]= a+b$ for all $a, b \in A$ and $[a,0]=0$ for all $a \in A_p$. This function obviously fulfills (S0)-(S3) and therefore $(R\left\langle A_p\right\rangle,*)$ with (\ref{eq:quasipro}) is the free polynomial algebra on the \textit{Lyndon} words by Theorem \ref{haupt}.
The harmonic sums defined in (\ref{defHsum}), where we only consider positive integers in the index set, are completely specified by a word $\ve w \in A_p^*$ and the upper summation limit.
E.g., the formula for the multiplication of harmonic sums given in (\ref{hsumproduct}) is identical to $*$ with $[a,b]= a + b$:
$$
S_{a\ve v}(n)* S_{b\ve w}(n)=S_{a(\ve v * b\ve w)}(n)+S_{b(a\ve v * \ve w)}(n)-S_{(a+b)(\ve v*\ve w)}(n). 
$$
Here we need the following definition.
\begin{definition}Let $\ve{a}, \ve{b} \in A^*$. If
$$
\ve{a}*\ve{b}=c_1\ve{d}_1+\cdots+c_m\ve{d}_m
$$
for $\ve{d}_i\in A^*, c_i\in\R$ then we define
$$S_{\ve{a}*\ve{b}}(n):=c_1 S_{\ve{d}_1}{x}+\cdots+c_m S_{\ve{d}_m}(n).$$
\end{definition}
Summarizing, the harmonic sums with indices in the natural numbers form a quasi shuffle algebra.\\
We are now considering the following set:
\begin{eqnarray}
\mathcal{S}_p(n)=\left\{q(S_{\ve a_1},\ldots,S_{\ve a_r})\left| \textnormal{ for all } r\in \N; \right. \ve a_1,\ldots,\ve a_r \in A^*;\ q\in \R[x_1,\ldots,x_r]\right\}.
\label{sp}
\end{eqnarray}
$\mathcal{S}_p(n)$ forms a commutative ring with infinitely many variables.
In addition, we define the ideal $\mathcal{I}_p$ on $\mathcal{S}_p(n):$
\begin{equation}
\mathcal{I}_p:=\left\{ S_{a\ve v}(n)S_{b\ve w}(n)-S_{a(\ve v * b\ve w)}(n)+S_{b(a\ve v * \ve w)}(n)-S_{(a\wedge b)(\ve v*\ve w)}(n)\left|a,b\in \N; \ve v,\ve w \in \N^* \right.\right\}.
\label{idealp}
\end{equation}
Note that $$\mathcal{S}_p(n)/\mathcal{I}_p \cong R\left\langle A\right\rangle$$ by construction. In particular, the linear expansion of $a\in \mathcal{S}_p(n)$ into harmonic sums gives a unique representation of $a+\mathcal{I}_p$. We remark that it has been shown in \cite{Minh2000} that $$R\left\langle A\right\rangle \cong \mathcal{S}_p(n)/\sim $$ where $a(n)\sim b(n)\Leftrightarrow \forall k\in \N a(k)=b(k),\ie$ the quasi-shuffle algebra is equivalent to the harmonic sums considered as sequences.

We now consider the alphabet 
$$A=\{-1,1,-2,2,-3,3,-4,4,\ldots\}.$$
We define the degree of a letter $\abs{\pm n}=n$ for all $n \in \N$, so in each degree there are two letters. We order the letters by $-n\prec n$ and $n \prec m$ for all $n, m \in \mathbb N$ with $\ n < m$, and extend this order lexicographically.
We define $[a,b]= \sign{a}\sign{b}(\abs{a}+\abs{b}) = a \wedge b$ for all $a, b \in A$ and $[a,0]=0$ for all $a \in A$. This function obviously fulfills (S0)-(S3) and therefore $(R\left\langle A\right\rangle,*)$ with (\ref{eq:quasipro}) is the free polynomial algebra on the \textit{Lyndon} words.
The harmonic sums defined in (\ref{defHsum}) (now we allow negative indices) are completely specified by a word $\ve w \in A^*$ and the upper summation limit.

E.g., we observe that the formula for the multiplication of harmonic sums given in (\ref{hsumproduct}) is identical to $*$ with $[a,b]= a \wedge b$: 
$$
S_{a\ve v}(n)* S_{b\ve w}(n)=S_{a(\ve v * b\ve w)}(n)+S_{b(a\ve v * \ve w)}(n)-S_{(a\wedge b)(\ve v*\ve w)}(n); 
$$
subsequently, we will write $S_{\ve a}(n)S_{\ve b}(n)$ instead of$S_{\ve a}(n)*S_{\ve b}(n)$ for any index sets $\ve a, \ve b.$
Summarizing, the harmonic sums with indices in $\Z/\{0\}$ form a quasi shuffle algebra.

We are now considering the following set:
\begin{eqnarray}
\mathcal{S}(n)=\left\{q(S_{\ve a_1},\ldots,S_{\ve a_r})\left| \textnormal{ for all } r\in \N; \right. \ve a_1,\ldots,\ve a_r \in A^*;\ q\in \R[x_1,\ldots,x_r]\right\}.
\label{s}
\end{eqnarray}
Moreover, we define the ideal $\mathcal{I}$ of $\mathcal{S}(n)$ by
\begin{equation}
\mathcal{I}:=\left\{ S_{a\ve v}(n)S_{b\ve w}(n)-S_{a(\ve v * b\ve w)}(n)+S_{b(a\ve v * \ve w)}(n)-S_{(a\wedge b)(\ve v*\ve w)}(n)\left|a,b\in \Z; \ve v,\ve w \in (\Z/\{0\})^* \right.\right\}.
\label{ideal}
\end{equation}
Note that $$\mathcal{S}(n)/\mathcal{I} \cong R\left\langle A\right\rangle $$ by construction. In particular, the linear expansion of $a\in \mathcal{S}(n)$ into harmonic sums gives a unique representation of $a+\mathcal{I}$. To our knowledge it has not been shown so far that $$R\left\langle A\right\rangle \cong \mathcal{S}(n)/\sim $$ where $a(n)\sim b(n)\Leftrightarrow \forall k\in \N a(k)=b(k)$. Nevertheless, we strongly believe in this fact; see also Remark \ref{X}.
\begin{example}
Let $A=S_{1, 2}(n)S_{2}(n)$ and $B=S_{2}(n)^2S_{1}(n) + S_{2}(n) S_{3}(n) - S_{2}(n) S_{2, 1}(n)$. In $\mathcal{S}(n)/\mathcal{I}$ we can test equality as follows. Using the product formula we get
\begin{eqnarray*}
A=-S_{1, 4}(n) - S_{3, 2}(n) + 2 S_{1, 2, 2}(n) + S_{2, 1, 2}(n)\\
B=-S_{1, 4}(n) - S_{3, 2}(n) + 2 S_{1, 2, 2}(n) + S_{2, 1, 2}(n),
\end{eqnarray*}
thus $A$ is equal to $B$.
\end{example}
If we expand all products, the number of sums which are contained in the expression will grow relatively fast. In the previous example the expanded version of the expressions contains 4 sums, while there are just 2 in $A$. This tends to be even worse if we deal with sums of higher weight or depth. Subsequently, we will introduce another way to handle the problem of unique representations such that the canonical representation of the expression is as small as possible (small with respect to a certain order, which has to be defined). We will attack this problem in Section \ref{deriverel}.

\section{The Number of \textit{Lyndon} Words or Basic Sums}
Subsequently, we will count the algebraic independent sums in $\mathcal{S}_p(n)/\mathcal{I}$, respectively $\mathcal{S}(n)/\mathcal{I}$ which we also call basic sums. In the light of Theorem \ref{haupt} this is connected to the number of \textit{Lyndon} words. The number of \textit{Lyndon} words of length $n$ over an alphabet of length $q$ is given by the first Witt formula \cite{Witt1937,Witt1956,Reutenauer1969}:
\begin{equation}
l_n(q)=\frac{1}{n}\sum_{d|n}{\mu(d)q^{n/d}},
\end{equation} 
where 
\begin{equation}
	\mu(n)=\left\{ 
		  	\begin{array}{ll}
						1\  & \textnormal{if } n = 1  \\
						0\  & \textnormal{if } n \textnormal{ is divided by the square of a prime}  \\
						(-1)^s\  & \textnormal{if } n \textnormal{ is the product of } s \textnormal{ different primes.} 
					\end{array} \right. \nonumber\\
\label{abmue}
\end{equation}
is the Möbius function.\\
As we would like to count the number of basic sums for all sums of a given index set individually, this relation cannot be used: we have to count the number of \textit{Lyndon} words belonging only to this index set. The respective number has been given in the same paper as the second Witt formula,
\begin{equation}
l_n(n_1,\ldots,n_q)=\frac{1}{n}\sum_{d|n}{\mu(d)\frac{(\frac{n}{d})!}{(\frac{n_1}{d})!\cdots(\frac{n_q}{d})!}}, \ \ n=\sum_{k=1}^q{n_k}
\label{wittform}
\end{equation} 
here $n_i$ denotes the multiplicity of the indices that appear in the index set. For more details on these aspects we refer to \cite{Reutenauer1969}. In other words, by Theorem \ref{haupt} and with $$\mathcal{S}_p(n)/\mathcal{I}_p \cong R\left\langle \N \right\rangle \cong \mathcal{S}_p(n)/\sim$$ the number of algebraic independent sums in $\mathcal{S}_p(n)/\mathcal{I}_p$ or in $\mathcal{S}_p(n)/\sim$ (considered as sequences) for a given index set is
\begin{equation}
l_n(n_1,\ldots,n_q)=\frac{1}{n}\sum_{d|n}{\mu(d)\frac{(\frac{n}{d})!}{(\frac{n_1}{d})!\cdots(\frac{n_q}{d})!}}, \ \ n=\sum_{k=1}^q{n_k}
\end{equation} 
again $n_i$ denotes the multiplicity of the indices that appear in the index set.
Moreover, the number of algebraic independent sums in $$\mathcal{S}(n)/\mathcal{I} \cong R\left\langle \Z/\{0\}\right\rangle$$ for a given index set is
\begin{equation}
l_n(n_1,\ldots,n_q)=\frac{1}{n}\sum_{d|n}{\mu(d)\frac{(\frac{n}{d})!}{(\frac{n_1}{d})!\cdots(\frac{n_q}{d})!}}, \ \ n=\sum_{k=1}^q{n_k}.
\end{equation} 
As pointed out, e.g., in \cite{Bluemlein2004,Bluemlein2008} the number of algebraic independent sums can be summarized in Table \ref{algebraictab}.
\begin{remark}
\label{X}
In the difference field setting of $\Pi \Sigma$-fields one can verify algebraic independence of sums algorithmically for a particular given finite set of sums; see \cite{Schneider2008}. We could easily verify up to weight 7 that the figures in Table \ref{algebraictab} are correct, interpreting the objects in $\mathcal{S}(n)/\sim$ ,$\ie$ as sequences. Nevertheless, unless we do not have a rigorous proof for $\mathcal{S}(n)/\mathcal{I} \cong \mathcal{S}(n)/\sim$, we can only assume that the figures in Table \ref{algebraictab} give an upper bound.
\end{remark}

\begin{table}
\begin{tabular}{|| r || r | r || r | r || }
\hline	
&  \multicolumn{4}{|c||}{Number of} \\
\cline{2-5}
Weight& Sums& a-basic sums& Sums $\neg \left\{-1\right\}$ &a-basic sums\\
\hline	
  1 &    2 &   2 &   1 &   1 \\
  2 &    6 &   3 &   3 &   2 \\
  3 &   18 &   8 &   7 &   4 \\
  4 &   54 &  18 &  17 &   7 \\
  5 &  162 &  48 &  41 &  16 \\ 
  6 &  486 & 116 &  99 &  30 \\
  7 & 1458 & 312 & 239 &  68 \\  
  8 & 4374 & 810 & 577 & 140 \\
\hline
\end{tabular}
\caption{Number of harmonic sums and number of sums which do not contain the index $\left\{-1\right\}$ in dependence on their weight; respectively, the numbers of basic sums (a-basic sums) by which all sums can be expressed using the algebraic relations; compare also \cite{Bluemlein2004,Bluemlein2009}.}
\label{algebraictab}
\end{table}

\section{Deriving of Relations between Harmonic Sums}
\label{deriverel}
In order to express the harmonic sums of a given depth in terms of a minimal set of harmonic sums, we will determine all algebraic relations of the harmonic sums up to weight 7 following the ideas from \cite{Bluemlein2004}.
We define the shuffle product on harmonic sums in the obvious way:
\begin{equation}
S_{a\ve v}(n)\shuffle S_{b\ve w}(n)=S_{a(\ve v \shuffle b\ve w)}(n)+S_{b(a\ve v \shuffle \ve w)}(n),
\label{hshuffpro} 
\end{equation}
for all $a,b \in A$ and $\ve v,\ve w \in A^*.$\\
As an example we state
\begin{eqnarray}
S_{a_1,a_2}(n)\shuffle S_{a_3,a_4}(n)&=&S_{a_1,a_2,a_3,a_4}(n)+S_{a_1,a_3,a_2,a_4}(n)+S_{a_1,a_3,a_4,a_2}(n)+ \nonumber\\
&&S_{a_3,a_4,a_1,a_2}(n)+S_{a_3,a_1,a_4,a_2}(n)+S_{a_3,a_1,a_2,a_4}(n).\nonumber
\end{eqnarray}
We get the following relations listed up to depth 4, for further relations up to depth 6 see \cite{Bluemlein2004}:
\begin{description}
\item[Depth 2]
\footnotesize
	\begin{eqnarray}
	0=S_{a_1}(n) \shuffle S_{a_2}(n) - S_{a_1 \wedge a_2}(n) - S_{a_1}(n) S_{a_2}(n)\nonumber
	\end{eqnarray}
\normalsize
\item[Depth 3]
\footnotesize
	\begin{eqnarray}
	0=S_{a_1}(n) \shuffle S_{a_2, a_3}(n) - S_{a_1 \wedge a_2, a_3}(n) - S_{a_2, a_1 \wedge a_3}(n) - S_{a_1}(n) S_{a_2, a_3}(n)\nonumber
	\end{eqnarray}
\normalsize
\item[Depth 4]	
\footnotesize
	\begin{eqnarray*}
	0&=&S_{a_1}(n) \shuffle S_{a_2, a_3, a_4}(n) - S_{a_1 \wedge a_2, a_3, a_4}(n) - S_{a_2, a_1 \wedge a_3, a_4}(n)
			-S_{a_2, a_3, a_1 \wedge a_4}(n) \\ 
	&&- S_{a_1}(n) S_{a_2, a_3, a_4}(n)  \\
 0&=&S_{a_1,a_2}(n)\shuffle S_{a_3,a_4}(n)+S_{a_1\wedge a_3,a_2\wedge a_4}(n)-S_{a_1,a_2}(n)S_{a_3,a_4}(n)-S_{a_1\wedge a_3,a_2,a_4}(n)\\ 
   &&-S_{a_1 \wedge a_3, a_4, a_2}(n)-S_{a_1, a_2 \wedge a_3, a_4}(n)-S_{a_1, a_3, a_2 \wedge a_4}(n)-S_{a_3, a_1 \wedge a_4, a_2}(n) \\ 
   &&- S_{a_3, a_1, a_2 \wedge a_4}(n)  
	\end{eqnarray*}
\normalsize		
\end{description}

As observed in \cite{Bluemlein1999} and worked out in detail in \cite{Bluemlein2004} these relations in combination with the Witt formula enable us to hunt systematically for all algebraic relations in $\mathcal{S}(n)/\mathcal{I}$.

\begin{example}
We now consider the case of harmonic sums of depth 3 as an example. We start with 3 different indices. There exist 6 different sums of this type:
$$
S_{a_1, a_2, a_3}(n), S_{a_1, a_3, a_2}(n), S_{a_2, a_1, a_3}(n), S_{a_2, a_3, a_1}(n), S_{a_3, a_1, a_2}(n), S_{a_3, a_2, a_1}(n).
$$
Now we use the relation we obtained for sums of depth 3 and apply it for all 6 permutations of the 3 indices. This leads to the following system of equations:
\footnotesize
\begin{eqnarray}
  S_{a_1, a_2, a_3}(n) + S_{a_2, a_1, a_3}(n) + S_{a_2, a_3, a_1}(n) &=& S_{a1}(n) S_{a2, a3}(n) + S_{a2, a1 \wedge a3}(n) + S_{a1 \wedge a2, a3}(n)\nonumber\\
  S_{a_1, a_3, a_2}(n) + S_{a_3, a_1, a_2}(n) + S_{a_3, a_2, a_1}(n) &=& S_{a1}(n) S_{a3, a2}(n) + S_{a3, a1 \wedge a2}(n) + S_{a1 \wedge a3, a2}(n)\nonumber\\
  S_{a_1, a_2, a_3}(n) + S_{a_1, a_3, a_2}(n) + S_{a_2, a_1, a_3}(n) &=& S_{a2}(n) S_{a1, a3}(n) + S_{a1, a2 \wedge a3}(n) + S_{a1 \wedge a2, a3}(n)\nonumber\\
  S_{a_2, a_3, a_1}(n) + S_{a_3, a_1, a_2}(n) + S_{a_3, a_2, a_1}(n) &=& S_{a2}(n) S_{a3, a1}(n) + S_{a3, a1 \wedge a2}(n) + S_{a2 \wedge a3, a1}(n)\nonumber\\
  S_{a_1, a_2, a_3}(n) + S_{a_1, a_3, a_2}(n) + S_{a_3, a_1, a_2}(n) &=& S_{a3}(n) S_{a1, a2}(n) + S_{a1, a2 \wedge a3}(n) + S_{a1 \wedge a3, a2}(n)\nonumber\\
  S_{a_2, a_1, a_3}(n) + S_{a_2, a_3, a_1}(n) + S_{a_3, a_2, a_1}(n) &=& S_{a3}(n) S_{a2, a1}(n) + S_{a2, a1 \wedge a3}(n) + S_{a2 \wedge a3, a1}(n)\nonumber
\end{eqnarray}
\normalsize
which is equivalent to:
\footnotesize
\begin{eqnarray}
\left(
	\begin{array}[pos]{cccccc}
		1&0&1&1&0&0 \\
		0&1&0&0&1&1 \\
		1&1&1&0&0&0 \\
		0&0&0&1&1&1 \\
		1&1&0&0&1&0 \\
		0&0&1&1&0&1 	
	\end{array}
\right)\cdot
\left(
	\begin{array}[pos]{c}
		S_{a_1, a_2, a_3}(n)\\
		S_{a_1, a_3, a_2}(n)\\
		S_{a_2, a_1, a_3}(n)\\
		S_{a_2, a_3, a_1}(n)\\
		S_{a_3, a_1, a_2}(n)\\
		S_{a_3, a_2, a_1}(n)	
	\end{array}
\right)=
\left(
	\begin{array}[pos]{c}
		S_{a_1}(n) S_{a_2, a_3}(n) + S_{a_2, a_1 \wedge a_3}(n) + S_{a_1 \wedge a_2, a_3}(n)\\
		S_{a_1}(n) S_{a_3, a_2}(n) + S_{a_3, a_1 \wedge a_2}(n) + S_{a_1 \wedge a_3, a_2}(n)\\
		S_{a_2}(n) S_{a_1, a_3}(n) + S_{a_1, a_2 \wedge a_3}(n) + S_{a_1 \wedge a_2, a_3}(n)\\
		S_{a_2}(n) S_{a_3, a_1}(n) + S_{a_3, a_1 \wedge a_2}(n) + S_{a_2 \wedge a_3, a_1}(n)\\
		S_{a_3}(n) S_{a_1, a_2}(n) + S_{a_1, a_2 \wedge a_3}(n) + S_{a_1 \wedge a_3, a_2}(n)\\
		S_{a_3}(n) S_{a_2, a_1}(n) + S_{a_2, a_1 \wedge a_3}(n) + S_{a_2 \wedge a_3, a_1}(n)\nonumber
	\end{array}
\right)
\end{eqnarray}
\normalsize
Since the rank of this linear system is 4, all sums with 3 different indices can be expressed by two chosen sums of depth 3 and sums of lower depth. To be more precise, solving this system leads to 
\footnotesize
\begin{eqnarray}
S_{a_1, a_2, a_3}(n) &=&  S_{a_1, a_2 \wedge a_3}(n) + S_{a_3}(n) S_{a_1, a_2}(n) - S_{a_3, a_1 \wedge a_2}(n) - S_{a_1}(n) S_{a_3, a_2}(n) + S_{a_3, a_2, a_1}(n), \nonumber\\ 
S_{a_1, a_3, a_2}(n) &=&  S_{a_1 \wedge a_3, a_2}(n) + S_{a_3, a_1 \wedge a_2}(n) + S_{a_1}(n) S_{a_3, a_2}(n) - S_{a_3, a_1, a_2}(n) - S_{a_3, a_2, a_1}(n), \nonumber\\
S_{a_2, a_1, a_3}(n) &=&  S_{a_1 \wedge a_2, a_3}(n) - S_{a_1 \wedge a_3, a_2}(n) - S_{a_3}(n) S_{a_1, a_2}(n) + S_{a_2}(n) S_{a_1, a_3}(n) + S_{a_3, a_1, a_2}(n), \nonumber\\
S_{a_2, a_3, a_1}(n) &=&  S_{a_1 \wedge a_3, a_2}(n) - S_{a_1, a_2 \wedge a_3}(n) - S_{a_2}(n) S_{a_1, a_3}(n) + S_{a_2, a_1 \wedge a_3}(n) + S_{a_1}(n) S_{a_2, a_3}(n) \nonumber\\
&& + S_{a_3, a_1 \wedge a_2}(n) + S_{a_1}(n) S_{a_3, a_2}(n) - S_{a_3, a_1, a_2}(n) - S_{a_3, a_2, a_1}(n)\nonumber
\end{eqnarray}
\normalsize
For two different indices we get the following system of equations:
\footnotesize
\begin{eqnarray}
\left(
	\begin{array}[pos]{ccc}
		2&1&0 \\
		0&1&2 \\
		1&1&1	
	\end{array}
\right)\cdot
\left(
	\begin{array}[pos]{c}
		S_{a_1, a_1, a_2}(n)\\
		S_{a_1, a_2, a_1}(n)\\
		S_{a_2, a_1, a_1}(n)
	\end{array}
\right)=
\left(
	\begin{array}[pos]{c}
		S_{a_1}(n) S_{a_1, a_2}(n) + S_{a_1, a_1 \wedge a_2}(n) + S_{a_1 \wedge a_1, a_2}(n)\\
		S_{a_1}(n) S_{a_2, a_1}(n) + S_{a_2, a_1 \wedge a_1}(n) + S_{a_1 \wedge a_2, a_1}(n)\\
		S_{a_2}(n) S_{a_1, a_1}(n) + S_{a_1, a_1 \wedge a_2}(n) + S_{a_1 \wedge a_2, a_1}(n).	\nonumber
	\end{array}
\right)
\end{eqnarray}
\normalsize
Since the rank of this linear system is 2, all sums with 2 different indices of depth 3 can be expressed by one chosen sum of depth 3 and sums of lower depth:
\footnotesize
\begin{eqnarray}
S_{a_1, a_1, a_2}(n) &=&  \frac{1}{2}[S_{a_1 \wedge a_1, a_2}(n) -  S_{a_1 \wedge a_2, a_1}(n) +  S_{a_1, a_1 \wedge a_2}(n) +  S_{a_1}(n) S_{a_1, a_2}(n) \nonumber \\ &&- S_{a_2, a_1 \wedge a_1}(n) -  S_{a_1}(n) S_{a_2, a_1}(n) + S_{a_2, a_1, a_1}(n)] + S_{a_2, a_1, a_1}(n), \nonumber \\
S_{a_1, a_2, a_1}(n) &=&  S_{a_1 \wedge a_2, a_1}(n) + S_{a_2, a_1 \wedge a_1}(n) + S_{a_1}(n) S_{a_2, a_1}(n) - 2 S_{a_2, a_1, a_1}(n). \label{depht32ind}
\end{eqnarray}
\normalsize
For a sum of depth three with 3 equal indices we obtain:
\footnotesize
\begin{eqnarray}
S_{a_1, a_1, a_1}(n) &=&  \frac{1}{3}[S_{a_1 \wedge a_1, a_1}(n) +  S_{a_1, a_1 \wedge a_1}(n) + S_{a_1}(n) S_{a_1, a_1}(n)]. \nonumber 
\end{eqnarray}
\normalsize
\begin{remark}
A different method which can be used for equal indices and is described in \cite{Bluemlein2004} leads to:
\footnotesize
\begin{eqnarray}
S_{a_1, a_1, a_1}(n) &=&  \frac{1}{6}[S_{a_1}^3(n) +  3 S_{a_1}*S_{a_1 \wedge a_1}(n) + 2*S_{a_1 \wedge a_1 \wedge a_1}(n)]. \nonumber 
\end{eqnarray}
\normalsize
\end{remark}
Applying the second Witt formula (\ref{wittform}), \ie
\begin{eqnarray}
l_n(1,1,1)=2\nonumber \\
l_n(2,1)=1\nonumber \\
l_n(3)=0\nonumber
\end{eqnarray}
we can see that we found the maximal number of algebraic independent sums. In other words, this strategy solved the problem of weight 3 completely.
\label{exgetrel}
\end{example}
Using \ttfamily HarmonicSums, \rmfamily these relations can be computed as follows:
\begin{fmma}
\begin{mma}
\In {\text{\bf GetDependentSSums[\{1, 1, 1\}]}\\
		\text{sums: \ 6}\\
		\text{dependent sums: \ 4}}\\
\Out {\{\text{S}[a1\_, a2\_, a3\_, n\_] \rightarrow \text{S}[a3, n] \text{S}[a1, a2, n] + \text{S}[a1, \text{SP}[a2, a3], n] - \text{S}[a1, n] \text{S}[a3, a2, n] - \text{S}[a3, \text{SP}[a1, a2], n] + \text{S}[a3, a2, a1, n] /; 
     \text{MyMemberQ}[{{a1, a2, a3}}, {a1, a2, a3}] \&\& a1 \neq a2 \neq a3, \\ 
   \text{S}[a1\_, a3\_, a2\_, n\_] \rightarrow \text{S}[a1, n] \text{S}[a3, a2, n] + \text{S}[a3, \text{SP}[a1, a2], n] + \text{S}[\text{SP}[a1, a3], a2, n] - \text{S}[a3, a1, a2, n] - \text{S}[a3, a2, a1, n] /; 
     \text{MyMemberQ}[{{a1, a2, a3}}, {a1, a2, a3}] \&\& a1 \neq a2 \neq a3, \\
   \text{S}[a2\_, a1\_, a3\_, n\_] \rightarrow -\text{S}[a3, n] \text{S}[a1, a2, n] + \text{S}[a2, n] \text{S}[a1, a3, n] + \text{S}[\text{SP}[a1, a2], a3, n] - \text{S}[\text{SP}[a1, a3], a2, n] + 
      \text{S}[a3, a1, a2, n] /; \text{MyMemberQ}[{{a1, a2, a3}}, {a1, a2, a3}] \&\& a1 \neq a2 \neq a3,\\ 
   \text{S}[a2\_, a3\_, a1\_, n\_] \rightarrow -\text{S}[a2, n] \text{S}[a1, a3, n] - \text{S}[a1, \text{SP}[a2, a3], n] + \text{S}[a1, n] \text{S}[a2, a3, n] + \text{S}[a2, \text{SP}[a1, a3], n] + 
      \text{S}[a1, n] \text{S}[a3, a2, n] + \text{S}[a3, \text{SP}[a1, a2], n] + \text{S}[\text{SP}[a1, a3], a2, n] - \text{S}[a3, a1, a2, n] - \text{S}[a3, a2, a1, n] /; 
     \text{MyMemberQ}[{{a1, a2, a3}}, {a1, a2, a3}] \&\& a1 \neq a2 \neq a3\}}\\
\end{mma}
\begin{mma}     
\In {\text{\bf GetDependentSSums[\{2, 1\}]}\\
		\text{sums: \ 3}\\
		\text{dependent sums: \ 2}}\\
\Out {\{\text{S}[a1\_, a1\_, a2\_, n\_] \rightarrow 
    \frac{1}{2} \text{S}[a1, n] \text{S}[a1, a2, n] + \frac{1}{2} \text{S}[a1, \text{SP}[a1, a2], n] - 
        \frac{1}{2} \text{S}[a1, n] \text{S}[a2, a1, n] - \frac{1}{2} \text{S}[a2, \text{SP}[a1, a1], n] + 
        \frac{1}{2} \text{S}[\text{SP}[a1, a1], a2, n] - \frac{1}{2} \text{S}[\text{SP}[a1, a2], a1, n] + 
        \text{S}[a2, a1, a1, n] /; a1 \neq a2,\\ 
  \text{S}[a1\_, a2\_, a1\_, n\_] \rightarrow 
    \text{S}[a1, n] \text{S}[a2, a1, n] + \text{S}[a2, \text{SP}[a1, a1], n] + \text{S}[\text{SP}[a1, a2], a1, n] - 
        2 \text{S}[a2, a1, a1, n] /; a1 \neq a2\}}\\
\end{mma} 
\end{fmma}

\section{Application of the Relations and Algebraic Simplification}
\label{Application of the Relations}
There are now basically three ways to use the strategy presented in the last section.

\subsection{First Strategy: Fixed Tables}

We can fix basic sums and compute tables up to a certain weight. Then we can use these tables to reduce expressions in harmonic sums 		to expressions where only the chosen basic sums appear. Tables up to weight 6 are produced in \cite{Bluemlein2004} whose basic sums are of particular interest in particle physics. For the package \ttfamily HarmonicSums \rmfamily these relations were recomputed and can now be used in the following way.

\begin{fmma}
\noindent Relations up to weight 6 are included in the package.
\begin{mma}
\In {\text{\bf expr = \text{S}[1, 2, -2, n] + \text{S}[-1, 2, 1, n] + \text{S}[-1, 2, -2, n] + 
      \text{S}[-2, 2, 1, n] + \text{S}[-2, 1, 2, n]}\\
     \text{\bf + \text{S}[1, -2, -1, n] + \text{S}[-1, 2, 1, n] +
       \text{S}[2, -1, 2, n] + \text{S}[-1, 2, 2, n];}}\\
\end{mma}          
\begin{mma}          
\In {\text{\bf ReduceToBasis[expr]}}\\       
\Out {\text{S}[-5, n] + \frac{1}{2} \text{S}[-2, n]^2 + \text{S}[-4, n] \text{S}[-1, n] + \text{S}[-4, n] \text{S}[1, n] + 
  \text{S}[-2, n] \text{S}[1, n] \text{S}[2, n] + \text{S}[-2, n] \text{S}[3, n] + \frac{1}{2} \text{S}[4, n] + 
  \text{S}[-1, n] [\text{S}[2, n]^2 + \text{S}[4, n]] - \text{S}[-4, -1, n] - \text{S}[-4, 1, n] + \text{S}[-3, -2, n] +
   \text{S}[-3, -1, n] + 2 \text{S}[-3, 2, n] + \text{S}[1, n] \text{S}[-2, -1, n] + \text{S}[2, n] \text{S}[-2, 1, n] -
   \text{S}[-1, n] [\text{S}[-4, n] + \text{S}[-2, n] \text{S}[2, n] - \text{S}[2, -2, n]] - 
  \text{S}[1, n] [\text{S}[-4, n] + \text{S}[-2, n] \text{S}[2, n] - \text{S}[2, -2, n]] + 
  \text{S}[-2, n] [\text{S}[-1, n] \text{S}[2, n] - \text{S}[2, -1, n]] - \text{S}[2, n] \text{S}[2, -1, n] - 
  2 \text{S}[-2, n] \text{S}[2, 1, n] - \text{S}[4, -1, n] - \text{S}[-2, -1, 1, n] - \text{S}[-2, 1, -1, n] + 
  \text{S}[-2, 2, -1, n] + 2 \text{S}[-2, 2, 1, n] - \text{S}[-1, 2, 2, n] + \text{S}[2, 1, -2, n] + 
  2 [\text{S}[-3, 1, n] + \text{S}[2, -2, n] + \text{S}[-1, n] \text{S}[2, 1, n] - \text{S}[2, -1, 1, n] - 
        \text{S}[2, 1, -1, n]]}\\
\end{mma}
\end{fmma}     
The disadvantage is however, that for certain applications it might be more reasonable to chose other basic sums. In this case we cannot use this strategy.	

\subsection{Second Strategy: Partly Fixed Tables}
We can compute tables for all possible index sets up to a certain depth, which is done in \cite{Bluemlein2004} up to depth 6. These tables have been extended up to depth 7 during the work on this thesis. The computation of these tables took about one and a half week CPU time on a 2.3 GHz machine, and ended in a file with about 300 Mbytes. We can use these tables to reduce expressions in harmonic sums to expressions where for each depth and index set at most that many sums appear as is given by the second Witt formula. 
\begin{fmma}
\noindent First we have to load the table where the relations are stored:
\begin{mma}
\In {\text{\bf DefineRelationTable["RelTab.m"]}}\\
\end{mma}
\begin{mma}
\In {\text{\bf expr = \text{S}[1, 2, -2, n] + \text{S}[-1, 2, 1, n] + \text{S}[-1, 2, -2, n] + 
      \text{S}[-2, 2, 1, n] + \text{S}[-2, 1, 2, n]}\\
     \text{\bf + \text{S}[1, -2, -1, n] + \text{S}[-1, 2, 1, n] +
       \text{S}[2, -1, 2, n] + \text{S}[-1, 2, 2, n];}}\\
\end{mma} 
\begin{mma}        
\In {\text{\bf ReduceToBasis[expr, Dynamic $\rightarrow$ Partly]}}\\       
\Out {\text{S}[-5, n] + 
  \text{S}[-3, n] \text{S}[2, n] + \frac{1}{2} \text{S}[-1, n] [\text{S}[2, n]^2 + \text{S}[4, n]] - 
  \text{S}[-4, 1, n] + \text{S}[1, n] \text{S}[2, -2, n] - \text{S}[-2, n] \text{S}[2, 1, n] + 
  \text{S}[3, -2, n] + \text{S}[-2, 1, 2, n] + 2 \text{S}[-2, 2, 1, n] + \text{S}[-1, 2, -2, n] + 
  2 \text{S}[-1, 2, 1, n] + \text{S}[1, -2, -1, n] - \text{S}[2, 2, -1, n]}\\
\end{mma}
\end{fmma}
However there is the disadvantage that although for each depth and index set the relations can be calculated quite generally, we have to fix in advance certain sums which have to be in the basis. We can easily see this in the following example:\\
We again consider sums of depth 3 with 2 different indices, so we look at the sums $S_{a_1, a_1, a_2}(n)$, $S_{a_1, a_2, 								a_1}(n)$, $S_{a_2, a_1, a_1}(n).$ According to the Witt formula there should be just one of the sums in the basis. In 										(\ref{depht32ind}) we took $S_{a_2, a_1, a_1}(n)$ to express the other sums. Therefore for example the sum $S_{1,1,2}(n)$ will 					always be replaced by $S_{2,1,1}(n)$ and sums of lower depth. Sometimes however it could be more convenient 				 to take $S_{2,1,1}(n)$ as a basis element.
				
\subsection{Third Strategy: Dynamic Reduction}
If we look at an expression consisting out of several different harmonic sums, we could first look at all the different index 					sets of harmonic sums that appear in the expression. For each index set $J$ we can use the Witt formula to get the number $n_J$ 				of basis elements, \ie the basic sums. Then for each index set we can use our method to set up a system of equations where we choose exactly $n$ of the 		permutations (and therefore $n$ harmonic sums) to be in the basis. Here we will try, as far as possible to choose sums which are 						already in the expression, since we do not want to introduce new sums. Sometimes however we will have to choose also sums which 					are not in the expression. By solving the equation system the rest of the sums can be expressed by harmonic sums out of the 						basis and sums of lower depth. So we end up with an expression, where for each index set $J$ at most $n_J$ different multiple 					harmonic sums appear.    
Again we can illustrate this by an example. 
\begin{example}
Consider an expression where the sums $S_{a_1, a_1, a_2}(n)$ and $S_{a_1, a_2, a_1}(n)$ appear. We will use one of the sums to express the other. Since the Witt formula gives $1,$ this is possible. In the reduced expression there will be either the first, the second sum or no sum with this index set. Using, e.g., the second strategy we would use (as defined in Example \ref{exgetrel}) $S_{a_2,a_1,a_1}(n)$, that is not in the original expression.
\end{example}

\begin{fmma}
\noindent Of course, now we do not have to load any tables, since they are computed online.
\begin{mma}
\In {\text{\bf expr = \text{S}[1, 2, -2, n] + \text{S}[-1, 2, 1, n] + \text{S}[-1, 2, -2, n] + 
      \text{S}[-2, 2, 1, n] + \text{S}[-2, 1, 2, n]}\\
     \text{\bf + \text{S}[1, -2, -1, n] + \text{S}[-1, 2, 1, n] +
       \text{S}[2, -1, 2, n] + \text{S}[-1, 2, 2, n];}}\\
\end{mma} 
\begin{mma}  
\In {\text{\bf ReduceToBasis[expr, n, Dynamic $\rightarrow$ True]}}\\       
\Out {\text{S}[-5, n] + 
  \text{S}[-3, n] \text{S}[2, n] + \frac{1}{2} \text{S}[-1, n] [\text{S}[2, n]^2 + \text{S}[4, n]] + 
  \text{S}[1, n] \text{S}[-2, 2, n] + \text{S}[-2, 3, n] - \text{S}[1, -4, n] - 
  \text{S}[-2, n] \text{S}[1, 2, n] - \frac{1}{2} \text{S}[2, -3, n] - \frac{1}{2} \text{S}[2, n] \text{S}[
      2, -1, n] - \frac{1}{2} \text{S}[4, -1, n] + \text{S}[-2, 1, 2, n] + 
  \text{S}[-1, 2, -2, n] + 2 \text{S}[-1, 2, 1, n] + \text{S}[1, -2, -1, n] + 
  2 \text{S}[1, 2, -2, n] + \frac{1}{2} \text{S}[2, -1, 2, n]}\\
\end{mma} 
\end{fmma}
This strategy works quite fine (which means fast) for index sets where the number of permutations is not too high. For example 					there are 60 permutations of an index set of the type $\left\{a_1,a_1,a_2,a_3,a_4\right\}$ in such a situation the strategy works really fast. 				But if we consider the case of an index set of the type $\left\{a_1,a_1,a_2,a_3,a_4,a_5,a_6\right\},$ there are 2520 										permutations. If we have to set up and solve the system of equations in this case, this will take hours.\\

\subsection{Algebraic Simplification}
We use the notions of \cite{Buchberger1983} concerning \itshape canonical algebraic simplification. \upshape Especially the first strategy can be used to obtain \itshape equivalent but simpler objects \upshape and to compute \itshape unique representations for equivalent objects \upshape in $\mathcal{S}(n)/\mathcal{I};$ here $I$ is again the ideal in (\ref{ideal}).
For $s,t \in \mathcal{S}(n),$ we say $s\leq t$ if the number of multiple harmonic sums in $s$ which are not basic sums is smaller or equal than the number of non basic sums in $t$. Now consider the following problem:\\

\bfseries Given: \normalfont $s\in \mathcal{S}(n)$.\\
\bfseries Find: \normalfont $t\in \mathcal{S}(n)$ such that $s+I=t+I$ and such that for all $a\in \mathcal{S}(n)$ with $a+I=t+I$: $t\leq a$.\\

Now let $T$ be the procedure which uses the first strategy. $T$ maps $\mathcal{S}(n)$ to $\mathcal{S}(n).$ For all $s, t \in\mathcal{S}(n)$ with $s+\mathcal{I}=t+\mathcal{I}$ we have the equality $T(s)=T(t)$ in $\mathcal{S}(n)$ by construction; $T(s)\neq T(t)$ would mean that algebraic relations are missing. Hence we can check equivalence easily.\\
Additionally, the number of non basic harmonic sums in $T(s)$ is zero; hence the above problem of simplification is solved completely.

\section{Relations between Similar Structures}
In the same way one can find relations for Euler-Zagier sums, here the product (\ref{zpro}) is the quasi shuffle product discussed in \cite{Hoffman}: 
\begin{eqnarray}
\epsilon *\ve w &=& \ve w* \epsilon = \ve w, \ \textnormal{for all}\ \ve w \in A^* \nonumber \\
a\ve u*b\ve v &=& a(\ve u*\ve v)+b(\ve u*\ve v)+[a,b](\ve u*\ve v), \ \textnormal{for all}\ a,b \in A;\ve u,\ve v \in A^*\nonumber
\end{eqnarray}
with $[a,b]=a\wedge b$ for all $a,b \in \Z$.\\
Again the number of \textit{Lyndon} words gives the number of algebraically independent Euler-Zagier sums viewed as elements in a ring similarly as $\mathcal{S}(n)/\mathcal{I}$ for harmonic sums and it gives at least an upper bound for the number of algebraic independent Euler-Zagier sums viewed as sequences. In the case of just positive indices it was proven in \cite{Minh2000} that the number of \textit{Lyndon} words is not only an upper bound for the number of algebraically independent Euler-Zagier sums viewed as sequences, it is in fact the number of algebraically independent Euler-Zagier sums viewed as sequences. It is likely that the same holds for index sets where we also allow negative indices.
Like for harmonic sums the package \ttfamily HarmonicSums \rmfamily provides a procedure to calculate the relations automatically:
\begin{fmma}
\In {\text{\bf GetDependentZSums[{1, 1, 1}]}\\
		\text{sums: \ 6}\\
		\text{dependent sums: \ 4}}\\
\Out {\{\text{Z}[a1\_, a2\_, a3\_, n\_] \rightarrow \text{Z}[a3, n] \text{Z}[a1, a2, n] - \text{Z}[a1, \text{SP}[a2, a3], n] - \text{Z}[a1, n] \text{Z}[a3, a2, n] + \text{Z}[a3, \text{SP}[a1, a2], n] + \text{Z}[a3, a2, a1, n] /; 
     \text{MyMemberQ}[{{a1, a2, a3}}, {a1, a2, a3}] \&\& a1 \neq a2 \neq a3, \\
   \text{Z}[a1\_, a3\_, a2\_, n\_] \rightarrow \text{Z}[a1, n] \text{Z}[a3, a2, n] - \text{Z}[a3, \text{SP}[a1, a2], n] - \text{Z}[\text{SP}[a1, a3], a2, n] - \text{Z}[a3, a1, a2, n] - \text{Z}[a3, a2, a1, n] /; 
     \text{MyMemberQ}[{{a1, a2, a3}}, {a1, a2, a3}] \&\& a1 \neq a2 \neq a3,\\ 
   \text{Z}[a2\_, a1\_, a3\_, n\_] \rightarrow -\text{Z}[a3, n] \text{Z}[a1, a2, n] + \text{Z}[a2, n] \text{Z}[a1, a3, n] - \text{Z}[\text{SP}[a1, a2], a3, n] + \text{Z}[\text{SP}[a1, a3], a2, n] + \text{Z}[a3, a1, a2, n] /; 
     \text{MyMemberQ}[{{a1, a2, a3}}, {a1, a2, a3}] \&\& a1 \neq a2 \neq a3, \\
   \text{Z}[a2\_, a3\_, a1\_, n\_] \rightarrow -\text{Z}[a2, n] \text{Z}[a1, a3, n] + \text{Z}[a1, \text{SP}[a2, a3], n] + \text{Z}[a1, n] \text{Z}[a2, a3, n] - \text{Z}[a2, \text{SP}[a1, a3], n] + \text{Z}[a1, n] \text{Z}[a3, a2, n] - 
      \text{Z}[a3, \text{SP}[a1, a2], n] - \text{Z}[\text{SP}[a1, a3], a2, n] - \text{Z}[a3, a1, a2, n] - \text{Z}[a3, a2, a1, n] /; 
     \text{MyMemberQ}[{{a1, a2, a3}}, {a1, a2, a3}] \&\& a1 \neq a2 \neq a3\}}\\
\end{fmma}     
We can generalize the harmonic sums to S-sums \cite{Moch2002}. For $a_i\in \Z$ and $x_i\in \R$ we define
\begin{equation}
	S_{a_1,\ldots ,a_k;x_1,\ldots ,x_k}(n)= \sum_{n\geq i_1 \geq i_2 \geq \cdots \geq i_k \geq 1} \frac{x_1^{i_1}}{i_1^{a_1}}\cdots
	\frac{x_k^{i_k}}{i_k^{a_k}}.
\end{equation}
The product of two S-sums yields:
\begin{eqnarray}
	S_{a_1,\ldots ,a_k;x_1,\ldots ,x_k}(n)*S_{b_1,\ldots ,b_l;y_1,\ldots ,y_l}(n)&=&
	\sum_{i=1}^n \frac{a_1^i}{i^{a_1}}S_{a_2,\ldots ,a_k;x_2,\ldots ,x_k}(i)*S_{b_1,\ldots ,b_l;y_1,\ldots ,y_l}(i) \nonumber\\
	&+&\sum_{i=1}^n \frac{b_1^i}{i^{\abs {b_1}}}S_{a_1,\ldots ,a_k;x_1,\ldots ,x_k}(i)*S_{b_2,\ldots ,b_l;y_2,\ldots ,y_l}(i) \nonumber\\
	&-&\sum_{i=1}^n \frac{(a_1*b_1)^i}{i^{a_1+b_1}}S_{a_2,\ldots ,a_k;x_2,\ldots ,x_k}(i)*S_{b_2,\ldots ,b_l;y_2,\ldots ,y_l}(i).\nonumber\\
	\label{ssumproduct}
\end{eqnarray}
Again recursive application of (\ref{ssumproduct}) leads to a linear combination of single S-sums.
As showed in \cite{Moch2002}, ideas from the previous sections can be carried over as follows.
We consider an alphabet $A$, where pairs $(m,x)$ form the letters. Moreover, we define the addition and the multiplication of words by
$$[(m_1,x_1),(m_2,x_2)]=(m_1+m_2,x_1*x_2)$$
and by
\begin{eqnarray}
(m_1,x_1)w_1*(m_2,x_2)w_2 &=& (m_1,x_1)(w_1*(m_2,x_2)w_2)+(m_2,x_2)(w_1*(m_2,x_2)w_2)\nonumber\\
&&-[(m_1,x_1),(m_2,x_2)](w_1*w_2).
\nonumber
\end{eqnarray}
Note that this multiplication is identical to the multiplication of S-sums (\ref{ssumproduct}). The degree of a letter is $\abs{(m,x)}=m$. Since (S0)-(S3) hold for $[\cdot,\cdot],$ the S-sums form a quasi-shuffle algebra.
Finally, we emphasize that we implemented algorithms to find algebraic relations between S-sums following the ideas presented in Section \ref{deriverel}.
\begin{fmma}
\In {\text{\bf GetDependentGeneralizedSSums[{2, 1}]}\\
		\text{sums: \ 3}\\
		\text{dependent sums: \ 2}}\\
\Out {\{\text{S}[a1\_, a1\_, a2\_, {x1\_, x1\_, x2\_}, 
      n\_] \rightarrow \frac{1}{2}\text{S}[a1, {x1}, n] \text{S}[a1, a2, {x1, x2}, 
            n] + \frac{1}{2} \text{S}[a1, a1 + a2, {x1, x1 x2}, n] + \frac{1}{2} \text{S}[
            2 a1, a2, {x1^2, x2}, n] - \frac{1}{2}\text{S}[a1, {x1}, n] \text{S}[a2, 
            a1, {x2, x1}, n] - \frac{1}{2} \text{S}[a2, 2 a1, {x2, x1^2}, 
            n] - \frac{1}{2} \text{S}[a1 + a2, a1, {x1 x2, x1}, n] + 
        \text{S}[a2, a1, a1, {x2, x1, x1}, n],\\ 
  \text{S}[a1\_, a2\_, a1\_, {x1\_, x2\_, x1\_}, n\_] \rightarrow 
    \text{S}[a1, {x1}, n] \text{S}[a2, a1, {x2, x1}, n] + \text{S}[a2, 2 a1, {x2, x1^2}, n] + 
      \text{S}[a1 + a2, a1, {x1 x2, x1}, n] - 2 \text{S}[a2, a1, a1, {x2, x1, x1}, n]\}}\\
\end{fmma} 
In a similar way we can generalize Euler-Zagier sums to Z-sums \cite{Moch2002}:
\begin{equation}
	Z_{a_1,\ldots ,a_k;x_1,\ldots ,x_k}(n)= \sum_{n\geq i_1 > i_2 > \cdots > i_k > 0} \frac{x_1^{i_1}}{i_1^{a_1}}\cdots
	\frac{x_k^{i_k}}{i_k^{a_k}}.
\end{equation}
The product of two Z-sums yields:
\begin{eqnarray}
	Z_{a_1,\ldots ,a_k;x_1,\ldots ,x_k}(n)*Z_{b_1,\ldots ,b_l;y_1,\ldots ,y_l}(n)&=&
	\sum_{i=1}^n \frac{a_1^i}{i^{a_1}}Z_{a_2,\ldots ,a_k;x_2,\ldots ,x_k}(i)*Z_{b_1,\ldots ,b_l;y_1,\ldots ,y_l}(i-1) \nonumber\\
	&+&\sum_{i=1}^n \frac{b_1^i}{i^{\abs {b_1}}}Z_{a_1,\ldots ,a_k;x_1,\ldots ,x_k}(i)*Z_{b_2,\ldots ,b_l;y_2,\ldots ,y_l}(i-1) \nonumber\\
	&+&\sum_{i=1}^n \frac{(a_1*b_1)^i}{i^{a_1+b_1}}Z_{a_2,\ldots ,a_k;x_2,\ldots ,x_k}(i)*Z_{b_2,\ldots ,b_l;y_2,\ldots ,y_l}(i-1).\nonumber
\end{eqnarray}
Similar as for harmonic sums and Euler-Zagier sums we can convert between S-sums and Z-sums \cite{Moch2002}.
\begin{fmma}
\begin{mma}
\In \text{\bf ZToS[Z[1, 3, 4, \{2, 1, 1\}, n]]}\\
\Out {\text{S}[8, \{2\}, n] - \text{S}[1, 7, \{2, 1\}, n] - \text{S}[4, 4, \{2, 1\}, n] + 
  \text{S}[1, 3, 4, \{2, 1, 1\}, n]}\\
\In \text{\bf SToZ[S[1, 3, 4, \{2, 1, 1\}, n]]}\\
\end{mma}
\begin{mma}
\Out {\text{Z}[8, \{2\}, n] + \text{Z}[1, 7, \{2, 1\}, n] + \text{Z}[4, 4, \{2, 1\}, n] + 
  \text{Z}[1, 3, 4, \{2, 1, 1\}, n]}\\
\end{mma}
\end{fmma}

%% file: harmonicpolylogs.tex
\chapter{Harmonic Polylogarithms and the Differentiation of Harmonic Sums }
\label{Harmonic Polylogarithms}

The bigger part of this chapter will deal with \itshape harmonic polylogarithms. \upshape Harmonic polylogarithms were first introduced in \cite{Remiddi2000} and are covered by Poincar\'{e} iterated integrals \cite{Poincare1884,Lappo-Danielevsky1953}. For the definition and basic properties of the harmonic polylogarithms we will follow more or less \cite{Remiddi2000}, however sometimes we will go more into detail (for example we will provide detailed proofs), and we will skip aspects which are not important for our considerations.\\
We discuss harmonic polylogarithms since it will turn out that they are connected to multiple harmonic sums via a special extension of the Mellin transform, which is frequently used in particle physics. The Mellin transforms of harmonic polylogarithms are expressions in multiple harmonic sums, while the inverse Mellin transforms of multiple harmonic sums are expressions in harmonic polylogarithms. Hence we can use the Mellin transform to convert between multiple harmonic sums and harmonic polylogarithms. As a direct consequence we can use the inverse Mellin transform to construct analytic continuations of multiple harmonic sums. An algorithm to compute the Mellin transform of harmonic polylogarithms is given in \cite{Remiddi2000}; we will analyze this algorithm in detail and will verify its correctness. Similarly, an algorithm for the computation of the inverse Mellin transform of multiple harmonic sums is given in \cite{Remiddi2000}; again we will work out the correctness of this algorithm and give additional insight from a more computer algebra point of view. In order to execute these algorithms, we will also need to consider harmonic sums at infinity and harmonic polylogarithms at one. Both do not have to be finite, but it will turn out that they are connected via power series expansions.\\ 
These algorithms are originally implemented in Vermaseren's package \ttfamily harmpol \rmfamily in the computer algebra system form \cite{Remiddi2000}. Harmpol also provides procedures to manipulate with harmonic polylogarithms. In addition, there is the \sffamily Mathematica\rmfamily-package \ttfamily hpl \rmfamily \cite{Maitre2006}, in which basic procedures for manipulating harmonic polylogarithms are available.\\ 
Since we can use the inverse Mellin transform to compute analytic continuations of multiple harmonic sums, we can define a differentiation on multiple harmonic sums as done in \cite{Bluemlein2008,Bluemlein2009,Bluemlein2009a}.\\
In the last part of this chapter we give an algorithm to do this differentiation automatically. In the package \ttfamily HarmonicSums \rmfamily the Mellin-transform and the inverse Mellin-transform together with procedures for harmonic polylogarithms which are also needed are implemented. In addition, the differentiation of multiple harmonic sums is implemented. Along with this differentiation new relations between multiple harmonic sums are going to emerge, and these relations will reduce the algebraic basis computed in the previous chapter.
\begin{figure}
\centering
\includegraphics[width=0.9\textwidth]{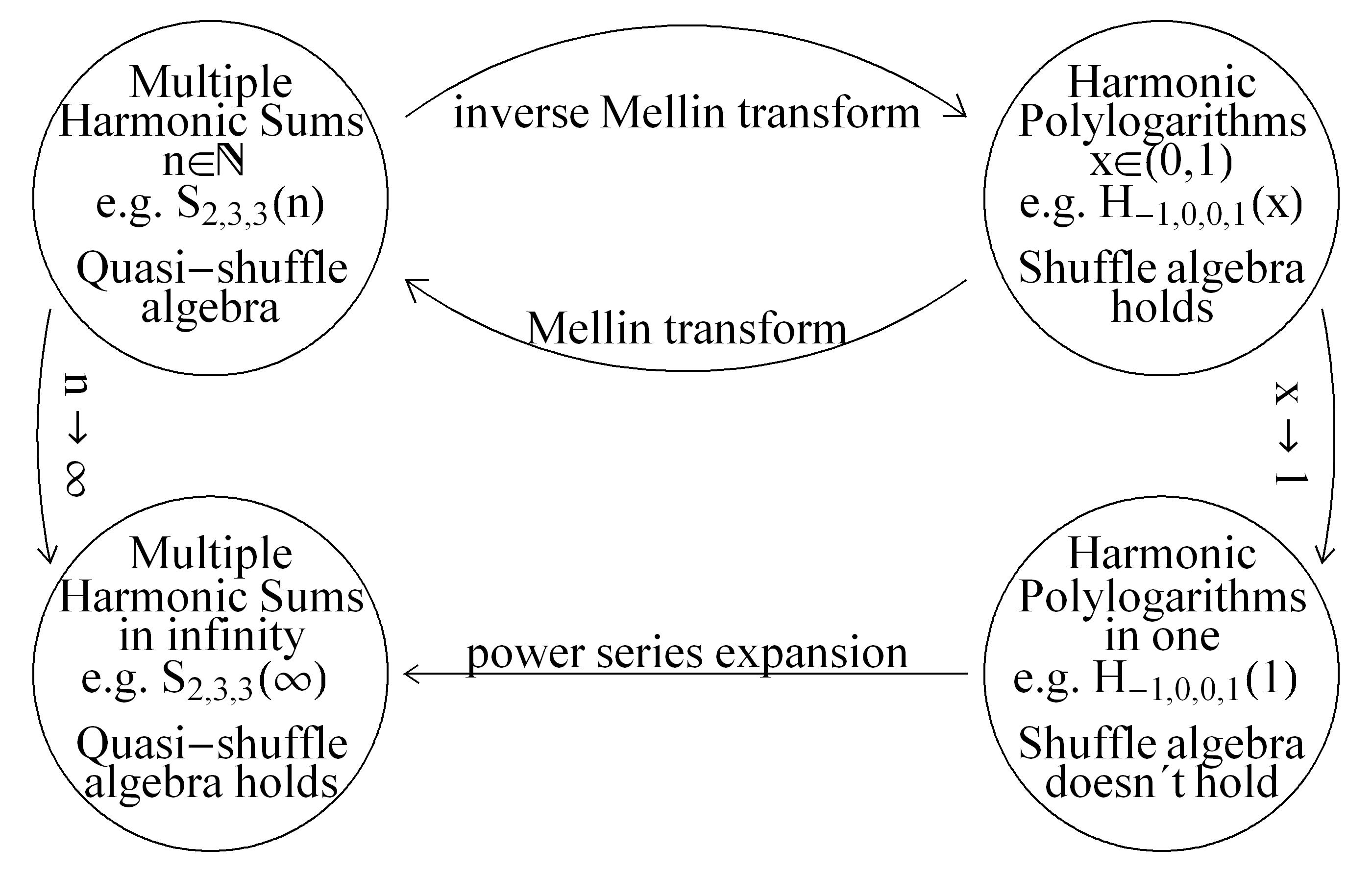}
\caption[height=5cm]{Connection between harmonic sums and harmonic polylogarithms.}
\label{connection}
\end{figure}

\section{Definition of Harmonic Polylogarithms}
In order to define harmonic polylogarithms as introduced in \cite{Remiddi2000}, we first define the function $f:\left\{0,1,-1\right\}\times(0,1)\mapsto \R$ by
\begin{eqnarray} 
f(0,x)&=&\frac{1}{x},\nonumber\\
f(1,x)&=&\frac{1}{1-x},\nonumber\\
f(-1,x)&=&\frac{1}{1+x}.\nonumber
\end{eqnarray}

Harmonic polylogarithms are functions in $C^\infty((0,1))$ and are defined inductively:
Let $\ve m=(m_w,m_{w-1},\ldots , m_1)$, $m_i \in \left\{0,1,-1\right\}$ be a vector of length $w \in \N$ and let $\ve 0_w$ be the vector of length $w$ whose $w$ components are all equal $0$. For $x\in (0,1)$ we define the harmonic polylogarithm $\HH{\ve m,x}:$
\begin{eqnarray}
\HH{(),x}&=&1,\nonumber\\
\HH{\ve 0_w,x} &=& \frac{1}{w!}(\log{x})^w, \nonumber\\
\HH{\ve m,x} &=& \int_0^x{f(m_w,y) \HH{(m_{w-1},m_{w-2},\ldots,m_1),y}dy} \ \textnormal{if} \ \ve m \neq \ve 0_w. \nonumber
\label{abpoly1}
\end{eqnarray}
The length of the vector $\ve m$ is called the weight of the polylogarithm $\HH{\ve m,x}.$\\
Examples are
\begin{eqnarray}
\HH{(1),x} &=& \int_0^x{\frac{dy}{1-y}}=-\log{(1-x)}, \nonumber\\
\HH{(-1),x}&=& \int_0^x{\frac{dy}{1+y}}=\log{(1+x)}, \nonumber\\
\HH{(-1,1),x}&=& \int_0^x{\frac{\H{1}{y}}{1+y}dy}= \int_0^x{\frac{\log{(1-y)}}{1+y}dy}. \nonumber
\end{eqnarray}

\begin{remark}(compare  \cite{Remiddi2000})
It follows from the definition that if $\ve m\neq \ve 0_w$, $\HH{\ve m,0}=0.$
If $m_w\neq 1$ or if $m_w=1$ and $m_v=0$ for all $v$ with $v<w$ then $\HH{\ve m,1}$ is finite.
In the remaining, cases, i.e., $m_w=1$ and $m_v\neq0$ for some $v$ with $v<w$, $\lim_{x\rightarrow 1^-} \HH{\ve m,x}$ behaves as a combination of powers of $\log(1-x)$ as we will see later in more detail.
We define $\HH{\ve m,0}:=\lim_{x\rightarrow 0^+} \HH{\ve m,x}$ and $\HH{\ve m,1}:=\lim_{x\rightarrow 1^-} \HH{\ve m,x}$ if the limits exist.
\label{finitness}
\end{remark}

\begin{notation}
From now on we will write $\H{m_w,m_{w-1},\ldots , m_1}{x}$ for $\HH{(m_w,m_{w-1},\ldots , m_1),x}.$
\label{abpoly2}
\end{notation}

\begin{remark}
For the derivatives we have for all $x\in (0,1)$ that $$ \frac{d}{d x} \H{\ve m}{x}=f(m_w,x)\H{m_{w-1},m_{w-2},\ldots,m_1}{x}. $$ 
\end{remark}

Subsequently, we present some identities between harmonic polylogarithms of the same argument. These identities will lead to a formula, which we can use to rewrite the product of harmonic polylogarithms to a sum of single harmonic polylogarithms.
\begin{lemma}
For $x\in (0,1)$, $q\in \N$ and $m_i \in \left\{0,1,-1\right\},$
\begin{eqnarray}
\H{m_1,\ldots,m_q}{x}&=&\H{m_1}{x}\H{m_2,\ldots,m_q}{x}\nonumber\\
&-&\H{m_2,m_1}{x}\H{m_3,\ldots,m_q}{x} \nonumber\\ 
&+&\H{m_3,m_2,m_1}{x}\H{m_4,\ldots,m_q}{x}\nonumber\\
&-&\cdots -(-1)^q\H{m_q,\ldots,m_1}{x}.
\label{intbyparts}
\end{eqnarray}
\end{lemma}
\begin{proof}
Case 1, $(m_1,\ldots,m_q)=\ve{0}_q$:\\
From the definition we get
$$\H{m_1,\ldots,m_q}{x}=\frac{1}{q!}\log{x}^q$$ and the right hand side of (\ref{intbyparts}) yields
\begin{eqnarray}
&&\log{x}\frac{1}{(q-1)!}\log^{q-1}{x}-\frac{1}{2!}\log^{2}{x}\frac{1}{(q-2)!}\log^{q-2}{x}\nonumber\\
&&+\frac{1}{3!}\log^{3}{x}\frac{1}{(q-3)!}\log^{q-3}{x}-\cdots-(-1)^q\frac{1}{q!}\log^{q}{x}\nonumber\\
&&=\log^{q}{x}\sum_{i=1}^{q}{\frac{(-1)^i}{i!(q-i)!}}=\log{x}^q\frac{1}{q!}.\nonumber
\end{eqnarray}
Case 2, $(m_1,\ldots,m_q) \neq \ve{0}_q$:\\
Using the definition and partial integration we get
\begin{eqnarray}
\H{m_1,\ldots,m_q}{x}&=&\int_0^x f(m_1,y)\H{m_2,\ldots,m_q}{y}dy\nonumber\\
&=&\H{m_1}{x}\H{m_2,\ldots,m_q}{x}-\int_0^x f(m_2,y)\H{m_3,\ldots,m_q}{y}dy.\nonumber
\end{eqnarray}
Continuing by partial integration we finally get the right hand side.
\end{proof}
We can use this lemma to get the following identity.
\begin{lemma}
For $x\in (0,1),$ $a \in \left\{0,1,-1\right\}$, $q\in \N$ and $m_i \in \left\{0,1,-1\right\}$,
\begin{eqnarray}
\H{a}{x}\H{m_q,\ldots,m_1}{x}&=&\H{a,m_q,m_{q-1},\ldots,m_1}{x}\nonumber\\
&+&\H{m_q,a,m_{q-1},\ldots,m_1}{x}\nonumber\\
&+&\H{m_q,m_{q-1},a,m_{q-2},\ldots,m_1}{x}\nonumber\\
&+&\cdots \nonumber\\
&+&\H{m_q,m_{q-1},\ldots,m_1,a}{x}.
\label{hpro1}
\end{eqnarray}
\end{lemma}

\begin{proof}
Let $q=1$: By (\ref{intbyparts}) we get $\H{a}{x}\H{m_1}{x}=\H{a,m_1}{x}.$
Assume (\ref{hpro1}) holds for $q:$
\begin{eqnarray}
\H{a}{x}\H{m_q,\ldots,m_1}{x}&=&\H{a,m_q,m_{q-1},\ldots,m_1}{x}\nonumber\\
&+&\H{m_q,a,m_{q-1},\ldots,m_1}{x}\nonumber\\
&+&\H{m_q,m_{q-1},a,m_{q-2},\ldots,m_1}{x}\nonumber\\
&+&\cdots \nonumber\\
&+&\H{m_q,m_{q-1},\ldots,m_1,a}{x}.\nonumber
\end{eqnarray}
Multiply both sides by $f(m_{q+1},y)$ and integrate: 
\begin{eqnarray}
&&\int_0^x{f(m_{q+1},y)\H{a}{y}\H{m_q,\ldots,m_1}{y}dy}\nonumber\\
&\ \ &=\int_0^x{f(m_{q+1},y)\H{a,m_q,m_{q-1},\ldots,m_1}{y}dy}\nonumber\\
&\ \ &+\int_0^x{f(m_{q+1},y)\H{m_q,a,m_{q-1},\ldots,m_1}{y}dy}\nonumber\\
&\ \ &+\int_0^x{f(m_{q+1},y)\H{m_q,m_{q-1},a,m_{q-2},\ldots,m_1}{y}dy}\nonumber\\
&\ \ &+\cdots + \int_0^x{f(m_{q+1},y)\H{m_q,m_{q-1},\ldots,m_1,a}{y}dy}.\nonumber
\end{eqnarray}
By partial integration on the left hand side and the definition of harmonic polylogarithms we get:
\begin{eqnarray}
&&\H{a}{x}\H{m_q,\ldots,m_1}{x}-\int_0^x{f(a,y)\H{m_{q+1},\ldots m_1}{y}}dy\nonumber\\
& \ \ &=\H{m_{q+1},a,m_q,\ldots,m_1}{x}+\H{m_{q+1},m_q,a,m_{q-1},\ldots,m_1}{x}\nonumber\\
& \ \ &+\H{m_{q+1},m_q,m_{q-1},a,m_{q-2},\ldots,m_1}{x}+\cdots+\H{m_{q+1},\ldots,m_1,a}{x}.\nonumber
\end{eqnarray}
Hence (\ref{hpro1}) holds for $q+1.$
\end{proof}

Similar to multiple harmonic sums, the shuffle product on words is extended to harmonic polylogarithms.
\begin{definition}Let $\ve{a}, \ve{b} \in A^*$. If
$$
\ve{a}\shuffle \ve{b}=c_1\ve{d}_1+\cdots+c_n\ve{d}_n
$$
for $\ve{d}_i\in A^*, c_i\in\R$ then we define
$$
\H{\ve{a}\shuffle\ve{b}}{x}:=c_1\H{\ve{d}_1}{x}+\cdots+c_n\H{\ve{d}_n}{x}.
$$
\label{polyshuff}
\end{definition}

\begin{definition}
Let $x \in (0,1)$, $a=(a_1, a_2,\ldots,a_p)$ and $b=(b_1, b_2,\ldots,b_q)$, where $a_i, b_i \in \left\{0,1,-1\right\}$, $p,q \in \N$. Let us interpret $a$ and $b$ as words, so $a=a_1 a_2\cdots a_p$ and $b=b_1 b_2 \cdots b_q$. We define the shuffle product of two harmonic polylogarithms by
$$
\H{a}{x}\shuffle \H{b}{x}:=\H{a_1(a_2 \cdots a_p \shuffle b)}{x}+\H{b_1(a \shuffle b_2\cdots b_q)}{x}
$$
where $\H{c_1 c_2 \cdots c_n}{x} := \H{c_1, c_2, \ldots, c_n}{x}$ for $c_i \in \left\{0,1,-1\right\}$ and where $\shuffle$ is defined as in Definition \ref{shuffdef}.  
\end{definition}

\begin{example} For $x \in (0,1)$,
\begin{eqnarray}
\H{1,0,-1}{x} \shuffle \H{0,1}{x}&=&\H{0, 1, 0, -1, 1}{x} + \H{0, 1, 0, 1, -1}{x} \nonumber\\
&+& 2 \H{0, 1, 1, 0, -1}{x} + \H{1, 0, -1, 0, 1}{x} \nonumber\\
&+& 2 \H{1, 0, 0, -1, 1}{x} + 2 \H{1, 0, 0, 1, -1, x}{x} \nonumber\\
&+& \H{1, 0, 1, 0, -1}{x}.\nonumber
\end{eqnarray}
\end{example}

\ttfamily HarmonicSums \rmfamily provides a procedure to carry out this product:
\begin{fmma}
\In \text{\bf ProductH[H[1,0,-1,x],H[0,1,x]]}\\
\Out {\text{H[0, 1, 0, -1, 1, x] + H[0, 1, 0, 1, -1, x] + 2 H[0, 1, 1, 0, -1, x] + H[1, 0, -1, 0, 1, x]}\\
			 \text{+ 2 H[1, 0, 0, -1, 1, x] + 2 H[1, 0, 0, 1, -1, x] + H[1, 0, 1, 0, -1, x]}}\\
\end{fmma}

\begin{thm}
Let $x \in (0,1)$, $a=(a_p, a_{p-1},\ldots,a_1)$ and $b=(b_q, b_{q-1},\ldots,b_1)$, where $a_i, b_i \in \left\{0,1,-1\right\}$, $p,q \in \N.$ For the product of two harmonic polylogarithms we get
\begin{equation}
\H{a}{x}\H{b}{x}=\H{a}{x}\shuffle \H{b}{x}. \label{hpro}
\end{equation}
\end{thm}

\begin{proof}
We already showed the case for $p=1$.
Now assume (\ref{hpro}) holds for $p>1:$
\begin{eqnarray}
\H{a_p, a_{p-1},\ldots,a_1}{x}\H{b_q, b_{q-1},\ldots,b_1}{x}=\H{a_p, a_{p-1},\ldots,a_1}{x}\shuffle \H{b_q, b_{q-1},\ldots,b_1}{x}.\nonumber
\end{eqnarray}
Multiply both sides by $f(a_{p+1},x)$ and integrate:
\begin{eqnarray}
&&\int_0^x f(a_{p+1},y) \H{a_p, a_{p-1},\ldots,a_1}{y}\H{b_q, b_{q-1},\ldots,b_1}{y}dy \nonumber \\
& \ \ & =\int_0^x f(a_{p+1},y)\left(\H{a_p, a_{p-1},\ldots,a_1}{y}\shuffle \H{b_q, b_{q-1},\ldots,b_1}{y}\right)dy.\nonumber
\end{eqnarray}
Partial integration of the left hand side gives
\begin{eqnarray}
&&\H{a_{p+1}, a_p,\ldots,a_1}{x}\H{b_q, b_{q-1},\ldots,b_1}{x}\nonumber \\
&&-\int_0^x{f(b_q,y)\left(\H{a_{p+1}, a_p,\ldots,a_1}{y}\shuffle \H{b_{q-1}, b_{q-2},\ldots,b_1}{y}\right)dy}.\nonumber
\end{eqnarray}
Let us abbreviate $\bar{a}=a_{p+1}a_p\cdots a_1,$ hence we may write 
$$\H{\bar{a}}{x} \shuffle \H{b}{x}=\H{a_{p+1}, a_p,\ldots,a_1}{x} \shuffle \H{b_q, b_{q-1},\ldots,b_1}{x}.$$
Expanding the first shuffle product, using the definition of the harmonic polylogarithms and moving it to the right hand side gives a sum of all those harmonic polylogarithms which appear in $\H{\bar{a}}{x} \shuffle \H{b}{x}$ and start with $b_{q}$, while expanding the second shuffle product gives a sum of all those harmonic polylogarithms which appear in $\H{\bar{a}}{x} \shuffle \H{b}{x}$ and start with $a_{p+1}.$ Since all harmonic polylogarithms in $\H{\bar{a}}{x} \shuffle \H{b}{x}$ either start with $b_q$ or $a_{p+1}$ we completed the proof. 
\end{proof}
The following remarks are in place.
\begin{remark}
For harmonic polylogarithms which are finite at $x=0$ and $x=1,$ (\ref{hpro}) can be extended to $x\in [0,1]$; however if one of the harmonic polylogarithms is not finite, (\ref{hpro}) does not hold (compare \cite{Remiddi2000}).
\label{hproexp}
\end{remark}
\begin{remark}
The product of two harmonic polylogarithms of weights $w_1$ and $w_2$ can be expressed as a linear combination of $(w_1+w_2)!/(w_1!w_2!)$ polylogarithms of weight $w=w_1+w_2$.
\end{remark}
\begin{remark}
We say that a harmonic polylogarithm $\H{a_1, a_2,\ldots,a_p}{x}$ has trailing zeros if there is an $i \in \left\{1,2,\ldots,p\right\}$ such that for all $j \in \N$ with $i\leq j \leq p$,  $a_j=0.$ \\
Likewise, we say that a harmonic polylogarithm $\H{a_1, a_2,\ldots,a_p}{x}$ has leading ones if there is an $i \in \left\{1,2,\ldots,p\right\}$ such that for all $j \in \N$ with $1 \leq j \leq i$,  $a_j=1.$  
\end{remark}
We can use (\ref{hpro1}) now to single out terms of $\log{x}$ from harmonic polylogarithms whose indices have trailing zeros. Let $a=0$ in (\ref{hpro1}); then after using the definition of harmonic polylogarithms we get:
\begin{eqnarray}
&&\log(x)\H{m_q,\ldots,m_1}{x}=\H{0,m_q,\ldots,m_1}{x} + \H{m_q,0,m_{q-1},\ldots,m_1}{x}\nonumber\\
&&+\H{m_q,m_{q-1},0,m_{q-2},\ldots,m_1}{x} + \cdots + \H{m_q,\ldots,m_1,0}{x},\nonumber
\end{eqnarray} 
or equivalently:
\begin{eqnarray}
&&\H{m_q,\ldots,m_1,0}{x} = \log(x)*\H{m_q,\ldots,m_1}{x} - \H{0,m_q,\ldots,m_1}{x}\nonumber\\ 
&&-\H{m_q,0,m_{q-1},\ldots,m_1}{x} -\cdots- \H{m_q,\ldots,0,m_1}{x}.\nonumber
\end{eqnarray} 
If $m_1$ is $0$ as well, we can move the last term to the left and can divide by two. This leads to
\begin{eqnarray}
&&\H{m_q,\ldots,0,0}{x} = \frac{1}{2}\left(\log(x)*\H{m_q,\ldots,m_2,0}{x} - \H{0,m_q,\ldots,m_2,0}{x}\right.\nonumber\\ 
&&\left.-\H{m_q,0,m_{q-1},\ldots,m_2,0}{x} -\cdots- \H{m_q,\ldots,0,m_2}{x}\right).\nonumber
\end{eqnarray} 
Now we can use (\ref{hpro1}) for all the other terms, and we get an identity which extracts the logarithmic singularities due to two trailing zeros. We can repeat this strategy as often as needed in order to extract all the powers of $\log(x)$ or equivalently $\H{0}{x}$ from a harmonic polylogarithm.
\begin{example} For $x \in (0,1)$,
\begin{eqnarray}
\H{1, -1, 0, 0}{x}&=&\frac{1}{2} \H{0}{x}^2 \H{1, -1}{x} - \H{0}{x} \H{0, 1, -1}{x} \nonumber\\
										&-& \H{0}{x} \H{1, 0, -1}{x} +\H{0, 0, 1, -1}{x} + \H{0, 1, 0, -1}{x}\nonumber\\
										&+& \H{1, 0, 0, -1}{x}\nonumber\\
										&=&\frac{1}{2} \log^2(x) \H{1, -1}{x} - \log(x) \H{0, 1, -1}{x} \nonumber\\
										&-& \log(x) \H{1, 0, -1}{x} +\H{0, 0, 1, -1}{x} + \H{0, 1, 0, -1}{x}\nonumber\\
										&+& \H{1, 0, 0, -1}{x}.\nonumber
\end{eqnarray}
\end{example}
This can be done with \ttfamily HarmonicSums \rmfamily as follows.
\begin{fmma}
\In \text{\bf RemoveTrailing0[H[1, -1, 0, 0, x]]}\\
\Out {\frac{{\text{H}[0,x]}^2\,\text{H}[1,-1,x]}{2} - \text{H}[0,x]\,\text{H}[0,1,-1,x] - \text{H}[0,x]\,\text{H}[1,0,-1,x] + \text{H}[0,0,1,-1,x] + \text{H}[0,1,0,-1,x] + \text{H}[1,0,0,-1,x]}\\
\end{fmma}

\begin{remark}
In our implementation we prefer to use \textnormal{H[0,x]} instead of \textnormal{Log[x]}.
\end{remark}

\begin{remark}
We can decompose a harmonic polylogarithm $\H{m_q, m_{q-1},\ldots,m_1}{x}$ in a univariate polynomial in $\H{0}{x}$ with coefficients in the harmonic polylogarithms without trailing zeros. If the harmonic polylogarithm has exactly $r$ trailing zeros, the highest power of $\H{0}{x},$ which will appear, is $r.$
\label{remarktrailing}
\end{remark}

Similarly, we can use (\ref{hpro1}) to extract powers of $\log(1-x)$, or equivalently $\H{1}{x}$ from harmonic polylogarithms whose indices have leading ones and hence are singular around $x=1$. Let $a=1$ in (\ref{hpro1}); then after using the definition of harmonic polylogarithms we get:
\begin{eqnarray}
&&\log(1-x)*\H{m_q,\ldots,m_1}{x}=\H{1,m_q,\ldots,m_1}{x} + \H{m_q,1,m_{q-1},\ldots,m_1}{x}\nonumber\\
&&+\H{m_q,m_{q-1},1,m_{q-2},\ldots,m_1}{x} + \cdots + \H{m_q,\ldots,m_1,1}{x},\nonumber
\end{eqnarray} 
or equivalently:
\begin{eqnarray}
&&\H{1,m_q,\ldots,m_1}{x} = \log{(1-x)}*\H{m_q,\ldots,m_1}{x} - \H{m_q,1,m_{q-1},\ldots,m_1}{x}\nonumber\\ 
&&-\H{m_q,m_{q-1},1,\ldots,m_1}{x} -\cdots- \H{m_q,\ldots,m_1,1}{x}.
\label{leadingone}
\end{eqnarray}

If $m_q$ is $1$ as well, we can move the second term to the left, and divide by two. This leads to
\begin{eqnarray}
&&\H{1,1,\ldots,m_1}{x} = \frac{1}{2}\left(\log{(1-x)}*\H{1,m_{q-1},\ldots,m_1}{x} - \H{1,m_{q-1},1,\ldots,m_1}{x}\right.\nonumber\\ 
&&\left.-\cdots- \H{1,m_{q-1},\ldots,m_1,1}{x}\right).\nonumber
\end{eqnarray} 
Now we can use (\ref{hpro1}) for all the other terms and we get an identity which extracts the singularities due to two leading ones. We can repeat this strategy as often as needed in order to extract all the powers of $\log(x-1)$ or equivalently $\H{1}{x}$ from a harmonic polylogarithm.

\begin{example} For $x \in (0,1)$,
\begin{eqnarray}
\H{1, 1, 0, -1}{x}&=&-\H{0, -1}{x} \H{1, 1}{x} - \H{1}{x} \H{0, -1, 1}{x}\nonumber\\
										&-& \H{1}{x} \H{0, 1, -1}{x} + \H{0, -1, 1, 1}{x} \nonumber\\
										&+& \H{0, 1, -1, 1}{x} + \H{0, 1, 1, -1}{x} \nonumber\\
										&=&-\frac{1}{2}*\H{0, -1}{x} (-\log(x-1))^2 + \log(x-1) \H{0, -1, 1}{x}\nonumber\\
										&+& \log(x-1) \H{0, 1, -1}{x} + \H{0, -1, 1, 1}{x} \nonumber\\
										&+& \H{0, 1, -1, 1}{x} + \H{0, 1, 1, -1}{x}. \nonumber
\end{eqnarray}
\end{example}
\begin{fmma}
\In \text{\bf RemoveLeading1[H[1, 1, 0, -1, x]]}\\
\Out {\text{H}[0, -1, x] \, \text{H}[1, 1, x] - \text{H}[1, x] \, \text{H}[0, -1, 1, x] - \text{H}[1, x]\, \text{H}[0, 1, -1, x] + 
  \text{H}[0, -1, 1, 1, x] + \text{H}[0, 1, -1, 1, x] + \text{H}[0, 1, 1, -1, x]}\\
\end{fmma} 

\begin{remark}
As before we prefer to use \textnormal{H[1,x]} instead of \textnormal{Log[1-x]} in our implementation.
\end{remark}

\begin{remark}
We can decompose a harmonic polylogarithm $\H{m_q, m_{q-1},\ldots,m_1}{x}$ in a univariate polynomial in $\H{1}{x}$ with coefficients in the harmonic polylogarithms without leading ones. If the harmonic polylogarithm has exactly $r$ leading ones, the highest power of $\H{1}{x},$ which will appear, is $r.$ So all divergences of harmonic polylogarithms for $x\rightarrow 1$ can be traced back to the basic divergence of $\H{1}{x}=-\log(1-x)$ for $x\rightarrow 1$.
\end{remark}

\begin{remark}
We can combine these two strategies to extract both, leading ones and trialling zeros. Hence, it is always possible to express a harmonic polylogarithm in a bivariate polynomial in $\H{0}{x}=\log(x)$ and $\H{1}{x}=\log(1-x)$ with coefficients in the harmonic polylogarithms without leading ones or trailing zeros, which are continuous on $[0,1]$ and finite at $x=0$ and $x=1.$ 
\end{remark}

\begin{example} For $x \in (0,1)$,
\begin{eqnarray}
\H{1, -1, 0}{x}&=&-\H{0}{x} \H{-1, 1}{x} + \H{1}{x} (\H{-1}{x} \H{0}{x} \nonumber\\
								 &-& \H{0, -1}{x}) + \H{0, -1, 1}{x} \nonumber\\
								 &=&-\log(x) \H{-1, 1}{x} + \log(1-x) (\H{-1}{x} \log(x) \nonumber\\
								 &-& \H{0, -1}{x}) + \H{0, -1, 1}{x} \nonumber
\end{eqnarray}
\end{example}
With \ttfamily HarmonicSums \rmfamily this can be handled as follows.
\begin{fmma}
\In \text{\bf RemoveLeading1[RemoveTrailing0[H[1, -1, 0, x]]]}\\
\Out {\text{H}[0, x] (\text{H}[-1, x]\,\text{H}[1, x] - \text{H}[-1, 1, x]) - \text{H}[1, x]\,\text{H}[0, -1, x] + 
  \text{H}[0, -1, 1, x]}\\
\end{fmma}


\section{Multiple Harmonic Sums at Infinity}
Later we will use the fact that the values of harmonic polylogarithms at $x=1$ are related to the values of multiple harmonic sums at infinity. In order to make this more precise, we first take a closer look at harmonic sums at infinity.\\
Of course, not all multiple harmonic sums are finite at infinity, since for example $\lim_{n\rightarrow \infty} S_1(n)$ does not exist.
In fact, we have the following lemma, compare \cite{Minh2000}:
\begin{lemma}
Let $a_1, a_2, \ldots a_p \in \Z/\{0\}$ for $p \in \N.$
The multiple harmonic sum $S_{a_1,a_2,\ldots,a_p}(n)$ is convergent, when $n\rightarrow \infty$, if and only if $a_1 \neq 1.$
\label{consumlem}
\end{lemma} 
  

Similar as for harmonic polylogarithms we say that a harmonic sum $S_{a_1, a_2,\ldots,a_p}(n)$ has leading ones if there is an $i \in \left\{1,2,\ldots,p\right\}$ such that for all $j \in \N$ with $1 \leq j \leq i$, $a_j=1.$ In a similar way as for harmonic polylogarithms we can use (\ref{hsumproduct}) to extract leading ones; for the next example compare \cite[Equ. 3.23]{Bluemlein2004}.

\begin{example} For $n\in\N,$
\begin{eqnarray}
S_{1,1,2}(n)&=&S_{1, 1}(n)S_{2}(n) + \frac 1 2 S_{4}(n) + S_{1}(n) (S_{3}(n) - S_{2, 1}(n)) - S_{3, 1}(n) \nonumber \\
&&+ S_{2, 1, 1}(n) + \frac 1 2 S_{4}(n) - S_{2, 2}(n)\nonumber
\end{eqnarray}  
\end{example}
Again this can be handled with \ttfamily HarmonicSums\rmfamily:
\begin{fmma}
\In \text{\bf SRemoveLeading1[S[1, 1, 2, n]]}\\
\Out {\frac{{\text{S}[1,n]}^2\,\text{S}[2,n]}{2} + \text{S}[1,n]\,\text{S}[3,n] + 
  \frac{\text{S}[4,n]}{2} - \text{S}[1,n]\,\text{S}[2,1,n] - \text{S}[3,1,n] + 
  \text{S}[2,1,1,n]}\\
\end{fmma} 

We will always be able to express a multiple harmonic sum as an expression consisting only of sums without leading ones and sums of the type $S_{\ve 1_w}(n)$ with $w\in \N.$ Since sums without leading ones are convergent, they do not produce any problems. For sums of the type $S_{\ve 1_w}(n)$ we have the following proposition which is a direct consequence of Corollary 3 of \cite{Dilcher1995} and Proposition 2.1 of \cite{Kirschenhofer1996}.

\begin{prop} For $w\geq 1$ and $n\in\N$, we have
$$S_{\ve 1_w}(n)=\sum_{i=1}^n(-1)^{i-1}{n\choose i}\frac{1}{i^w}=-\frac{1}{w}Y_w(\ldots,(j-1)!S_j(n),\ldots),$$
where $Y_w(\ldots,x_i,\ldots)$ are the Bell polynomials.
\end{prop}
Note that the explicit formulas were found empirically in \cite{Bluemlein1999}; for a determinant evaluation formula we refer to \cite{Bluemlein1999,Bruno1881,Berndt1985}.
Using this proposition we can express each sum of the type $S_{\ve 1_w}(n)$ by sums of the form $S_{i}(n),$ where $i\in\Z$. Hence we can decompose each multiple harmonic sum $S_{a_1, a_2,\ldots,a_p}(n)$ in a univariate polynomial in $S_1(n)$ with coefficients in the convergent harmonic sums. So all divergences can be traced back to the basic divergence of $S_1(n)$ when $n\rightarrow\infty$.
\begin{notation}
We will define the symbol $S_1(\infty):=\lim_{n\rightarrow\infty}S_1(n),$ and every time it is present, we are in fact dealing with limit processes. For $k\in N$ with $k>1$ the harmonic sums $S_k(\infty)$ turn into zeta-values and we will sometimes write $\zeta_k$ for $S_k(\infty).$
\end{notation} 
  
\begin{example}
\begin{eqnarray}
\lim_{n\rightarrow\infty}{S_{1,1,2}(n)}&=&\lim_{n\rightarrow\infty}\left\{\frac 1 2 S_{1}(n)^2 S_{2}(n) + S_{1}(n) (S_{3}(n)-S_{2, 1}(n))\right.\nonumber\\
					& &- \left. S_{3, 1}(n) + \frac 1 2 S_{4}(n) + S_{2, 1, 1}(n)\right\}\nonumber \\
					&=& S_1(\infty)^2\frac{\zeta_2}2-S_1(\infty)\zeta_3+\frac{9\zeta_2^2}{10}\nonumber
\end{eqnarray}
\end{example} 
For the computation of the actual values of the harmonic sums at infinity see \cite{Vermaseren1998}.

\section{Values at One and Power Series Expansion of Harmonic Polylogarithms}
Due to trailing zeros in the index, harmonic polylogarithms in general do not have regular Taylor series expansions. To be more precise, the expansion is separated into two parts, one in $x$ and one in $\log(x)$. E.g., it can be seen easily that trailing zeros are responsible for powers of $\log(x)$ in the expansion since \cite{Remiddi2000}:
$$
\H{\ve{0}_k}{x}=\frac{1}{k!}\int_0^x{x^n\log^k(x)}=x^{n+1}\sum_{i=0}^k{\frac{(-1)^{k-i}}{i!}\frac{\log^i(x)}{(m+1)^{k-i+1}}}.
$$
Subsequently, we only consider the case without trailing zeros in more detail. For weight one we get following the well known expansions.
\begin{lemma} For $x\in \left[0\right.,\left.1\right),$
\begin{eqnarray}
\H{1}{x}&=&-\log(1-x)=\sum_{i=1}^\infty{\frac{x^i}{i}},\nonumber\\
\H{-1}{x}&=&\log(1+x)=\sum_{i=1}^\infty{-\frac{(-x)^i}{i}}.\nonumber
\end{eqnarray}
\label{powexp1}
\end{lemma}
For higher weights we proceed inductively. If we assume that $\ve m$ has no trailing zeros and that for $x\in[0,1]$ we have
$$
\H{\ve m}{x}=\sum_{i=1}^\infty{\frac{\sigma^i x^i}{i^a}S_{\ve n}(i)},
$$
then for $x\in \left[0\right.,\left.1\right)$ the following theorem holds.
\begin{thm} For $x\in \left[0\right.,\left.1\right),$
\begin{eqnarray}
\H{0,\ve m}{x}&=&\sum_{i=1}^\infty{\frac{\sigma^ix^i}{i^{a+1}}S_{\ve n}(i)},\nonumber\\
\H{1,\ve m}{x}&=&\sum_{i=1}^\infty{\frac{x^i}{i}S_{\sigma a ,\ve n}(i-1)}=\sum_{i=1}^\infty{\frac{x^i}{i}S_{\sigma a ,\ve n}(i)}-
							\sum_{i=1}^\infty{\frac{\sigma^ix^i}{i^{a+1}}S_\ve n(i)},\nonumber\\
\H{-1,\ve m}{x}&=&-\sum_{i=1}^\infty{\frac{(-1)^ix^i}{i}S_{-\sigma a ,\ve n}(i-1)}=-\sum_{i=1}^\infty{\frac{(-1)^ix^i}{i}S_{-\sigma a ,\ve n}(i)}+
							\sum_{i=1}^\infty{\frac{\sigma^ix^i}{i^{a+1}}S_\ve n(i)}.\nonumber
\end{eqnarray}
\label{powexp}
\end{thm}
In order to prove this theorem, we need the following lemma:
\begin{lemma}
Let $a \in \N,\sigma \in \{-1,1\}$ and $\ve n$ be a vector with integer entries. Then
$$
\lim_{i\rightarrow \infty}{\left|\frac{S_{\ve n}(i+1)}{S_{\ve n}(i)}\right|}=1,
$$
and the power series
$$
\sum_{i=1}^\infty x^i \frac{\sigma^i S_{\ve n}(i)}{i^a}
$$
is uniformly convergent for $x \in (0,1).$
\label{ratiotest1}
\end{lemma}
\begin{proof}
If $\ve n$ does not have a leading one then $lim_{i\rightarrow \infty}{S_{\ve n}(i)}$ exists, so it is clear that the first statement holds. On the other hand if $\ve n$ has leading ones, we can first extract them, and we are only left to show the statement for $\ve n=(1):$
$$
\lim_{i\rightarrow \infty}{\left|\frac{S_{1}(i+1)}{S_{1}(i)}\right|}=\lim_{i\rightarrow \infty}{\frac{\frac{1}{i+1}+S_1(i)}{S_1(i)}}=
\lim_{i\rightarrow \infty}{\frac{1}{(1+i)S_1(i)}+1}=1.
$$
In order to prove the second statement we apply the ratio test \cite[p. 205]{Heuser2003}:
\begin{eqnarray}
\left|\frac{\frac{x^{i+1}\sigma^{i+1} S_\ve n(i+1)}{(i+1)^a}}{\frac{x^i\sigma^i S_\ve n(i)}{i^a}}\right|= 
\left|\frac{xS_\ve n(i+1)i^a}{(i+1)^aS_\ve n(i)}\right|=x\underbrace{\frac{i^a}{(i+1)^a}}_{\stackrel{i\rightarrow \infty}{\longrightarrow}1}
\underbrace{\left|\frac{S_\ve n(i+1)}{S_\ve n(i)}\right|}_{\stackrel{i\rightarrow\infty}{\longrightarrow}1}\stackrel{i\rightarrow\infty}{\longrightarrow}x.\nonumber
\end{eqnarray} 
Hence using the criterion of Weierstrass \cite[p. 555]{Heuser2003} the power series converges uniformly for $x \in (0,1)$.
\end{proof}

Now we can give the proof of Theorem \ref{powexp}.
\begin{proof}
From the definition of harmonic polylogarithms and the \textit{theorem of dominated convergence of Lebesgue} (see e.g. \cite{Temme1996}), which we can apply due to Lemma \ref{ratiotest1} and which will allow us the exchanges of sum and integral, we get:
\begin{eqnarray}
\H{0,\ve m}{x}&=&\int_0^x{\frac{1}{y}\sum_{i=1}^\infty{\frac{\sigma^i y^i}{i^a}S_\ve n(i)}dy}\nonumber\\
							&=&	\sum_{i=1}^\infty{\frac{\sigma^i}{i^a}S_\ve n(i)\int_0^x{y^{i-1}dy}}\nonumber\\
							&=&\sum_{i=1}^\infty{\frac{\sigma^ix^i}{i^{a+1}}S_\ve n(i)},\nonumber
\end{eqnarray}
\begin{eqnarray}					
\H{1,\ve m}{x}&=&\int_0^x{\frac{1}{1-y}\sum_{i=1}^\infty{\frac{\sigma^i y^i}{i^a}S_\ve n(i)}dy}=																											\int_0^x{\left(\sum_{k=0}^\infty{y^k}\right)\left(\sum_{i=1}^\infty{\frac{\sigma^i y^i}{i^a}S_\ve n(i)}\right) dy}\nonumber\\
					&=&\int_0^x{\sum_{i=0}^\infty{y^{i+1}\sum_{k=0}^i{\frac{\sigma^{k+1}}{(k+1)^a}S_\ve n(k+1)}}dy}\nonumber\\	
					&=&\int_0^x{\sum_{i=0}^\infty{y^{i+1}S_{\sigma a,\ve n}(i+1)}dy}=\int_0^x{\sum_{i=1}^\infty{y^{i}S_{\sigma a,\ve n}(i)}dy}\nonumber\\					  &=&\sum_{i=1}^\infty{\frac{x^{i+1}}{i+1}S_{\sigma a,\ve n}(i)}=\sum_{i=0}^\infty{\frac{x^{i+1}}{i+1}S_{\sigma a,\ve n}(i)}\nonumber\\
					&=&\sum_{i=1}^\infty{\frac{x^{i}}{i}S_{\sigma a,\ve n}(i-1)}.\nonumber
\end{eqnarray}
The third part follows analogously to the second part.				
\end{proof}

\begin{example} For $x\in \left[0,1\right],$
$$
\H{-1,1}x=\sum_{i = 1}^{\infty}\frac{x^i}{i^2} - \sum_{i = 1}^{\infty}\frac{{\left( -x \right) }^i\,\text{S}_{-1}(i)}{i}.
$$
\end{example}
With  \ttfamily HarmonicSums \rmfamily this job can be done as follows.
\begin{fmma}
\begin{mma}
\In \text{\bf HToS[H[-1, 1, x]]}\\
\Out {\H{-1,1}x=\sum_{i = 1}^{\infty}\frac{x^i}{i^2} - \sum_{i = 1}^{\infty}\frac{{\left( -x \right) }^i\,\text{S}[-1,i]}{i}}\\
\end{mma}
\begin{mma}
\In \text{\bf HToS[H[-1, 1, 0, 1, x]]}\\
\Out {-\left( \sum_{i = 1}^{\infty}\frac{x^i}{i^4} \right)  + \sum_{i = 1}^{\infty}\frac{{\left( -x \right) }^i\,\text{S}[-3,i]}{i} + 
  \sum_{i = 1}^{\infty}\frac{x^i\,\text{S}[2,i]}{i^2} - \sum_{i = 1}^{\infty}\frac{{\left( -x \right) }^i\,\text{S}[-1,2,i]}{i}}\\
\end{mma}
\end{fmma} 
If we want to analyse the values of harmonic polylogarithms at $x=1,$ we can first extract trailing zeros. According to Remark \ref{remarktrailing} we get a polynomial in $\log(x)$ with coefficients in the harmonic polylogarithms without trailing zeros. If we send $x$ to one and if the coefficients are continuous for $x\in [0,1]$ then only the constant term (the term without a power of $\log(x)$) will remain, since $\log(1)=0$. Note that he coefficients may not be continuous for $x \in [0,1]$ since they can contain harmonic polylogarithms with leading ones. However, according to the following lemma even the terms with coefficients which contain harmonic polylogarithms with leading ones will vanish if they contain a power of $\log(x)$ and if we send $x$ to one. Hence we only have to look at those functions without trailing zeros.
\begin{lemma}
For $p, q \in \N$,
\begin{equation}
\lim_{x\rightarrow1}{\log^p(x)\log^q(1-x)}=0\nonumber
\end{equation}
\end{lemma}
\begin{proof}
The proof follows by de l'Hospital's rule (see e.g. \cite[p. 287]{Heuser2003}) and induction.
\end{proof}
As worked out earlier, the expansion of the harmonic polylogarithms without trailing zeros is a combination of sums of the form:
$$
\sum_{i=1}^\infty x^i \frac{\sigma^i S_\ve n(i)}{i^a},\ \sigma \in \{-1,1\}, \ a \in \N.
$$
For $x\rightarrow1$ these sums turn into harmonic sums at infinity if $\sigma a\neq1$:
$$
\sum_{i=1}^\infty x^i \frac{\sigma^i S_\ve n(i)}{i^a}\rightarrow S_{\sigma a,\ve n}(\infty).
$$
If $\sigma a=1,$ these sums turn into
$$
\sum_{i=1}^\infty x^i \frac{S_\ve n(i)}{i}.
$$
Sending $x$ to one gives:
$$
\lim_{x\rightarrow1}\sum_{i=1}^\infty x^i \frac{S_\ve n(i)}{i}=\infty
$$
We see that these limits do not exist: this corresponds to the infiniteness of the harmonic sums with leading ones: $\lim_{k\rightarrow \infty}S_{1,\ve n}(k)=\infty;$ see Lemma \ref{consumlem}.
Hence the values of harmonic polylogarithms at one are related to the values of the multiple harmonic sums at infinity.
\begin{example}
$$
\H{-1,1,0}{1}=-2S_{3}(\infty)+S_{-1,-2}(\infty)+S_{-2,-1}(\infty)
$$
\end{example}
Using \ttfamily HarmonicSums \rmfamily we can carry out this evaluation at 1 automatically.
\begin{fmma}
\begin{mma}
\In \text{\bf RemoveTrailing0[H[-1, 1, 0, x]]}\\
\Out {\text{H}[0, x]\, \text{H}[-1, 1, x] - \text{H}[-1, 0, 1, x] - \text{H}[0, -1, 1, x]}\\
\end{mma} 
\begin{mma}
\In \text{\bf $\%$ /. H[0, x]$\rightarrow$ 0}\\
\Out {-\text{H}[-1, 0, 1, x] - \text{H}[0, -1, 1, x]}\\
\end{mma} 
\begin{mma}
\In \text{\bf $\%$ // HToS}\\
\Out {-2\,\left( \sum_{i = 1}^{\infty}\frac{x^i}{i^3} \right)  + \sum_{i = 1}^{\infty}\frac{{\left( -x \right) }^i\,\text{S}[-2,i]}{i} + 
  \sum_{i = 1}^{\infty}\frac{{\left( -x \right) }^i\,\text{S}[-1,i]}{i^2}}\\
\end{mma}
\noindent We send x to 1. Afterwards we rewrite the output in terms of harmonic sums, this is done by using the function TransformToSSums.
\begin{mma}
\In \text{\bf $\%$ /. x $\rightarrow$ 1 // TransformToSSums}\\
\Out {-2\,\text{S}[3,\infty] + \text{S}[-2,-1,\infty] + \text{S}[-1,-2,\infty]}\\
\end{mma}
\noindent We give a slightly bigger example: 
\begin{mma}
\In \text{\bf HToS[H[-1, 0, 1, 0, 0, 1, x]] /. x $\rightarrow$ 1 // TransformToSSums}\\
\Out {-\text{S}[6,\infty] + \text{S}[-1,-5,\infty] + \text{S}[3,3,\infty] - \text{S}[-1,-2,3,\infty]}\\
\end{mma}
\end{fmma}

\begin{remark}
Details on the procedure \normalfont{TransformToSSums} can be found in Chapter~\ref{Summation of Multiple Harmonic Sums}.
\end{remark}
 
There is one additional aspect, which causes difficulties when values at $x=1$ are considered. As mentioned earlier (see Remark \ref{hproexp}) the product formula for harmonic polylogarithms (\ref{hpro}) does not hold if one of the objects is divergent. So we have to be careful when we work with polylogarithms at $x=1$ which are divergent. We state the same example as in \cite{Remiddi2000} to illustrate this problem (the subleading terms will be wrong):\\
From Lemma \ref{powexp1} we have
\begin{eqnarray}
\H1x&=&\sum_{i=1}^\infty\frac{x^i}i,\nonumber
\end{eqnarray}
using (\ref{hpro}) (here we require that $x<1$) and Theorem \ref{powexp} we get
\begin{eqnarray}
(\H1x)^2&=&2\H{1,1}x=2\sum_{i=1}^\infty x^i \left(\frac{S_1(i)}i-\frac{1}{i^2}\right).\nonumber
\end{eqnarray}
Sending $x$ to one we get
\begin{eqnarray}
\H11&=&\lim_{x\rightarrow 1}{\sum_{i=1}^\infty\frac{x^i}i}=S_1(\infty),\nonumber \\
2\H{1,1}1&=&\lim_{x\rightarrow 1}2\sum_{i=1}^\infty x^i \left(\frac{S_1(i)}i-\frac{1}{i^2}\right)=2(S_{1,1}(\infty)-S_2(\infty))\nonumber \\
				&=&(S_1(\infty))^2-S_2(\infty).\nonumber
\end{eqnarray}
We see  that
$$
(\lim_{x\rightarrow 1}\H1x)^2 \neq \lim_{x\rightarrow 1}(\H1x)^2.
$$
We can however use multiple harmonic sums to solve this problem, since there the product also holds for sums at infinity. Namely, we can replace a sum at infinity by a sum with upper bound $N$ where $N$ is large but finite. Then we can use the shuffle algebra to cancel possible divergencies and afterwards we send $N$ to infinity. Following this approach we get
$$
(\lim_{x\rightarrow 1}\H1x)^2=(S_1(\infty))^2=2S_{1,1}(\infty)-S_2(\infty)=2\lim_{x\rightarrow 1}\H{1,1}x+\lim_{x\rightarrow 1}\H2x.
$$
This way is consistent and it will allow us to define the Mellin transform properly in the next section. Summarizing, we can again express all divergencies by the basic divergency $S_1(\infty),$ which is equal to $\lim_{x\rightarrow 1}{\H1x}$.

\section{The Mellin-Transform of Harmonic Polylogarithms}

In order to accomplish differentiation of harmonic sums, we look at the so called Mellin-transform of harmonic polylogarithms; compare \cite{Paris2001}.
\begin{definition}[Mellin-Transform]
Let $f(x)$ be a locally integrable function on $(0,1)$ and $n \in \R$. Then the Mellin transform is defined by:
\begin{eqnarray}
\M{f(x)}{n}&=&\int_0^1{x^nf(x)dx}, \ \textnormal{when the integral converges.}\nonumber
\label{abmell}
\end{eqnarray}
\end{definition}

\begin{remark}
A locally integrable function on $(0,1)$ is one that is absolutely integrable on all closed subintervals of $(0,1)$.
\end{remark}

For $f(x)=1/(1-x)$ the Mellin transform is not defined since the integral $\int_0^1\frac{x^n}{1-x}$ does not converge. Like in \cite{Remiddi2000} we extend the definition of the Mellin transform to include functions like $1/(1-x)$ as follows

\begin{definition}
Let $f(x)$ be a locally integrable function on $(0,1)$ and let $g(x)$ be a harmonic polylogarithm which is finite at $x=1,$ and let $p\in \N.$ Then we define  
\begin{eqnarray}
\Mp{f(x)}{n}&=&\M{f(x)}{n}=\int_0^1{x^nf(x)dx},\nonumber\\
\Mp{\frac{g(x)}{1-x}}n&=&\int_0^1{\frac{x^n g(x)-g(1)}{1-x}dx},\nonumber\\
\Mp{\frac{g(x)\log^p(1-x)}{1-x}}n&=&\int_0^1{\frac{(x^n g(x)-g(1))\log^p(1-x)}{1-x}dx}\nonumber
\label{abmellplus}
\end{eqnarray}
when the integrals converge. The function $f$ is supposed to be a harmonic polylogarithm without leading ones when the factor $1/(1-x)$ is present.
\end{definition}

According to the definition, if we want to compute the Mellin-transform of harmonic polylogarithms, we have to take care of harmonic polylogarithms which are not finite at $x=1$. As mentioned in Remark \ref{finitness} only harmonic polylogarithms with a leading one (and not followed only by zeros) are not finite at $x=1$. By using equation (\ref{leadingone}) we can extract all the powers of $\log(1-x)$ which cause the infiniteness at $x=1.$ After this extraction the remaining harmonic polylogarithms are finite at $x=1.$ The following lemma will guarantee the existence of the Mellin transform of harmonic polylogarithms.

\begin{lemma}
For a harmonic polylogarithm $\H{\ve m}x$ the integrals $$\int_0^1{x^n\H{\ve m}x dx} \ \textnormal{and} \ \int_0^1{\frac{x^n\H{\ve m}x dx}{1+x}}$$ converge for $x\in[0,1]$ and $n\in\R$. If $\H{\ve m}x$ does not have leading ones and $p \in \N$ then $$\int_0^1{\frac{x^n \H{\ve m}x-\H{\ve m}1}{1-x}dx} \ \textnormal{and} \ \int_0^1{\frac{(x^n \H{\ve m}x-\H{\ve m}1)\log^p(1-x)}{1-x}dx}$$ converge for $x\in[0,1]$ and $n\in\R$.
\end{lemma}
\begin{proof}
Let $n \in \R$ and consider $\int_0^1{x^n\H{\ve m}x dx}.$ After extracting trailing zeros and leading ones we get a combination of integrals of the form:
$$
\int_0^1{x^n\log^a(x)\H{\overline{\ve m}}x \log^b(1-x) dx}, \ a,b \in \N_0, 
$$
where $\H{\overline{\ve m}}x$ is continuous on [0,1].
It suffices to prove, that these integrals converge. $\abs{\H{\overline{\ve m}}{x}}$ is continuous on $[0,1]$ and hence has a maximum $C$ on $[0,1]$. Thus
\begin{eqnarray}
\abs{\int_0^1{x^n\log^a(x)\H{\overline{\ve m}}x\log^b(1-x)dx}}&\leq&\int_0^1{x^n\abs{\log^a(x)}
				\underbrace{\abs{\H{\overline{\ve	m}}x}}_{\leq C}\abs{\log^b(1-x)} dx}\nonumber\\
		&\leq& C \int_0^1{x^n\abs{\log^a(x)}\abs{\log^b(1-x)} dx}\nonumber\\
		&=&\lim_{\epsilon\rightarrow0}{C\int_{\epsilon}^{\frac{1}{2}}{\underbrace{x^n}_{\leq1}\abs{\log^a(x)} 																		 \underbrace{\abs{\log^b(1-x)}}_{\leq \abs{\log{\frac{1}{2}}}^b=:C_1} dx}}\nonumber\\
	  &&+	\lim_{\epsilon\rightarrow0}{C\int_{\frac{1}{2}}^{1-\epsilon}{\underbrace{x^n}_{\leq 1}\underbrace{\abs{\log^a(x)}}_{\leq 								\abs{\log{\frac{1}{2}}}^a=:C_2}\abs{\log^b(1-x)} dx}}\nonumber\\
	  &\leq&CC_1\lim_{\epsilon\rightarrow0}{\int_{\epsilon}^{\frac{1}{2}}\abs{\log^a(x)}dx}\nonumber\\
	  		&& +CC_2\lim_{\epsilon\rightarrow0}{\int_{\frac{1}{2}}^{\epsilon}\abs{\log^b(1-x)}dx}.\nonumber
\end{eqnarray}
Since the last two integrals converge, we finished this part of the proof. \\
The convergence of the second integral $\int_0^1{\frac{x^n\H{\ve m}x dx}{1+x}}$ can be proved analogously.
From now on we assume that $\H{\ve m}x$ does not have leading ones, hence $\H{\ve m}1$ is finite.
If we find $x_1\in \R$ with $0<x_1<1$ such that for all $x\in (x_1,1)$ we have
\begin{equation}
\abs{\frac{x^n \H{\ve m}x-\H{\ve m}1}{1-x}}\leq(1-x)^{-\frac{1}{2}},
\label{ungl}
\end{equation}
the convergence of $\int_0^1{\frac{x^n \H{\ve m}x-\H{\ve m}1}{1-x}dx}$ is proved; this follows since the integral $\int_0^1{(1-x)^{-\frac{1}{2}}}$ converges.
Since $x^n \H{\ve m}x-\H{\ve m}1$ is continuous on $[0,1],$ there is $x_2$ with $0<x_2<1$ such that $x^n \H{\ve m}x-\H{\ve m}1$ is either $\leq 0$ or $\geq 0$ on $[x_2,1].$ We will prove the $\geq 0$ case, the other case follows analogously. For $x\in (x_2,1)$ we want to show $(\ref{ungl})$ or equivalently: 
$$
\frac{x^n \H{\ve m}x-\H{\ve m}1}{(1-x)^{\frac{1}{2}}}\leq 1.
$$
Let $\ve m=(m_1,m_2,\ldots)$ with $m_1\in \{0,-1\}$ by assumption. With de l'Hospital's rule, see e.g. \cite[p. 287]{Heuser2003}, we get: 
\begin{eqnarray*}
\lim_{x\rightarrow 1}\frac{x^n \H{\ve m}x-\H{\ve m}1}{(1-x)^{\frac{1}{2}}}&=&\lim_{x\rightarrow 1}\frac{n x^{n-1} \H{\ve m}x-\frac{x^n}{1-m_1}\H{m_2,m_3,\ldots}x}{-\frac{1}{2}(1-x)^{-\frac{1}{2}}}\\&=&\lim_{x\rightarrow 1}\left(-2 n x^{n-1}(1-x)^{\frac{1}{2}}\H{\ve m}x\right)\\
&&+\lim_{x\rightarrow 1}\left(2\frac{x^n}{1-m_1}(1-x)^{\frac{1}{2}}\H{m_2,m_3,\ldots}x\right).
\end{eqnarray*}
Since $\H{m}1$ is finite, the first limit exists and is zero. If $m_2\neq1$ then $\H{m_2,m_3,\ldots}1$ exists and so the second limit is finite and equals zero. If $m_2=1$, we can extract leading ones and split the limit. We will get limits of the form
$$
\lim_{x\rightarrow 1}\left(2\frac{x^n}{x-m_1}(1-x)^{\frac{1}{2}}\log{(1-x)}^k \H{\ve a,\ldots}x\right)
$$
where $\H{\ve a,\ldots}x$ is finite at 1 and $k\in\N$. Since $\lim_{x\rightarrow 1}(1-x)^{\frac{1}{2}}\log{(1-x)}^k=0,$ all these limits vanish. Hence there is $x_1\in [x_2,1]$ such that (\ref{ungl}) holds and so $\int_0^1{\frac{x^n \H{\ve m}x-\H{\ve m}1}{1-x}dx}$ converges.

For the last integral we want to find again $x_1$ such that   
\begin{equation}
\abs{\frac{(x^n \H{\ve m}x-\H{\ve m}1)\log^p(1-x)}{(1-x)^{\frac{1}{2}}}}\leq 1
\label{ungl2}
\end{equation}
holds for all $x$ with $x_1<x<1.$ Again we can find $x_2,$ with $0<x_2<1$ such that $(x^n \H{\ve m}x-\H{\ve m}1)\log^p(1-x)$ is either $\leq 0$ or $\geq 0$ on $(x_2,1).$\\
We will prove the $\geq 0$ case, the other case follows analogously. Using again de l'Hospital's rule we get
\begin{eqnarray*}
&&\lim_{x\rightarrow 1}\frac{(x^n \H{\ve m}x-\H{\ve m}1)\log^p(1-x)}{(1-x)^{\frac{1}{2}}}=\\
\ \ \ \ &&\lim_{x\rightarrow 1}\left(2(1-x)^{\frac{1}{2}}\frac{k(x^n \H{\ve m}x-\H{\ve m}1)\log^{p-1}(1-x)}{1-x}\right)\\
\ \ \ \ &&-\lim_{x\rightarrow 1}\left(2(1-x)^{\frac{1}{2}}\log^k{(1-x)}\left(n x^{n-1} \H{\ve m}x-\frac{x^n}{1-m_1}\H{m_2,m_3,\ldots}x\right)\right).
\end{eqnarray*}
For the second limit we can use the same arguments as we used in the proof of the convergence of the third integral to show that its limit is zero. For the first limit we can use induction on $p$ to see that the limit is zero too. Hence there is $x_1\in [x_2,1]$ such that (\ref{ungl2}) holds and so $\int_0^1{\frac{(x^n \H{\ve m}x-\H{\ve m}1)\log^p(1-x)}{1-x}dx}$ converges.
\end{proof}

\begin{remark}
From now on we will call the extended Mellin transform $M^+$ just Mellin transform and we will write $M$ instead of $M^+.$
\end{remark}

\begin{example} For $x\in [0,1]$ and $n\in \R$  
\begin{eqnarray}
\M{\H{1,0}{x}}{n}&=&\int_0^1{x^n\H{1,0}x dx},\nonumber\\
\M{\frac{\H{0,1}{x}}{1-x}}{n}&=&\int_0^1{\frac{x^n \H{0,1}x-\H{0,1}1}{1-x}dx},\nonumber\\
\M{\frac{\H{1,0}{x}}{1-x}}{n}&=&\M{\frac{\H{0}x\H{1}x}{1-x}-\frac{\H{0,1}x}{1-x}}{n},\nonumber\\
&=&\M{\frac{\H{0}x\H{1}x}{1-x}}{n}-\M{\frac{\H{0,1}x}{1-x}}{n},\nonumber\\
&=&\int_0^1{\frac{(x^n \H{0,1}x-\H{0,1}1)\log{(1-x)}}{1-x}dx}-\int_0^1{\frac{x^n \H{0,1}x-\H{0,1}1}{1-x}dx};\nonumber
\end{eqnarray}
compare also \cite{Bluemlein1999}.
\end{example}

Now we want to study how we can actually calculate the Mellin transform of harmonic polylogarithms. We will proceed recursively. It is quite easy to get the Mellin transforms of harmonic polylogarithms with weight one \cite{Remiddi2000}; see also \cite{Bluemlein1999}.
\begin{lemma} For $n\in \N$ we have
\begin{eqnarray}
\M{\H{0}{x}}n&=& -\frac{1}{(n+1)^2}, \nonumber\\
\M{\H{1}{x}}n&=& \frac{S_1(n+1)}{n+1}, \nonumber\\
\M{\H{-1}{x}}n&=& (-1)^n\frac{S_{-1}(n+1)}{n+1}+\frac{\H{-1}1}{n+1}(1+(-1)^n). \nonumber
\end{eqnarray}
\label{weight1mel}
\end{lemma}
\begin{proof}
We just give a proof of the first identity. By using partial integration we get for $n\in \N:$
\begin{eqnarray}
\M{\H{0}{x}}n&=&\int_0^1{x^n\H{0}x dx}=\left.\frac{x^{n+1}}{n+1}\H{0}{x}\right|_0^1-\int_0^1{\frac{x^{n+1}}{n+1}\frac{1}{x}dx}\nonumber\\
&=&-\int_0^1{\frac{x^n}{n+1}dx}=\left.-\frac{x^{n+1}}{(n+1)^2}\right|_0^1=-\frac{1}{(n+1)^2}.\nonumber
\end{eqnarray}
\end{proof}
The higher weight results can now be obtained by recursion; see \cite{Remiddi2000,Bluemlein1999}.
\begin{lemma} For $n\in \N$ and $\ve m\in \left\{0,-1,1 \right\}^k$,
\begin{eqnarray}
\M{\H{0,\ve m}{x}}n &=&\frac{\H{0,\ve m}1}{n+1}-\frac{1}{n+1} \int_0^1 x^n \H{\ve m}{x}dx, \nonumber\\
\M{\H{1,\ve m}{x}}n &=&\frac{\H{1,\ve m}1}{n+1}+\frac{1}{n+1}\left(  \int_0^1 x^n \H{\ve m}{x}dx - \int_0^1 \frac{x^n}{1-x} \H{\ve m}{x}dx\right), \nonumber\\
\M{\H{-1,\ve m}{x}}n &=&\frac{\H{-1,\ve m}1}{n+1}-\frac{1}{n+1}\left( \int_0^1 x^n \H{\ve m}{x}dx - \int_0^1 \frac{x^n}{1+x} \H{\ve m}{x}dx\right). \nonumber																										
\end{eqnarray}
\label{melnotweighted}
\end{lemma}
\begin{proof}
We get the following results by partial integration:
\begin{eqnarray}
\int_0^1 x^n \H{0,\ve m}{x}dx &=& \left. \frac{x^{n+1}}{n+1}\H{0,\ve m}{x}\right|_0^1-\int_0^1 \frac{x^n}{n+1}\H{\ve m}{x}dx\nonumber\\
													&=&\frac{\H{0,\ve m}1}{n+1}-\frac{1}{n+1} \int_0^1 x^n \H{\ve m}{x}dx, \nonumber\\
\int_0^1 x^n \H{1,\ve m}{x}dx &=&\lim_{\epsilon \rightarrow 1} \int_0^{\epsilon} x^n \H{1,\ve m}{x}dx \nonumber\\
									&=&\lim_{\epsilon \rightarrow 1}\left( \left. \frac{x^{n+1}}{n+1}\H{1,\ve m}{x}\right|_0^{\epsilon}-\int_0^{\epsilon} 																					\frac{x^{n+1}}{(n+1)(1-x)}\H{\ve m}{x}dx \right) \nonumber\\
													&=&\frac{\H{1,\ve m}1}{n+1}-\frac{1}{n+1} \int_0^1 x^n (-1+\frac{1}{1-x})\H{\ve m}{x}dx, \nonumber\\
\int_0^1 x^n \H{-1,\ve m}{x}dx&=& \left. \frac{x^{n+1}}{n+1}\H{-1,\ve m}{x}\right|_0^1-\int_0^1 \frac{x^{n+1}}{(n+1)(1+x)}\H{\ve m}{x}dx\nonumber\\
													&=&\frac{\H{-1,\ve m}1}{n+1}-\frac{1}{n+1} \int_0^1 x^n (1-\frac{1}{1+x})\H{\ve m}{x}dx. \nonumber											\end{eqnarray}
\end{proof}

\begin{remark}
We want to point out again that $\H{1,\ve m}{1}:=\lim_{x \rightarrow 1}{\H{1,\ve m}{x}}$ and hence we are dealing with limit processes.
\end{remark}

As we could see in the previous lemma, we will need the results for harmonic polylogarithms weighted by $1/(1 + x)$ or $1/(1-x)$. We start with harmonic polylogarithms with weight one; see \cite{Remiddi2000,Bluemlein1999}.

\begin{lemma}For $n\in\N$ and $x\in(0,1),$
\begin{eqnarray}
\M{\frac{\H{0}{x}}{1-x}}{n} &=& S_2(n)-\zeta_2,\nonumber\\
\M{\frac{\H{0}{x}}{1+x}}{n} &=& (-1)^n(-\frac{\zeta_2}{2}-S_{-2}(n)),\nonumber\\
\M{\frac{\H{-1}{x}}{1-x}}{n} &=&\frac{\log^2(2)-\zeta_2}{2}+\log(2)(S_{-1}(n)-S_{1}(n))+S_{-1,-1}(n),\nonumber\\
\M{\frac{\H{-1}{x}}{1+x}}{n} &=&(-1)^n\left(\frac{\log^2(2)}{2}+\log(2)(S_{-1}(n)-S_{1}(n))-S_{1,-1}(n)\right),\nonumber\\
\M{\frac{\H{1}{x}}{1-x}}{n}&=&-S_{1,1}(n),\nonumber\\
\M{\frac{\H{1}{x}}{1+x}}{n}&=&(-1)^n\left(S_{-1,1}(n)+\frac{\zeta_2-\log^2(2)}{2}\right).\nonumber																		
\end{eqnarray}
\label{grundlemma}
\end{lemma}

\begin{proof}
We use again partial integration and Lemma \ref{weight1mel}. In addition we use tables by Vermaseren (see \cite{Vermaseren1998} and \cite{Remiddi2000}) where we can find values of multiple harmonic sums at infinity and harmonic polylogarithms at one. To verify the exchange of integrals and summations we used the \textit{theorem of dominated convergence of Lebesgue} (see e.g. \cite{Temme1996}):
\begin{eqnarray}
\M{\frac{\H{0}{x}}{1-x}}{n} &=& \int_0^{1} \frac{x^n \H{0}{x}}{1-x}dx= \int_0^{1} \sum_{i=n}^{\infty}x^i \H{0}{x}dx\nonumber\\
                              &=& \sum_{i=n}^{\infty} \int_0^{1} x^i \H{0}{x}dx=-\sum_{i=n}^{\infty}\frac{1}{(i+1)^2}\nonumber\\
                              &=& -\sum_{i=n+1}^{\infty}\frac{1}{i^2}=S_2(n)-S_2(\infty),\nonumber\\                                
\M{\frac{\H{0}{x}}{1+x}}{n}   &=& \int_0^{1} \frac{x^n \H{0}{x}}{1+x}dx=\int_0^{1}(-1)^n\sum_{i=n}^{\infty}(-x)^i \H{0}{x}dx\nonumber\\
                              &=& (-1)^n \sum_{i=n}^{\infty} \int_0^{1}(-x)^i \H{0}{x}dx=(-1)^n 																																					\sum_{i=n}^{\infty}\frac{(-1)^{i+1}}{(i+1)^2}\nonumber\\
                              &=& \sum_{i=n+1}^{\infty}\frac{(-1)^i}{i^2}=(-1)^n(S_{-2}(\infty)-S_{-2}(n)), \nonumber\\
\M{\frac{\H{-1}{x}}{1-x}}{n}&=& \int_0^{1}\frac{x^n \H{-1}{x}-\H{-1}1}{1-x}dx\nonumber\\
														&=& \int_0^{1}\sum_{i=n}^{\infty}(x^i \H{-1}{x}-x^{i-n}\H{-1}1)dx\nonumber
\end{eqnarray}
\begin{eqnarray} 
                               &=&\sum_{i=n}^{\infty} \int_0^{1}(x^i\H{-1}{x}-x^{i-n}\H{-1}1)dx\nonumber\\
                               &=&\sum_{i=n}^{\infty}\left(\int_0^{1} x^i\H{-1}{x}dx-\int_0^{1} x^{i-n}\H{-1}1dx \right)\nonumber\\
                               &=& \lim_{N\rightarrow \infty}\sum_{i=n}^{N}\left( \frac{S_{-1}(i+1)(-1)^i+\H{-1}1((-1)^i+1)}{i+1}
                                         -\frac{\H{-1}1}{i-n+1}\right)\nonumber\\
                               &=& \lim_{N\rightarrow \infty}\left(\sum_{i=n}^{N} \frac{S_{-1}(i+1)(-1)^i+\H{-1}1((-1)^i+1)}{i+1}
                                         -\sum_{i=1}^{N-n}\frac{\H{-1}1}{i}\right)\nonumber\\
  														 &=& \lim_{N\rightarrow \infty}\left(S_{-1,-1}(n)+\H{-1}1(S_{-1}(n)-S_1(n)) 																															-S_{-1,-1}(N+1)\right.\nonumber\\
  														 &&-\H{-1}1(S_{-1}(N+1)-S_1(N+1))\left.-\H{-1}1S_1(N-n)\right)\nonumber\\
                               &=& S_{-1,-1}(n)+\H{-1}1(S_{-1}(n)-S_1(n))-S_{-1,-1}(\infty)-\H{-1}1S_{-1}(\infty)\nonumber\\
                               &=& \frac{\log^2(2)-\zeta_2}{2}+\log(2)(S_{-1}(n)-S_{1}(n))+S_{-1,-1}(n),\nonumber\\                                                   
\M{\frac{\H{-1}{x}}{1+x}}{n}   &=& \int_0^{1} \frac{x^n \H{-1}{x}}{1+x}dx=\int_0^{1}(-1)^n\sum_{i=n}^{\infty}(-x)^i 																												\H{-1}{x}dx\nonumber\\
                               &=& (-1)^n \sum_{i=n}^{\infty} \int_0^{1} (-x)^i \H{-1}{x}dx\nonumber\\
                               &=& (-1)^n\sum_{i=n}^{\infty} \frac{S_{-1}(i+1)+\H{-1}1((-1)^i+1)}{i+1}\nonumber\\
                               &=& (-1)^n\lim_{N\rightarrow \infty}\left(S_{1,-1}(N)+\H{-1}1(S_1(N)-S_{-2}(N))-
                                    		S_{1,-1}(n)\right.\nonumber\\
                               &&\left.-\H{-1}1(S_1(n)-S_{-2}(n))	\right)\nonumber\\
                               &=&(-1)^n\left(\frac{\log^2(2)}{2}+\log(2)(S_{-1}(n)-S_{1}(n))-S_{1,-1}(n)\right),\nonumber\\                                
\M{\frac{\H{1}{x}}{(1-x)}}{n} &=& \int_0^{1} \frac{x^n \H{1}{x}-\H{1}{x}}{1-x}dx= 
                            		\sum_{i=n}^{\infty} \int_0^{1} x^i(\H{1}{x}-x^{i-n}\H{1}{x})dx\nonumber\\
                              &=& \lim_{N\rightarrow \infty}\sum_{i=n}^{N}\left(\frac{S_{1}(i+1)}{i+1}
                                         -\frac{S_{1}(i-n+1)}{i-n+1}\right)\nonumber\\
                              &=& \lim_{N\rightarrow \infty}\left(S_{1,1}(N+1)-S_{1,1}(n)-S_{1,1}(N-n+1)\right)\nonumber\\
                              &=& -S_{1,1}(n),\nonumber\\                              
\M{\frac{\H{1}{x}}{1+x}}{n}&=&\int_0^1 x^n\frac{\H{1}{x}}{1+x}dx=(-1)^n\sum_{i=n}^{\infty}\int_0^{1}(-x)^i \H{1}{x}dx\nonumber\\
                            &=&(-1)^n\sum_{i=n}^{\infty} \frac{(-1)^i S_1(i+1)}{i+1}
                                  =-(-1)^n\sum_{i=n+1}^{\infty} \frac{(-1)^i S_1(i)}{i}\nonumber\\
                            &=&-(-1)^n(S_{-1,1}(\infty)-S_{-1,1}(n))\nonumber\\
                            &=&(-1)^n\left(S_{-1,1}(n)+\frac{\zeta_2-\log^2(2)}{2}\right).\nonumber                                          
\end{eqnarray}
\end{proof}

The higher weight results can be again obtained by recursion, however we will first show uniform convergence for a series which will pop up.

\begin{remark}
In the following recursions we will use the following definitions: For $n\in\N$ and $x\in (0,1),$
\begin{eqnarray*}                               
S_{()}(n)&:=&S(n):=1,\\
\textnormal{H}_{()}(x)&:=&\textnormal{H}(x):=1.                                     
\end{eqnarray*}
\end{remark}

\begin{lemma}
Let $\ve m$ be a vector with components in $\{0,1,-1\}$, $\ve p$ be the empty vector or a vector with components in the non zero integers and let $k \in \N\cup\{0\}$. Then the series
$$
\sum_{i=1}^\infty{ x^i \H{\ve m}{x} \frac{S_\ve p(i+1)}{(i+1)^k}}
$$
converges uniformly for $x\in (0,1).$
\end{lemma}
\begin{proof}
We apply the ratio test \cite[p. 205]{Heuser2003} and use Lemma \ref{ratiotest1}:
\begin{eqnarray*}
\lim_{i\rightarrow \infty}{\left|\frac{x^{i+1} \H{\ve m}{x} \frac{S_\ve p(i+2)}{(i+2)^k}}{ x^i \H{\ve m}{x} \frac{S_\ve p(i+1)}{(i+1)^k}}\right|}&=&\lim_{i\rightarrow \infty}{\left|x\frac{S_\ve p(i+2)}{S_\ve p(i+1)}\frac{(i+1)^k}{(i+2)^k}\right|}\\
&=&x\lim_{i\rightarrow \infty}{\left|\frac{S_\ve p(i+2)}{S_\ve p(i+1)}\right|}\lim_{i\rightarrow \infty}{\left|\frac{(i+1)^k}{(i+2)^k}\right|}\\&=&x.
\end{eqnarray*}
Hence using the criterion of Weierstrass \cite[p. 555]{Heuser2003} the series converges uniformly for $x\in (0,1)$.
\end{proof}
In the following three lemmas we will present the recursion (compare \cite{Remiddi2000}), which we will use to compute the Mellin-transform of general harmonic polylogarithms.
 
\begin{lemma}For $x\in(0,1)$, $n\in\N$, $k\in\N\cup\{0\}$, $\ve m$ a vector with components in $\{0,1,-1\}$ or $m=()$ and $\ve p$ a vector with entries in $\N$ or $\ve p=()$ we have:
\begin{eqnarray}
\int_0^1 \sum_{i=n}^\infty{ (-x)^i \H{0,\ve m}{x} \frac{S_\ve p(i+1)}{(i+1)^k}}dx &=& -\H{0,\ve m}1(S_{-k-1,\ve p}(\infty)-S_{-k-1,\ve p}(n))\nonumber\\
							&&- \int_0^1 \sum_{i=n}^\infty (-x)^i \H{m}{x}\frac{S_\ve p(i+1)}{(i+1)^{k+1}}dx,\nonumber
\end{eqnarray}
\begin{eqnarray}							
\int_0^1 \sum_{i=n}^\infty{ (-x)^i \H{-1,\ve m}{x} \frac{S_\ve p(i+1)}{(i+1)^k}}dx &=& -\H{-1,\ve m}1(S_{-k-1,\ve p}(\infty)-S_{-k-1,\ve p}(n))\nonumber\\
							&&- \int_0^1 \sum_{i=n}^\infty x^i \H{\ve m}{x}\Bigl(-(-1)^iS_{k+1,\ve p}(i+1)\Bigr. \nonumber\\																				        &&+ \Bigl.(-1)^i\frac{S_\ve p(i+1)}{(i+1)^{k+1}}+(-1)^iS_{k+1,\ve p}(n)\Bigr)dx.\nonumber													\end{eqnarray}
\label{genrec1}
\end{lemma}

\begin{proof}
We give a proof of the first identity. The second follows similarly. By Lemma \ref{melnotweighted}:
\begin{eqnarray*}
\int_0^1 \sum_{i=n}^\infty{ (-x)^i \H{0,\ve m}{x} \frac{S_\ve p(i+1)}{(i+1)^k}}dx &=&
\sum_{i=n}^\infty{\frac{(-1)^iS_\ve p(i+1)}{(i+1)^k}\int_0^1 x^i \H{0,\ve m}{x}dx}\\
&=&\sum_{i=n}^\infty \frac{(-1)^iS_\ve p(i+1)}{(i+1)^k}\left(\frac{\H{0,\ve m}{1}}{i+1}\right.\\
		&&-\left.\frac{1}{i+1}\int_0^1 x^i\H{\ve m}{x}dx\right)\\
&=&-\H{0,\ve m}1(S_{-k-1,\ve p}(\infty)-S_{-k-1,\ve p}(n))\\
							&&- \int_0^1 \sum_{i=n}^\infty (-x)^i \H{m}{x}\frac{S_\ve p(i+1)}{(i+1)^{k+1}}dx.
\end{eqnarray*}
\end{proof}

\begin{lemma}For $n\in\N$, $k\in\N\cup\{0\}$, $\ve m$ a vector with components in $\{0,1,-1\}$ or $m=()$ and $\ve p$ a vector with entries in $\N$ or $\ve p=()$ we have in the ring of formal power series:
\begin{eqnarray}
\int_0^1 \sum_{i=n}^\infty{ (-x)^i \H{1,\ve m}{x} \frac{S_\ve p(i+1)}{(i+1)^k}}dx &=& -\H{1,\ve m}1(S_{-k-1,\ve p}(\infty)-S_{-k-1,\ve p}(n))\nonumber\\
							&&- \int_0^1 \sum_{i=n}^\infty x^i \H{\ve m}{x}\Bigl(-S_{-k-1,\ve p}(i+1)\Bigr. \nonumber\\ 																							      &&- \Bigl. (-1)^i\frac{S_\ve p(i+1)}{(i+1)^{k+1}}+S_{-k-1,\ve p}(n)\Bigr)dx.\nonumber
\end{eqnarray}
\label{genrec2}
\end{lemma}

\begin{lemma}For $n\in\N$, $k\in\N\cup\{0\}$, $\ve m$ a vector with components in $\{0,1,-1\}$ or $m=()$ and $\ve p$ a vector with entries in $\N$ or $\ve p=()$ we have in the ring of formal power series:
\begin{eqnarray}
\int_0^1 \sum_{i=n}^\infty{ x^i \H{0,\ve m}{x} \frac{S_\ve p(i+1)}{(i+1)^k}}dx &=& \H{0,\ve m}1(S_{k+1,\ve p}(\infty)-S_{k+1,\ve p}(n))\nonumber\\
							&&- \int_0^1 \sum_{i=n}^\infty x^i \H{\ve m}{x}\frac{S_\ve p(i+1)}{(i+1)^{k+1}}dx,\nonumber
\end{eqnarray}
\begin{eqnarray}							
\int_0^1 \sum_{i=n}^\infty{ x^i \H{1,\ve m}{x} \frac{S_\ve p(i+1)}{(i+1)^k}}dx &=& \H{1,\ve m}1(S_{k+1,\ve p}(\infty)-S_{k+1,\ve p}(n))\nonumber\\
							&&- \int_0^1 \sum_{i=n}^\infty x^i \H{\ve m}{x}\Bigl(S_{k+1,\ve p}(i+1)\Bigr. \nonumber\\      																								  &&- \Bigl. \frac{S_\ve p(i+1)}{(i+1)^{k+1}}-S_{k+1,\ve p}(n)\Bigr)dx,\nonumber\\
\int_0^1 \sum_{i=n}^\infty{ x^i \H{-1,\ve m}{x} \frac{S_\ve p(i+1)}{(i+1)^k}}dx &=& \H{-1,\ve m}1(S_{k+1,\ve p}(\infty)-S_{k+1,\ve p}(n))\nonumber\\
							&&- \int_0^1 \sum_{i=n}^\infty x^i \H{\ve m}{x}\Bigl((-1)^iS_{-k-1,\ve p}(i+1)\Bigr. \nonumber\\																				        &&+ \Bigl.\frac{S_\ve p(i+1)}{(i+1)^{k+1}}-(-1)^iS_{-k-1,\ve p}(n)\Bigr)dx.\nonumber\\
\end{eqnarray}
\label{genrec3}																					
\end{lemma}
\begin{proof}
We only show the third identity. It is implied by
\begin{eqnarray*}
\int_0^1 \sum_{i=n}^\infty{ x^i \H{-1,\ve m}{x} \frac{S_\ve p(i+1)}{(i+1)^k}}dx &=&
\sum_{i=n}^\infty{\frac{S_\ve p(i+1)}{(i+1)^k}\int_0^1 x^i \H{-1,\ve m}{x}dx}\\
&=&\sum_{i=n}^\infty \frac{S_\ve p(i+1)}{(i+1)^k}\left(\frac{\H{-1,\ve m}{1}}{i+1}-\int_0^1\frac{x^i\H{\ve m}{x}}{i+1}dx\right.\\
		&&+\left.\frac{1}{i+1}\int_0^1 \frac{x^i}{1+x}\H{\ve m}{x}dx\right)\\
&=&\H{-1,\ve m}1(S_{k+1,\ve p}(\infty)-S_{k+1,\ve p}(n))\\
							&&- \int_0^1 \sum_{i=n}^\infty x^i \H{\ve m}{x}\frac{S_\ve p(i+1)}{(i+1)^{k+1}}dx\\
							&&+\int_0^1\sum_{i=n}^\infty \frac{S_\ve p(i+1)}{(i+1)^{k+1}}x^i\sum_{j=0}^\infty{(-x)^j\H{\ve m}{x}}dx\\
							&=&\H{-1,\ve m}1(S_{k+1,\ve p}(\infty)-S_{k+1,\ve p}(n))\\
							&&- \int_0^1 \sum_{i=n}^\infty x^i \H{\ve m}{x}\Bigl((-1)^iS_{-k-1,\ve p}(i+1)\Bigr.\\																				          &&+ \Bigl.\frac{S_\ve p(i+1)}{(i+1)^{k+1}}-(-1)^iS_{-k-1,\ve p}(n)\Bigr)dx;
\end{eqnarray*}
in the last equation we used:
\begin{eqnarray*}
\sum_{i=0}^\infty \frac{S_\ve p(i+1)}{(i+1)^{k+1}}x^i\sum_{j=0}^\infty{(-x)^j\H{\ve m}{x}}&=&
\sum_{i=0}^\infty {\sum_{j=0}^i{ \frac{S_\ve p(j+1)}{(j+1)^{k+1}}x^j(-x)^{i-j}\H{\ve m}{x}}}\\
&=&	\sum_{i=0}^\infty{(-1)^i x^i\H{\ve m}{x}\sum_{j=0}^{i}{(-1)^j\frac{S_\ve p(j+1)}{(j+1)^{k+1}}}}
\end{eqnarray*}
\begin{eqnarray*}
\ \ \ \ \ \ \ \ \ \ \ \ \ \ \ \ &=&	\sum_{i=0}^\infty{(-1)^i x^i\H{\ve m}{x}\sum_{j=1}^{i+1}{-(-1)^j\frac{S_{\ve p}(j)}{(j)^{k+1}}}}\\
&=&	-\sum_{i=0}^\infty{(-1)^i x^i\H{\ve m}{x}S_{-k-1,\ve p}(i+1)}.\\
\end{eqnarray*}
\end{proof}

\begin{remark}
In Lemma \ref{genrec1} the left sides and the right sides converge and hence the equalities hold analytically.  
In Lemma \ref{genrec2} the left side converges, while the convergence of the right hand side is not obvious since $\H{1,\ve m}1$ is present. In Lemma \ref{genrec3} the left sides do not converge and so do the right hand sides. \\
The left hand sides of \ref{genrec1} - \ref{genrec3} can appear in the computation of the Mellin-transform and hence we use the equalities formally. It turns out that the divergencies cancel in the end and all divergent parts disappear. At this point one would have to prove that the Mellin-transforms and the results from these recursions are equal in an analytic point of view, however this is omitted here.
\end{remark}

We are now able to compute Mellin transforms of all harmonic polylogarithms. We can use the previous lemma even if neither $1/(1 + x)$ nor $1/(1-x)$ is present, since
$$
1=\frac{1}{1+x}+\frac{x}{1+x}.
$$
Hence we do not need Lemma \ref{melnotweighted} to compute the Mellin transform of harmonic polylogarithms (but we used it to establish the Lemmas \ref{genrec1} - \ref{genrec3}). In the base case one has to evaluate harmonic polylogarithms at $x=1$ and harmonic sums at infinity; for this task we use tables from \cite{Remiddi2000} which contain expressions up to weight 9. The results of the Mellin transforms should be finite, hence the divergencies introduced in the Lemmas \ref{genrec1}-\ref{genrec3} should cancel. The result of the Mellin transform is an expression with multiple harmonic sums in the argument $n$. The algorithm can be summarized in Algorithm~\ref{MellinAlgorithm}.
We illustrate the usage of the previous lemmas on two examples:
\begin{example} For $n\in \N,$ we perform the following formal calculus. 
\begin{eqnarray*}
\M{\frac{\H{1}{x}}{1+x}}{n}&=&(-1)^n\int_0^1 \sum_{i=n}^\infty{ (-x)^{i} \H{1}{x}}dx =(-1)^n\int_0^1 \sum_{i=n}^\infty{ (-x)^{i} \H{1}{x} \frac{S(i+1)}{(i+1)^0}}dx\\
&=&-(-1)^n\H{1}1(S_{-1}(\infty)-S_{-1}(n))- (-1)^n\int_0^1 \sum_{i=n}^\infty x^i\Bigl(-S_{-1}(i+1)\Bigr.\\
\end{eqnarray*}
\begin{eqnarray*}
&&- \Bigl. (-1)^i\frac{1}{(i+1)^{1}}+S_{-1}(n)\Bigr)dx\\
&=&-(-1)^n\H{1}1(S_{-1}(\infty)-S_{-1}(n))-(-1)^n\int_0^1 \sum_{i=n}^\infty -x^iS_{-1}(i+1)dx\\
&&-(-1)^n\int_0^1 \sum_{i=n}^\infty x^i\frac{(-1)^{i+1}}{i+1}dx-S_{-1}(n)(-1)^n\int_0^1 \sum_{i=n}^\infty x^idx\\
&=&(-1)^n\left(-\H{1}1(S_{-1}(\infty)-S_{-1}(n))+S_{1,-1}(\infty)-S_{1,-1}(n)-S_{-2}(\infty)\right.\\
&&\left.+S_{-2}(n)+S_{-1}(n)(S_1(\infty)-S_1(n))\right),\\
\M{\frac{\H{-1,0,1}{x}}{1-x}}{n}&=&\int_0^1\frac{x^n\H{-1,0,1}{x}-\H{-1,0,1}{1}}{1-x}dx=\int_0^1 \sum_{i=0}^\infty x^{i} \left(x^n\H{-1,0,1}{x}\right.\\
&&\left.-\H{-1,0,1}{1}\right)dx\\
&=&\int_0^1{ \left(\sum_{i=n}^\infty{ x^{i} \H{-1,0,1}{x}} - \H{-1,0,1}{1} \sum_{i=0}^\infty{x^{i}}\right)}dx\\
&=&\int_0^1{ \sum_{i=n}^\infty{ x^{i} \H{-1,0,1}{x}}}dx - \int_0^1{\H{-1,0,1}{1} \sum_{i=0}^\infty{x^{i}}}dx\\
&=&\int_0^1{ \sum_{i=n}^\infty{ x^{i} \H{-1,0,1}{x}}}dx - \H{-1,0,1}{1}\H1x\\
&=&(-1)^n\left(S_{-1,1}(n)+\frac{\zeta_2-\log^2(2)}{2}\right);
\end{eqnarray*}
see Lemma \ref{grundlemma}.
\end{example}

\begin{algorithm}
\label{MellinAlgorithm}
\caption{Mellin Transform}
\begin{algorithmic}
\State \bfseries input:\ \normalfont $f(x)=\H{\ve m}x$, $f(x)=\frac{\H{\ve m}x}{1+x}$ or $f(x)=\frac{\H{\ve m}x}{1-x}$
\Procedure{Mellin}{$f(x),n$}
\State $r=0$
\If{$f(x)=\H{\ve m}x$}
	\State $s=\Call{Mellin}{\frac{f(x)}{1+x},n}+\Call{Mellin}{\frac{f(x)}{1+x},n+1}$
\ElsIf{$f(x)=\frac{\H{\ve m}x}{1+x}$}
	\State $s=\int_0^1 \sum_{i=n}^\infty{ (-x)^i \H{\ve m}{x} \frac{S(i+1)}{(i+1)^0}}dx$\Comment{Use \ref{genrec1}-\ref{genrec3}}
\ElsIf{$f(x)=\frac{\H{\ve m}x}{1-x}$ and $\ve m[1]\neq1$}\Comment{first index $\neq1$}
	\State $s=\int_0^1 \sum_{i=n}^\infty{ x^i \H{\ve m}{x} \frac{S(i+1)}{(i+1)^0}}dx$\Comment{Use \ref{genrec1}-\ref{genrec3}}
	\State $r=\H{\ve m}{1}*\int_0^1 \sum_{i=n}^\infty{ x^i \H{}{x} \frac{S(i+1)}{(i+1)^0}}dx$\Comment{Use \ref{genrec1}-\ref{genrec3}}
\ElsIf{$f(x)=\frac{\H{\ve m}x}{1-x}$ and $\ve m[1]=1$}\Comment{first index $=1$}
	\State $s=\int_0^1 \sum_{i=n}^\infty{ x^i \H{\ve m}{x} \frac{S(i+1)}{(i+1)^0}}dx$\Comment{Use \ref{genrec1}-\ref{genrec3}}
	\State $r=\Call{RemoveLeading1}{\H{\ve m}{x}}$
	\State $nsum=$number of summands in $r$
	\State modify $r$ in the following way:
	\For{$i=1$ to $nsum$}
			\If{the $i$th summand of $r$ equals $c*\H{1,1,\ldots,1}{x}$} \Comment{$c \in \R$}
				\State add 1 to the index set and set $x=1$
			\ElsIf{it equals $c*\H{\ve y}{x}\H{1,1,\ldots,1}{x}$, $\ve y\neq(1,1,\ldots,1)$}
				\State add 1 to the index of $\H{1,1,\ldots,1}{x}$ and set $x=1$
			\ElsIf{it equals $c*\H{\ve y}{x}$, $\ve y\neq(1,1,\ldots,1)$}	
				\State multiply $\H{1}{1}$ to the summand and set $x=1$
			\EndIf
	\EndFor
\EndIf
\State $s=s-r$
\State transform all $\H{\ve y}{1}$ in $s$ to $S_{\overline{\ve y}}(\infty).$\Comment{Use Theorem \ref{powexp}}
\State expand all products in $s$\Comment{Use (\ref{hsumproduct})}
\State \textbf{return} $s$
\EndProcedure
\end{algorithmic}
\end{algorithm}
Finally let us look at the result for a bigger example:
\begin{example}
\begin{eqnarray}
M\left(\frac{\H{-1,1,0,-1}x}{1-x}\right)&=&\frac{63\,\zeta_2^2\,\log(2)}{20} + \frac{\log(2)^5}{30} + 3\,\log(2)\,\textnormal{Li}_4\left(\frac{1}{2}\right)\nonumber\\&+& 
 11\,\textnormal{Li}_5\left(\frac{1}{2}\right) + \frac{33\,\zeta_2^2\,S_{-1}(n)}{20} + \frac{3\,\zeta_2\,\log(2)^2\,S_{-1}(n)}{4} \nonumber\\&-&
 \frac{\log(2)^4\,S_{-1}(n)}{6} - 4\,\textnormal{Li}_4\left(\frac{1}{2}\right)\,S_{-1}(n) \nonumber\\&-&
 \frac{33\,\zeta_2^2\,S_{1}(n)}{20} - \frac{3\,\zeta_2\,\log(2)^2\,S_{1}(n)}{4} \nonumber
\end{eqnarray}
\begin{eqnarray} 
&+&
 \frac{\log(2)^4\,S_{1}(n)}{6} + 4\,\textnormal{Li}_4\left(\frac{1}{2}\right)\,S_{1}(n) \nonumber\\&+&
 \log(2)\,S_{-1, -1, -2}(n) + \frac{\zeta_2\,S_{-1, -1, 1}(n)}{2} \nonumber\\&-&
 \log(2)\,S_{-1, -1, 2}(n) + S_{-1, -1, -2, -1}(n) + \frac{3\,\zeta_2\,\zeta_3}{4} \nonumber\\&-&
 \frac{13\,\log(2)^2\,\zeta_3}{8} - \frac{13\,\log(2)\,S_{-1}(n)\,\zeta_3}{4} \nonumber\\&+&
 \frac{13\,\log(2)\,S_{1}(n)\,\zeta_3}{4} - \frac{845\,\zeta_5}{64}.\nonumber
\end{eqnarray}
\end{example}
The Mellin transform of harmonic polylogarithms is implemented in \ttfamily HarmonicSums\rmfamily: 
\begin{fmma}
\begin{mma}
\In \text{\bf Mellin[H[-1, 1, 0, x]/(1 + x), x, n]}\\
\Out {{\left( -1 \right)}^n\,\left( \text{H}[-1,-1,1,0,1] + \text{H}[-1,1,0,1]\,S[-1,n] - \text{H}[-1,1,0,1]\,S[1,n] + S[1,-1,2,n] 						\right)}\\
\end{mma} 
\begin{mma}
\In \text{\bf $\%$ // ReplaceByKnownFunctions}\\
\Out {{\left( -1 \right) }^n\,\Bigl(\Bigr. 2\,\text{li4half} + \frac{{\text{ln2}}^4}{12} - 
    \frac{{\text{ln2}}^2\,\text{z2}}{4} - \frac{7\,{\text{z2}}^2}{8} + 
    \frac{13\,\text{ln2}\,\text{z3}}{8} + \frac{\text{ln2}\,\text{z2}\,\text{S}[-1,n]}{2} - 
    \text{z3}\,\text{S}[-1,n] - \frac{\text{ln2}\,\text{z2}\,\text{S}[1,n]}{2} + 
    \text{z3}\,\text{S}[1,n] + \text{S}[1,-1,2,n]\Bigl. \Bigr)}\\
\end{mma} 
\begin{mma}
\In \text{\bf Mellin[H[-1, 1, 0, -1, x]/(1 - x), x, n] // ReplaceByKnownFunctions // Expand}\\
\Out {11\,\text{li5half} + 3\,\text{li4half}\,\text{ln2} + \frac{{\text{ln2}}^5}{30} + 
  \frac{63\,\text{ln2}\,{\text{z2}}^2}{20} - \frac{13\,{\text{ln2}}^2\,\text{z3}}{8} + 
  \frac{3\,\text{z2}\,\text{z3}}{4} - \frac{845\,\text{z5}}{64} - 
  4\,\text{li4half}\,\text{S}[-1,n] - \frac{{\text{ln2}}^4\,\text{S}[-1,n]}{6} + 
  \frac{3\,{\text{ln2}}^2\,\text{z2}\,\text{S}[-1,n]}{4} + 
  \frac{33\,{\text{z2}}^2\,\text{S}[-1,n]}{20} - 
  \frac{13\,\text{ln2}\,\text{z3}\,\text{S}[-1,n]}{4} + 4\,\text{li4half}\,\text{S}[1,n] + 
  \frac{{\text{ln2}}^4\,\text{S}[1,n]}{6} - 
  \frac{3\,{\text{ln2}}^2\,\text{z2}\,\text{S}[1,n]}{4} - 
  \frac{33\,{\text{z2}}^2\,\text{S}[1,n]}{20} + 
  \frac{13\,\text{ln2}\,\text{z3}\,\text{S}[1,n]}{4} + \text{ln2}\,\text{S}[-1,-1,-2,n] + 
  \frac{\text{z2}\,\text{S}[-1,-1,1,n]}{2} - \text{ln2}\,\text{S}[-1,-1,2,n] + 
  \text{S}[-1,-1,-2,-1,n]}\\
\end{mma}
\end{fmma}

\begin{remark}
The procedure \textnormal{ReplaceByKnownFunctions} uses tables by Vermaseren to replace harmonic sums at infinity and harmonic polylogarithms in one by known functions. For these known functions we use the notation used by Vermaseren. The tables can be found at www.nikhef.nl/$\sim$form. 
\end{remark}

\section{The Inverse Mellin Transform of Multiple Harmonic Sums}
\label{Inverse Mellin Transform}
We already saw that the Mellin transform of harmonic polylogarithms can be expressed in terms of multiple harmonic sums. Now we want to go the other way round.\\
In order to get the inverse Mellin transform of multiple harmonic sums, we have to consider distributions. In fact we will consider distributions which are either differentiable functions in the class of $C^\infty((0,1))$ or are Dirac-$\delta$-distributions $\delta(1-x) \in D'[0,1]$. We want to point out that the inverse Mellin transform of constants is $\delta(x-1)$ where $\delta(x)$ is the Dirac-$\delta$-distribution and that $\delta(x-1)$ will only appear as the inverse Mellin transform of constants. Since we are mainly interested in the differentiation of multiple harmonics sums and since a constant vanishes if it is differentiated, we will skip further details on the inverse Mellin transform of constants. 
Subsequently, we will write $\textnormal{M}^{-1}(c)$ for the inverse Mellin transform of a constant $c$.\\
To prepare the stage we give some properties of the shuffle product of harmonic sums; here the product among harmonic sums is defined as for harmonic polylogarithms in Definition \ref{polyshuff}.

\subsection{Properties of the Shuffle Product}

\begin{definition}
For $a_k\in \Z/\{0\}$ with $1\leq k\leq l$ and $c \in \Z/\{0\}$, we define 
\begin{eqnarray*}
S_{a_1,a_2,\ldots,a_{k-1},\left\langle a_k, c \right\rangle, a_{k+1},\ldots,a_l}(n)&:=&
 S_{a_1,a_2,\ldots,a_{k-1}, a_k, c , a_{k+1},\ldots,a_l}(n)\\
&&+S_{a_1,a_2,\ldots,a_{k-1},a_k, a_{k+1},c,a_{k+2},\ldots,a_l}(n)\\
&&+\cdots
+S_{a_1,a_2,\ldots,a_{l-2},a_{l-1},c,a_l}(n)\\
&&+S_{a_1,a_2,\ldots,a_{l-1},a_l,c}(n);
\end{eqnarray*}
similarly, we define
\begin{eqnarray*}
(a_1,a_2,\ldots,a_{k-1},\left\langle a_k, c \right\rangle, a_{k+1},\ldots,a_l)&:=&
\{(a_1,a_2,\ldots,a_{k-1}, a_k, c , a_{k+1},\ldots,a_l),\\
&&(a_1,a_2,\ldots,a_{k-1},a_k, a_{k+1},c,a_{k+2},\ldots,a_l)\\
&&,\cdots
,(a_1,a_2,\ldots,a_{l-2},a_{l-1},c,a_l),\\
&&(a_1,a_2,\ldots,a_{l-1},a_l,c)\}.
\end{eqnarray*}
\label{placehold}
\end{definition}

\begin{example}
\begin{eqnarray*}
S_{1,\left\langle 2, 3 \right\rangle,4,5}(n)&=&S_{1,2,3,4,5}(n)+S_{1,2,4,3,5}(n)+S_{1,2,4,5,3}(n).\\
\end{eqnarray*}
\end{example}

\begin{lemma}
Let $a_i,b_j\in \Z/\{0\}$, for $1\leq i\leq k$ and $1\leq j\leq l$ and consider the shuffle product $S_{a_1,a_2,\ldots,a_k}(n) \shuffle S_{b_1,b_2,\ldots,b_l}(n)$. Let $P$ be the representation of 
$$
S_{a_1,a_2,\ldots,a_k}(n) \shuffle S_{b_1,b_2,\ldots,b_l}(n)
$$
such that it is linear in the harmonic sums \footnote{Note that this is always possible by using (\ref{hshuffpro})}. If we replace every occurrence of $b_l$ in $P$ by $\left\langle b_l, b_{l+1} \right\rangle$ and call the result $\overline{P}$, then we get (using Definition \ref{placehold}) that $\overline{P}$ is a representation of
$$
S_{a_1,a_2,\ldots,a_k}(n)\shuffle S_{b_1,b_2,\ldots,b_l,b_{l+1}}(n)
$$
\label{sl1}
\end{lemma}
We give an example which illustrates this lemma.
\begin{example} Consider
$$
S_{a_1,a_2}(n)\shuffle S_{b_1}=S_{a_1,a_2,b_1}(n)+S_{a_1,b_1,a_2}(n)+S_{b_1,a_1,a_2}(n).
$$
If we replace $b_1$ by $\left\langle b_1, b_2 \right\rangle$ we get:
\begin{eqnarray*}
S_{a_1,a_2,\left\langle b_1, b_2 \right\rangle}(n)+S_{a_1,\left\langle b_1, b_2 \right\rangle,a_2}(n)+S_{\left\langle b_1, b_2 \right\rangle,a_1,a_2}(n)&=&
S_{a_1,a_2,b_1,b_2}(n)+S_{a_1,b_1,b_2,a_2}(n)\\
&&+S_{a_1,b_1,a_2,b_2}(n)+S_{b_1,b_2,a_1,a_2}(n)\\
&&+S_{b_1,a_1,b_2,a_2}(n)+S_{b_1,a_1,a_2,b_2}(n)\\
&=&S_{a_1,a_2}(n)\shuffle S_{b_1,b_2}.
\end{eqnarray*}
\end{example}

\begin{remark}
If we define $\ve a \cup \ve b$ for two vectors $\ve a=(a_1,a_2,\ldots,a_k)$ and $\ve b=(b_1,b_2,\ldots,b_l)$ as the set of all mergings of $\ve a$ and $\ve b$ in which the relative orders of the $a_i$ and $b_j$ are preserved (for example $(a_1,a_2,b_1,b_2,\ldots,b_l,a_3,a_4,\ldots,a_k)\in \ve a \cup \ve b$ but $(a_2,a_1,b_1,b_2,\ldots,b_l,a_3,a_4,\ldots,a_k) \notin \ve a \cup \ve b)$ then the shuffle product is nothing else but
$$
S_{\ve a}(n)\shuffle S_{\ve b}(n)=\sum_{\ve c\in \ve a \cup \ve b}{S_{\ve c}(n)}.
$$
\label{remshuff}
\end{remark}

We give a proof of Lemma \ref{sl1}:
\begin{proof}
Let $\ve a=(a_1,a_2,\ldots,a_k)$ and $\ve b=(b_1,b_2,\ldots,b_l)$, let $\overline{A}$ be the set $\ve a \cup \ve b,$ however replace each $b_l$ in $\overline{A}$ by $\left\langle b_l, b_{l+1} \right\rangle$ and let $B=(a_1,a_2,\ldots,a_k) \cup (b_1,b_2,\ldots,b_l,b_{l+1}).$ Let $A$ be the union of the sets which arise after using Definition \ref{placehold} for all elements in $\overline{A}.$ It remains to show that $A=B$.\\
We start by showing $A \subseteq B:$ Take an element $\ve x \in A$. The relative order of the $a_i$ and $b_j$ in $\ve a \cup \ve b$ is preserved and the exchange of $b_l$ by $\left\langle b_l, b_{l+1} \right\rangle$ does not change the order of $a_i$ and $b_j.$ Due to Definition \ref{placehold} $b_l$ is left to $b_{l+1}$ for all harmonic sums in $A.$ Hence $x$ has to be in $B.$\\
Now we show $B \subseteq A:$ Obviously it is always possible to find for a given $x\in B$ an element $y\in \overline{A}$ such that after using Definition \ref{placehold} $x$ occurs in $y$ (for example let $x=(a_1,b_1,b_2,\ldots ,b_{l-1},b_l,a_2,\ldots,a_k,b_{l+1})$ then take $y=(a_1,b_1,b_2,\ldots,b_{l-1},\left\langle b_l, b_{l+1} \right\rangle,a_2,\ldots,a_k)$). Hence $x\in A.$ This completes the proof.
\end{proof}

\subsection{The \upshape Most Complicated \itshape Harmonic Sum}

In order to get the inverse Mellin transform of multiple harmonic sums, we exploit the following order on the set of harmonic sums.

\begin{definition}[Order on harmonic sums]
Let $S_{\ve m_1}(n)$ and $S_{\ve m_2}(n)$ be harmonic sums with weights $w_1$, $w_2$ and depths $d_1$ and $d_2,$ respectively. Then
$$
		  	\begin{array}{ll}
						S_{\ve m_1}(n) \prec S_{\ve m_2}(n), \ \textnormal{if} \ w_1<w_2, &\\
						S_{\ve m_1}(n) \prec S_{\ve m_2}(n), \ \textnormal{if} \ w_1=w_2 \ \textnormal{and} \ d_1<d_2.& 
				\end{array}
$$
For a set of harmonic sums we call a harmonic sum \upshape most complicated \itshape if it is a largest function with respect to $\prec$. 
\label{sord} 
\end{definition}

\begin{example}
We have
$$S_{1,1}(n)\prec S_{1,2}(n)\prec S_{1,1,1}(n)\prec S_{2,1,1}(n).$$
If we consider the set of these 4 sums then $S_{2,1,1}(n)$ is \itshape most complicated\upshape. Note that $\prec$ is not a linear ordering, e.g.,
$$
S_{1,2,3}(n) \nprec S_{2,1,3}(n) \ \textnormal{and} \ S_{2,1,3}(n) \nprec S_{1,2,3}(n).
$$
If we consider the set of these 2 sums then $S_{1,2,3}(n)$ and $S_{2,1,3}(n)$ are \itshape most complicated\upshape.
\end{example}

In the following we show that in the Mellin transform of a harmonic polylogarithm there is only one \itshape most complicated \upshape harmonic sum. 

\begin{prop}
In the Mellin transform of a harmonic polylogarithm $\H{\ve m}x$ weighted by $1/(1-x)$ or $1/(1+x)$ there is only one \upshape most complicated \itshape harmonic sum $S_{\overline{\ve m}}(n),\ie$
$$
\M{\frac{\H{\ve m}{x}}{1\pm x}}{n}=S_{\overline{\ve m}}(n)+t
$$
where all harmonic sums in $t$ occur linearly and are smaller then $S_{\overline{\ve m}}(n)$.
\label{singlmostcomp}
\end{prop}

Let us consider the Mellin transform of the polylogarithm $\H{m_1 \ve m}x$ weighted by $\frac{1}{1-x}$ or $\frac{1}{1+x}$ and let $w$ be the weight of $\H{m_1 \ve m}{x}.$ By looking at the recursions in the Lemmas~\ref{genrec1}, \ref{genrec2} and \ref{genrec3} we can see that the multiple harmonic sums appearing in the Mellin transform have weight less or equal $w+1.$ We consider now those harmonic sums in the Mellin transform which have weight $w+1$ and which have the maximal depth and so are \itshape most complicated\upshape. 

Where do these sums emerge from? Again looking at the recursions in the Lemmas~\ref{genrec1}, \ref{genrec2} and \ref{genrec3} we see that if $m_1=1$ then we will have to do either (depending if there is a factor $(-1)^i$ present) the following 3 integrals 
\begin{eqnarray*}
&&\int_0^1 \sum_{i=n}^\infty x^i \H{\ve m}{x}S_{k+1,\ve p}(i+1)dx,\\
&&\int_0^1 \sum_{i=n}^\infty x^i \H{\ve m}{x}\frac{S_\ve p(i+1)}{(i+1)^{k+1}}dx,\\
&&\int_0^1 \sum_{i=n}^\infty x^i \H{\ve m}{x}S_{k+1,\ve p}(n)dx,
\end{eqnarray*}
or these 3 integrals
\begin{eqnarray*} 
&&\int_0^1 \sum_{i=n}^\infty x^i \H{\ve m}{x}S_{-k-1,\ve p}(i+1)dx,\\
&&\int_0^1 \sum_{i=n}^\infty x^i \H{\ve m}{x}(-1)^i\frac{S_\ve p(i+1)}{(i+1)^{k+1}}dx,\\
&&\int_0^1 \sum_{i=n}^\infty x^i \H{\ve m}{x}S_{-k-1,\ve p}(n)dx
\end{eqnarray*}
in the next step.
Likewise, we have to do three similar integrals if $m_1=-1.$ In this situation ($m_1=1$ or $m_1=-1$) the \itshape most complicated \upshape harmonic sums can only emerge from the first or the third integral; there the weight of the polylogarithm is reduced by one and the depth of the harmonic sum is raised by one; in the second integral however the weight of the harmonic polylogarithm is reduced by one, but the depth of the harmonic sum is not raised. In the situation $m_1=0$ there is just one integral and the depth of the harmonic sum is not changed. So these sums emerge from those recursion branches where in each step either a first or a third integral is used, whenever the entry in $\ve m$ is $1$ or $-1.$

How do these sums look like? Let us start with an example.
\begin{example}
Consider the Mellin transform of $\frac{\H{0,0,1,0,1}x}{1+x}.$
If we take the branch where always the first integral is used, we get the sum $S_{1,2,3}(n).$
If we take the branch where always the third integral is used, we get $S_{1}(n)*S_{2}(n)*S_{3}(n).$ This product will produce sums of highest depth: 
\begin{eqnarray*}
S_{1}(n)*S_{2}(n)*S_{3}(n)&=&
S_{1, 2, 3}(n) + S_{1, 3, 2}(n) + S_{2, 1, 3}(n) + 
  S_{2, 3, 1}(n) \\&&+ S_{3, 1, 2}(n) + S_{3, 2, 1}(n)+t.
\end{eqnarray*}
Here $t$ is a linear expression in sums of weight 6 with depth less then 3.
\end{example}

\begin{notation}
Let $\ve a=(a_1,a_2,\ldots,a_k)$ and $\ve b=(b_1,b_2,\ldots,b_l).$ We denote the concatenation product by $\ve a \cdot \ve b.$ Hence
$\ve a \cdot \ve b=(a_1,a_2,\ldots,a_k,b_1,b_2,\ldots,b_l).$
\end{notation}
\begin{definition}
Let $\ve{a}=(a_1,a_2,\ldots,a_k)$ be a vector of length $k$, $a_i\in \Z \backslash \left\{0\right\}$ and let $d, n \in \N.$  We define 
$$\textnormal{rev}(\ve a):=S_{a_k,a_{k-1},\ldots,a_1}(n)$$
and 
$$
\textnormal{part}(\ve a):=\left\{S_{\ve{b}_1}(n)\cdot S_{\ve{b}_2}(n)\cdots S_{\ve{b}_d}(n) \left| \right. (\ve{b}_1\cdot\ve{b}_2\cdot \ldots \cdot\ve{b}_d)=\ve a\right\}.
$$
Let $x=S_{\ve{b}_1}(n)\cdot S_{\ve{b}_2}(n)\cdots S_{\ve{b}_d}(n)$ be an element of $\textnormal{part}(\ve a)$. We call $d$ the length of $x$.
\label{abdefpart}
\end{definition}

\begin{example}
For $\ve a=(1,2,3)$, 
$$
\textnormal{part}(\ve a)=\left\{S_{1,2,3}(n),S_{1,2}(n)S_{3}(n),S_{1}(n)S_{2,3}(n),S_{1}(n)S_{2}(n)S_{3}(n)\right\}.
$$
\end{example}

In order to get all \itshape most complicated \upshape harmonic sums, we first look at the branch where we only use the first integral and look at the index set $\ve a$ of this multiple harmonic sum (let $\sigma$ be the sign of this sum). All \itshape most complicated \upshape harmonic sums have weight and depth equal to the weight and depth of $S_{\ve a}(n)$. Each element in $\textnormal{part}(\ve a)$ corresponds to exactly one branch where we use in each step either the first or the third integral. The sign of each element depends on $\sigma$ and the length of the element. If the length of the element is $d$ then the sign is $-\sigma\,(-1)^{d}.$ Combining all possible branches we get the following sum 
$$\sum_{x\in\textnormal{part}(\ve a)}{(-1)^{\textnormal{Length}(x)}x}.$$
If we work out the sum and linearize the harmonic sums, a lot of cancellation will take place, in fact we have the following lemma.
\begin{lemma}
Let $\ve{a}=(a_1,a_2,\ldots,a_k)$ be a vector of length $k$, $a_i\in \Z \backslash \left\{0\right\}$, then
$$
\sum_{x\in\textnormal{part}(\ve a)}{(-1)^{\textnormal{Length}(x)}x}=(-1)^{k}\textnormal{rev}(\ve a)+t,
$$
where $t$ is a linear combination of sums of depth less than $k$. 
\label{partlem}
\end{lemma}

\begin{remark}
If we know the set 
$$\textnormal{part}(\ve a)=\left\{S_{\ve{b}_1}(n)\cdot S_{\ve{b}_2}(n)\cdots S_{\ve{b}_d}(n) \left| \right. (\ve{b}_1\cdot\ve{b}_2\cdot  \ldots \cdot \ve{b}_d)=\ve a\right\}$$ for $\ve a=(a_1,\ldots,a_k),$ we can get the set $\textnormal{part}(\overline{\ve a})$ for $\overline{\ve a}=(a_1,\ldots,a_k,a_{k+1})$:
\begin{eqnarray}
\textnormal{part}(\overline{\ve a})&=&
\underbrace{\left\{S_{\ve{b}_1}(n)\cdot S_{\ve{b}_2}(n)\cdots S_{\ve{b}_d}(n)\cdot S_{a_{k+1}}(n) \left| \right. (\ve{b}_1\cdot \ve{b}_2 \cdot \ldots \cdot \ve{b}_d)=\ve a\right\}}_{\textnormal{part}_1(\overline{\ve a}):=} \nonumber \\
&&\stackrel{\bullet}{\bigcup}
\underbrace{\left\{S_{\ve{b}_1}(n)\cdot S_{\ve{b}_2}(n)\cdots S_{\ve{b}_d,a_{k+1}}(n) \left| \right. (\ve{b}_1 \cdot \ve{b}_2 \cdot \ldots \cdot \ve{b}_d)=\ve a\right\}}_{\textnormal{part}_2(\overline{\ve a}):=}.
\label{partrec}
\end{eqnarray}
\label{reminv}
\end{remark}

Now we can give the proof of Lemma \ref{partlem}.
\begin{proof}
We proceed by induction on the length $k$ of $\ve{a}=(a_1,a_2,\ldots,a_k).$ If $k=1,$ then $\textnormal{part}(\ve a) =\left\{S_{a_1}(n)\right\}$ and
$$
\sum_{x\in\textnormal{part}(\ve a)}{(-1)^{\textnormal{Length}(x)}x}=-S_{a_1}(n)=(-1)^1S_{\underbrace{\textnormal{rev}(\ve a)}_{a_1}}(n).
$$
Thus Lemma \ref{partlem} is true for $k=1.$
Let us assume that the lemma holds for $k$ and let $\overline{\ve a}=(a_1,a_2,\ldots,a_{k+1}).$ With the notation given in Remark \ref{reminv},
\begin{eqnarray*}
\sum_{x\in\textnormal{part}(\overline{\ve a})}{(-1)^{\textnormal{Length}(x)}x}&=&
\underbrace{\sum_{x\in\textnormal{part}_1(\overline{\ve a})}{(-1)^{\textnormal{Length}(x)}x}}_{\textnormal{A}:=}
+\underbrace{\sum_{x\in\textnormal{part}_2(\overline{\ve a})}{(-1)^{\textnormal{Length}(x)}x}}_{\textnormal{B}:=}.
\end{eqnarray*}
Let us first look at A:
$$
\textnormal{A}=-S_{a_{k+1}}(n)\sum_{x\in\textnormal{part}(\ve a)}{(-1)^{\textnormal{Length}(x)}x}.
$$
According to the induction hypothesis we get 
$$
\sum_{x\in\textnormal{part}(\ve a)}{(-1)^{\textnormal{Length}(x)}x}=(-1)^{k}\textnormal{rev}(\ve a)+t
$$
where $t$ is a linear combination of sums of depth less than $k$. Therefore the sums of highest depth in A are emerging from $S_{a_{k+1}}(n)S_{a_k,a_{k-1},\ldots,a_1}(n).$ By using (\ref{hsumproduct}) we get
$$
\textnormal{A}=-(-1)^{k}(S_{a_{k+1},a_{k},\ldots,a_1}(n)+S_{a_{k},a_{k+1},\ldots,a_1}(n)+\cdots+S_{a_{k},\ldots,a_1,a_{k+1}}(n))+t_A
$$
where $t_A$ is a linear combination of sums of depth less than $k+1$.\\
Let us now look at B. The sums of highest depth in B have depth $k+1$ and result from the first and second sum in (\ref{hsumproduct}) since in the third sum the depth is reduced. Hence we can neglect the third sum and just look at the first 2 sums. This implies that it suffices to look at the shuffle product and not at the quasi-shuffle product to detect the sums of highest depth. By the induction assumption, the only harmonic sum in $\sum_{x\in\textnormal{part}(\ve a)}{(-1)^{\textnormal{Length}(x)}x}$ of highest depth is $S_{a_k,a_{k-1},\ldots,a_1}(n)$. Using Lemma \ref{sl1} and (\ref{partrec}) we get that the sums of highest depth in B originate from $S_{\left\langle a_k, a_{k+1} \right\rangle ,a_{k-1},\ldots,a_1}(n)$. We obtain
\begin{eqnarray*}
\textnormal{B}=(-1)^{k}S_{\left\langle a_k, a_{k+1} \right\rangle ,a_{k-1},\ldots,a_1}(n)&=&(-1)^{k}(S_{a_{k},a_{k+1},\ldots,a_1}(n)+S_{a_{k},a_{k-1},a_{k+1},a_{k-2}\ldots,a_1}(n)\\
&&+\cdots+S_{a_{k},\ldots,a_1,a_{k+1}}(n))+t_B
\end{eqnarray*}
where $t_B$ is a linear combination of sums of depth less than $k+1$. Summarizing we get
\begin{eqnarray*}
\sum_{x\in\textnormal{part}(\overline{\ve a})}{(-1)^{\textnormal{Length}(x)}x}&=&A+B=-(-1)^{k}S_{a_{k+1},a_{k},\ldots,a_1}(n)+t_A+t_B\\
&=&(-1)^{k+1}S_{a_{k+1},a_{k},\ldots,a_1}(n)+t_A+t_B.
\end{eqnarray*}
This completes the proof of Lemma \ref{partlem}.
\end{proof}

\begin{example}
For $\ve a=(1,2,3)$, 
\begin{eqnarray*}
\sum_{x\in\textnormal{part}(\ve a)}{(-1)^{\textnormal{Length}(x)}x}&=& -S_{1,2,3}(n)+S_{1,2}(n)S_{3}(n)+S_{1}(n)S_{2,3}(n)-S_{1}(n)S_{2}(n)S_{3}(n)\\
&=&-S_{3,2,1}(n)+ S_{5, 1}(n)+ S_{3,3}(n)-S_{6}(n).
\end{eqnarray*}
\end{example}

Using the considerations above we get that Proposition \ref{singlmostcomp} holds with $S_{\overline{\ve m}}(n)=\textnormal{rev}(\ve a).$ The harmonic sum $S_{\overline{\ve m}}(n)$ originates from the branch of the recursion where always the third integral is used. \\
We emphasize once more that we are mainly interested in the computation of a harmonic polylogarithm for which a given harmonic sum is the \itshape most complicated \upshape harmonic sum in the Mellin transform. But first we go the other way round. We have to determine $\textnormal{rev}(\ve a).$

\subsubsection{Finding the \itshape Most Complicated \upshape Harmonic Sum}

\bfseries Given: \normalfont a harmonic polylogarithm $\H{\ve m}{n}.$\\
\bfseries Find: \normalfont the \itshape most complicated \upshape harmonic sum $S_{\ve a}(n)$ and its sign $\sigma$ in $\M{\frac{\H{\ve m }x}{1\pm x}}{n}.$

Of course, we could just compute the Mellin transform and look for the most complicated sum in the result. However, in our considerations it is crucial that this job can be attacked directly. In general, this is possible by applying the Lemmas \ref{genrec1}, \ref{genrec2} and \ref{genrec3} recursively. The \itshape most complicated \upshape harmonic sum originates from the branch of the recursion where always the third integral is used. Looking at the recursions in the Lemmas~\ref{genrec1}, \ref{genrec2} and \ref{genrec3} we see that there are basically two parts of the recursions: the first does not contain the factor $(-1)^i,$ and the second contains this factor. If the harmonic polylogarithm is weighted by $1/(1-x),$ then we are in the first part, otherwise if there is the factor $1/(1+x),$ we are in the second part. We start with the following observation: If we start at the first index and go right, we see that after an index $-1$ we are always in the second part, while we are always in the first part after an index $1.$ An index zero has no effect.\\
First we determine the sign $\sigma$ of the \itshape most complicated \upshape harmonic sum in $\M{\frac{\H{\ve m }x}{1\pm x}}{n}$.

\bfseries Given: \normalfont  a harmonic polylogarithm $\H{\ve m}{n}.$\\
\bfseries Find: \normalfont  the sign $\sigma$ of the \itshape most complicated \upshape harmonic sum $S_{\ve a}(n)$ in $\M{\frac{\H{\ve m }x}{1\pm x}}{n}.$

We start with $\sigma:=-1$ since there is a minus in front of all integrals in the recursions. It depends on which part of the recursion we are how the sign is changed:

\Tree[.{Part 1 ($(-1)^i$ is present)} [.1 {$\sigma \rightarrow \sigma$\\stay} ] [.0 {$\sigma \rightarrow -\sigma$\\stay} ] [.-1 {$\sigma \rightarrow \sigma$\\goto part 2} ] ]
\Tree[.{Part 2 ($(-1)^i$ is not present)} [.1 {$\sigma \rightarrow -\sigma$\\ goto part 1} ] [.0 {$\sigma \rightarrow -\sigma$\\ stay } ] [.-1 {$\sigma \rightarrow -\sigma$\\ stay} ] ]

Analyzing these diagrams we notice that a zero always changes the sign, an other index just changes the sign, if the next nonzero index to the left of it is $-1.$ In addition, the first nonzero index changes the sign if the factor $1/(1+x)$ is present, otherwise if the factor $1/(1-x)$ is present, the sign remains unchanged. Summarizing, it is enough to count the indices less or equal zero. If there are $k$ indices less or equal zero then the sign is:
\begin{eqnarray}
	\sigma &=&\left\{ 
		  	\begin{array}{ll}
						-(-1)^k,\  \textnormal{if } $1/(1-x)$ \textnormal{ is present} & \\
						(-1)^k,\  \textnormal{if } $1/(1+x)$ \textnormal{ is present.} & 
					\end{array} \right. 
\label{polysign}					
\end{eqnarray}
Now we determine the index set of  the most complicated sum:

\bfseries Given: \normalfont a harmonic polylogarithm $\H{\ve m}{n}.$\\
\bfseries Find: \normalfont the index set $\ve a$ of the \itshape most complicated \upshape harmonic sum $S_{\ve a}(n)$ in $\M{\frac{\H{\ve m }x}{1\pm x}}{n}.$

Note that the index set of the harmonic polylogarithm and its \itshape most complicated \upshape harmonic sum are closely related if we use the modified notation (Notation \ref{not1}) of harmonic sums as follows. From the recursions in the Lemmas \ref{genrec1}, \ref{genrec2} and \ref{genrec3} we get that the index set of the harmonic polylogarithm and the index set of the \itshape most complicated \upshape harmonic sum in the Mellin transform of the polylogarithm are almost equal. However, some signs may change and there is an extra index. This extra index arises since the recursions result in a harmonic sum of the form $S_k(n)$, $k=\pm1$  in the base case (just $1/(1-x)$ or $1/(1+x)$ and no harmonic polylogarithm is present). So we append the index 1 to the index set. How are the signs going to change? This depends again on the part of recursion we have to consider. If we are in the first part (of the diagram above) the sign is not changed, if we are in the second part, the sign is changed.\\
So we can proceed as follows: We go from right to the left in $\ve m$. If the next nonzero index to the left is 1, then the actual sign of the index stays the same, otherwise if it is $-1$, the sign changes. The leftmost nonzero index keeps its sign if the factor is $1/(1-x)$, otherwise if the factor $1/(1+x)$ is present, the sign is changed. After this we switch the usual notation of harmonic sums.

\begin{algorithm}
\label{mostcompalg}
\caption{Computation of the Most Complicated Harmonic Sum in the Mellin Transform of a Weighted Harmonic Polylogarithm}
\begin{algorithmic}
\Procedure{MostComplicated}{$\frac{\H{\ve m}x}{1+w*x}$}\Comment{$w=\pm1$}
\State $\ve a=\ve m$
\State Append the index $1$ to $\ve a$
\For{$i=\text{Length}(\ve a)$ to 1}
			\If{$\ve a[i]\neq 0$}
					\State j=i-1;
					\While{$j>0$ and $\ve{a}[j]=0$}
							\State j=j-1;
					\EndWhile
					\If{$j>0$ and $\ve{a}[j]<0$}
							\State $\ve{a}[i]=-\ve{a}[i]$
					\EndIf
					\If{$j=0$ and $w=1$}
							\State $\ve{a}[i]=-\ve{a}[i]$
					\EndIf
			\EndIf		
\EndFor
\State $k=$number of indices $\leq 0$ in $\ve m$
\State \textbf{return} $w (-1)^k S_{\ve a}(n)$	
\EndProcedure
\end{algorithmic}
\end{algorithm}

\begin{example}   
Consider $\H{1,-1,0,1,0,-1}x/(1-x).$ There are $4$ indices less or equal $1$, so the overall sign is $-1$ by (\ref{polysign}).
Start with $\left\{1,-1,0,1,0,-1,1\right\}$ as the index set for the harmonic sum; we added the extra index $1$. Looking at the signs of the indices we arrive at $\left\{1,-1,0,-1,0,-1,-1\right\}.$ So the most complicated harmonic sum in the Mellin transform of $\H{1,-1,0,1,0,-1}x/(1-x)$ is $-S_{1,-1,-2,-2,-1}(n).$
\end{example}

Summarizing, we end up at Algorithm~\ref{mostcompalg} to find the most complicated harmonic sum and its sign in $\M{\frac{\H{\ve m }x}{1\pm x}}{n}.$

\begin{remark}
Due to our considerations we notice that the mapping which maps a weighted polylogarithm onto the \upshape most complicated \itshape harmonic sum in its Mellin transform is injective.
\label{injectiv1}
\end{remark}

\subsubsection{The Reverse Direction}
We are now able to compute the \itshape most complicated \upshape harmonic sum in the Mellin transform of a harmonic polylogarithm. Now we are ready to consider the following problem.

\bfseries Given: \normalfont a harmonic sum $S_{\ve a}(n).$\\
\bfseries Find: \normalfont a harmonic polylogarithm $\H{\ve m}x$ such that $S_{\ve a}(n)$ is the \itshape most complicated \upshape harmonic sum in $\M{\frac{\H{\ve m }x}{1\pm x}}{n}.$

Analyzing Algorithm \ref{mostcompalg}, we notice that the number of negative indices of the \itshape most complicated \upshape harmonic sum in the Mellin transform of  $\H{\ve m}{x}/(1-x)$ is even no matter how $\ve m$ looks like. On the other hand, this number is odd if we consider the Mellin transform of $\H{\ve m}{x}/(1+x)$, again it does not matter how $\ve m$ looks like. 
So we can decide how the polylogarithm has to be weighted in order to get a special harmonic sum, say $S_{\ve a}(n),$ to be the \itshape most complicated \upshape one in the Mellin transform of the weighted polylogarithm. Knowing this, we look at $S_{\ve a}(n)$ in the expanded notation and do the reverse of Algorithm \ref{mostcompalg}: First we drop the last index, then we go from left to right; the first nonzero index changes its sign if the logarithm is weighted by $1/(1+x)$, otherwise if it is weighted by $1/(1-x),$ the first nonzero index holds its sign. The second nonzero index changes its sign if the first nonzero index is negative in the current configuration, otherwise it holds its sign. Proceeding in this manner to the last nonzero index we will get the index set $\ve m$, such that $S_{\ve a}(n)$ is the \itshape most complicated \upshape harmonic sum in the Mellin transform of $\H{\ve m}{x}$ weighted by the already known factor. It remains to decide about the sign of the polylogarithm. In order to do that, we can use (\ref{polysign}) where $k$ is the number of indices less or equal zero in $\ve m.$

Let us summarize our considerations; compare \cite{Remiddi2000}:
Given a multiple harmonic sum $S_{\ve a}(n)$, we want to compute a polylogarithm $\H{\ve m}{x}/(1\pm x)$ such that $S_{\ve a}(n)$ is the \itshape most complicated \upshape harmonic sum in the Mellin transform of $\H{\ve m}{x}/(1\pm x)$. Let us look at the expanded index set of the harmonic sum:
\begin{itemize}
	\item If the number of negative indices is even, set $\sigma=-1$, and there will be a factor $f=1/(1-x);$ otherwise set $\sigma=1,$ there will be the factor $f=1/(1+x)$.
	\item Drop the last index. Take the remaining indices as the indices of the polylogarithm.
	\item The first (leftmost) nonzero index changes its sign if $f=1/(1+x)$, otherwise if $f=1/(1-x),$ the first index holds its sign.
	\item Working from the second nonzero index to the right, each nonzero index will change its sign if the nonzero index to the left is negative, otherwise it holds its sign.
	\item Multiply the term by $\sigma(-1)^k$ where $k$ is the number of indices which are less or equal zero in the current configuration.
\end{itemize}

\begin{remark}
Due to our considerations we notice that the mapping, which maps a harmonic sum $S_{\ve a}(n)$ onto a weighted harmonic polylogarithm $\H{\ve m}x/(1\pm x)$ such that $S_{\ve a}(n)$ is the \upshape most complicated \itshape harmonic sum in the Mellin transform of this $\H{\ve m}x/(1\pm x),$ is injective. Hence, using Remark \ref{injectiv1}, it follows that the mapping which maps a weighted polylogarithm onto the \upshape most complicated \itshape harmonic sum in its Mellin transform is bijective.
\label{injectiv2}
\end{remark}

\subsection{Computation of the Inverse Mellin Transform}
The computation of the inverse Mellin transform of a harmonic sum now is straightforward \cite{Remiddi2000}:
\begin{itemize}
	\item Locate the most complicated harmonic sum.
	\item Construct the corresponding harmonic polylogarithm.
	\item Add it and subtract it.
	\item Perform the Mellin transform to the subtracted version. This will cancel the original harmonic sum.
	\item Repeat the above steps until there are no more harmonic sums.
	\item Let $c$ be the remaining constant term and replace it by $\textnormal{M}^{-1}(c)$, or equivalently, multiply $c$ by $\delta(1-x)$ (see the beginning of this section).
\end{itemize}
A detailed description is given in Algorithm~\ref{invmellalg}.
\begin{algorithm}
\label{invmellalg}
\caption{Inverse Mellin Transform}
\begin{algorithmic}
\Procedure{InvMellin}{$S_{\ve a}(n),x$}\Comment{$x$ is argument of output function}
\State $\ve b=\ve a$
\State $i=$number of negative indices in $\ve b$
\If{$i$ is even}
\State $f=1/(1-x)$
\State $\sigma=-1$
\Else
\State $f=1/(1+x)$
\State $\sigma=1$
\State $\ve{b}[1]=-\ve{b}[1]$\Comment{first index of $\ve b$ changes sign}
\EndIf
\For{$i=2$ to Length($\ve{b}$)}
			\If{$\ve{b}[i-1]<0$}
			\State $\ve{b}[i]=-\ve{b}[i]$
			\EndIf
\EndFor
\State transform $S_{\ve b}(n)$ to its expanded notation\Comment{See \ref{not1}}
\State drop the last index of $\ve b$
\State $k=$number of indices $\leq 0$ in $\ve a$
\State $m=S_{\ve a}(n)-\sigma(-1)^k\Call{Mellin}{f\H{\ve a}{x},n}$
\While{$M$ contains a harmonic sum}
\State $s=$most complicated harmonic sum in $M$\Comment{Use \ref{sord}}
\State $M=M-\Call{InvMellin}{s,x}$
\EndWhile
\State multiply the constant term in $M$ by $\delta(1-x)$
\State \textbf{return} $M$	
\EndProcedure
\end{algorithmic}
\end{algorithm}

\begin{example}
Consider the sum $S_{1,1,2,1}(n)$, which is $S_{1,1,0,1,1}(n)$ in the expanded index set notation (see Notation \ref{not1}). This sum is the most complicated harmonic sum in the Mellin transform of $\H{1,1,0,1}x/(1-x).$ The Mellin transform of $\H{1,1,0,1}x/(1-x)$ yields:
$$
-4 \zeta_5 - \zeta_2 S_{1, 1, 1}(n) + S_{1, 1, 2, 1}(n).
$$
Therefore after adding and subtracting we arrive at:
$$
\frac{\H{1,1,0,1}x}{1-x}+4 \zeta_5 + \zeta_2 S_{1, 1, 1}(n).
$$
Now we have to consider $S_{1, 1, 1}(n)$. This sum is the most complicated harmonic sum in the Mellin transform of $-\H{1,1}x/(1-x).$ The Mellin transform of $-\H{1,1}x/(1-x)$ is $S_{1, 1, 1}(n).$ We get:
$$
\frac{\H{1,1,0,1}x}{1-x}+4 \zeta_5 - \zeta_2 \frac{\H{1,1}x}{1-x}.
$$
There is no harmonic sum left. So we multiply the constant term by $\delta(1-x)$ and arrive at the inverse Mellin transform of $S_{1,1,2,1}(n):$
$$
\frac{\H{1,1,0,1}x}{1-x}- \zeta_2 \frac{\H{1,1}x}{1-x}+4 \zeta_5 \delta(1-x).
$$
\end{example}

With \ttfamily HarmonicSums \rmfamily this can be carried out as follows:
\begin{fmma}
\begin{mma}
\In \text{\bf InvMellin[S[1, 1, 2, 1, n], x]}\\
\Out {4\,\text{z5}\,\text{Delta1x}+\frac{\text{z2}\,\text{H}[1, 1, x]}{-1+x}-\frac{\text{H}[1, 1, 0, 1, x]}{-1 + x}}\\
\end{mma}
\begin{mma}
\In \text{\bf $\%$ // ReplaceByKnownFunctions}\\
\Out {4\,\text{Delta1x}\,\text{z5} + \frac{\text{z2}\,\text{H}[1,1,x]}{-1 + x} - 
  \frac{\text{H}[1,1,0,1,x]}{-1 + x}}\\
\end{mma}
\begin{mma}
\In \text{\bf InvMellin[S[2, -1, -1, n], x]}\\
\Out {\frac{(-1)^n \text{H}[-1, 1]^2 \text{H}[0, x] + (-1)^n \text{H}[-1, 1] \text{H}[0, -1, 1] - (-1)^n \text{H}[-1, 
            1] \text{H}[0, -1, x]}{1 + x} + \frac{1}{1 - x}(\text{H}[-1, 1]^2 \text{H}[0, x] - 
        \text{H}[-1, 1] \text{H}[0, 1, 1] + \text{H}[-1, 1] \text{H}[0, 1, x] + \text{H}[0, 1, -1, 1] - 
        \text{H}[0, 1, -1, x]) + 
  \text{Delta1x} (-\text{H}[-1, 1]^2 \text{H}[-1, 0, 1] - \text{H}[-1, 1]^2 \text{H}[0, -1, 1] + 
        \text{H}[-1, 1] \text{H}[1, 1] \text{H}[0, 1, 1] - \text{H}[-1, 1]^2 \text{H}[1, 0, 1] + 
        \text{H}[-1, 1] \text{H}[-1, 0, -1, 1] - \text{H}[1, 1] \text{H}[0, 1, -1, 1] - 
        \text{H}[-1, 1] \text{H}[1, 0, 1, 1] + \text{H}[1, 0, 1, -1, 1])}\\
\end{mma}
\begin{mma}
\In \text{\bf $\%$ // ReplaceByKnownFunctions}\\
\Out {\text{Delta1x}\,\text{li4half} + \frac{\text{Delta1x}\,{\text{ln2}}^4}{24} + 
  \frac{5\,\text{Delta1x}\,{\text{ln2}}^2\,\text{z2}}{4} + 
  \frac{\text{ln2}\,\text{z2}}{2 - 2\,x} + 
  \frac{{\left( -1 \right) }^n\,\text{ln2}\,\text{z2}}{2\,\left( 1 + x \right) } + 
  \frac{\text{Delta1x}\,{\text{z2}}^2}{40} + \frac{\text{z3}}{-1 + x} + 
  {\text{ln2}}^2\,\left( \frac{1}{1 - x} + \frac{{\left( -1 \right) }^n}{1 + x} \right) \,\text{H}[0,x] - 
  \frac{{\left( -1 \right) }^n\,\text{ln2}\,\text{H}[0,-1,x]}{1 + x} + 
  \frac{\text{ln2}\,\text{H}[0,1,x]}{1 - x} + \frac{\text{H}[0,1,-1,x]}{-1 + x}}\\
\end{mma}
\end{fmma}

\section{Differentiation of Multiple Harmonic Sums}
We are now able to calculate the inverse Mellin transform of harmonic sums. It turned out that they are usually linear combinations of harmonic polylogarithms weighted by the factors $1/(1\pm x)$ and they can be distribution-valued. By computing the Mellin transform of a harmonic sum we find in fact an analytic continuation of the sum to $n\in\R.$ For explicitly given analytic continuations see \cite{Bluemlein1999,Bluemlein2000,Bluemlein2004,Bluemlein2005,Bluemlein2008}. The differentiation of harmonic sums in the physic literature has been considered the first time in \cite{Bluemlein1999}. Worked out in \cite{Bluemlein2008,Bluemlein2009,Bluemlein2009a} this allows us to consider differentiation with respect to $n$, since we can differentiate the analytic continuation. Afterwards we may transform back to harmonic sums with the Mellin transform. Differentiation turns out to be relatively easy if we represent the harmonic sum using its inverse Mellin transform as the following example suggests (see \cite[p.81]{Paris2001}):
\begin{example}
$$\frac{d}{d n}\int_0^1{x^n\H{-1}{x}dx}=\int_0^1{x^n\log{(x)}\H{-1}{x}dx}=\int_0^1{x^n\H{-1,0}{x}dx}+\int_0^1{x^n\H{0,-1}{x}dx}.$$
\end{example}

Based on this example, if we want to differentiate $S_\ve a(n)$ with respect to $n$ we can proceed as follows:
\begin{itemize}
	\item Calculate the inverse Mellin transform of $S_\ve a(n).$
	\item Set the constants to zero and multiply the remaining terms of the inverse Mellin transform by $\H0x$. This is in fact 									differentiation with respect to	$n$. 
	\item Calculate the Mellin transform of the multiplied inverse Mellin transform of $S_\ve a(n).$ 
\end{itemize}

\begin{example}
Let us differentiate $S_{2,1}(n).$ The inverse Mellin transform is:
$$
\frac{\H{0,1}x-\zeta_2}{1-x}+\textnormal{M}^{-1}(2\zeta_3).
$$

Hence we have:
\begin{eqnarray*}
S_{2,1}(n)&=&\int_0^1{x^n\frac{\H{0,1}x-\zeta_2}{1-x}dx}+\int_0^1{x^n\textnormal{M}^{-1}(2\zeta_3)dx}\\
		      &\stackrel{\H{0,1}1=\zeta_2}{=}&\int_0^1{\frac{x^n\H{0,1}x-\H{0,1}1}{1-x}dx}+\int_0^1{\frac{x^n\zeta_2-\zeta_2}{1-x}dx}
		      		+\int_0^1{x^n\textnormal{M}^{-1}(2\zeta_3)dx}\\
					&=&\M{\frac{\H{0,1}x}{1-x}}{n}-\M{\frac{\zeta_2}{1-x}}{n}+2\zeta_3.\nonumber
\end{eqnarray*}
Differentiating the right hand side with respect to $n$ yields:
\begin{eqnarray*}
&& \int_0^1{x^n\log{(x)}\frac{\H{0,1}x-\zeta_2}{1-x}dx}=\int_0^1{x^n\H{0}x\frac{\H{0,1}x-\zeta_2}{1-x}dx}\\
 &&\ \ \ = \int_0^1{2 x^n\frac{\H{0,0,1}x}{1-x}+x^n\frac{\H{0,1,0}x}{1-x}-\zeta_2x^n\frac{\H{0}x}{1-x}dx}\\
 &&\ \ \ = \int_0^1{2 x^n\frac{\H{0,0,1}x-\H{0,0,1}1}{1-x}+x^n\frac{\H{0,1,0}x-\H{0,1,0}1}{1-x}-\zeta_2x^n\frac{\H{0}x-\H{0}1}{1-x}dx}\\
 &&\ \ \ = 2\int_0^1{\frac{x^n\H{0,0,1}x-\H{0,0,1}1}{1-x}dx}+\int_0^1{\frac{x^n\H{0,1,0}x-\H{0,1,0}1}{1-x}dx}\\
 &&\ \ \ \ \ \ -\zeta_2 \int_0^1{\frac{x^n\H{0}x-\H{0}1}{1-x}dx}\\
 &&\ \ \ = 2\M{\frac{\H{0,0,1}x}{1-x}}{n}+\M{\frac{\H{0,1,0}x}{1-x}}{n}-\zeta_2 \M{\frac{\H{0}x}{1-x}}{n}\\
 &&\ \ \ = \frac{7\zeta_2^2}{10}+\zeta_2S_2(n)-S_{2,2}(n)-2S_{3,1}(n).
\end{eqnarray*}
\end{example}

\begin{algorithm}
\caption{Differentiation of harmonic sums}
\begin{algorithmic}
\Procedure{DifferentiateHarmonicSum}{$S_{\ve a}(n)$}
\State $h=\Call{InvMellin}{S_{\ve a}(n),x}$
\State replace $\delta(1-x)$ in $h$ by $0$
\State $h=\H{0}{x}*h$
\State expand all products in $h$\Comment{Use (\ref{hpro})}
\State $nsum=$number of summands in $h$
\For{$i=1$ to $nsum$}
\State replace $i$th summand of $h$ by $\Call{Mellin}{i\text{th summand},n}$
\State \textbf{return} $h$				
\EndFor
\EndProcedure
\end{algorithmic}
\end{algorithm}

The package \ttfamily HarmonicSums \rmfamily provides a tool to differentiate harmonic sums:
\begin{fmma}
\begin{mma}
\In \text{\bf DifferentiateSSum[S[2, 1, n], n] // ReplaceByKnownFunctions}\\
\Out {\frac{7\,{\text{z2}}^2}{10} + \text{z2}\,\text{S}[2,n] - \text{S}[2,2,n] - 
  2\,\text{S}[3,1,n]}\\
\end{mma} 
\begin{mma}
\In \text{\bf DifferentiateSSum[S[5, n], n] // ReplaceByKnownFunctions}\\
\Out {\frac{8\,{\text{z2}}^3}{7} - 5\,\text{S}[6,n]}\\
\end{mma}
\begin{mma}
\In \text{\bf DifferentiateSSum[S[2, -1, 2, n], n] // ReplaceByKnownFunctions // Expand}\\
\Out {4\,\text{s6} + 4\,\text{li4half}\,\text{z2} + \frac{{\text{ln2}}^4\,\text{z2}}{6} - 
  {\text{ln2}}^2\,{\text{z2}}^2 - \frac{121\,{\text{z2}}^3}{120} - \frac{11\,{\text{z3}}^2}{4} - 
  4\,\text{li4half}\,\text{S}[2,n] - \frac{{\text{ln2}}^4\,\text{S}[2,n]}{6} + 
  {\text{ln2}}^2\,\text{z2}\,\text{S}[2,n] + \frac{13\,{\text{z2}}^2\,\text{S}[2,n]}{40} + 
  2\,\text{z3}\,\text{S}[2,-1,n] - \text{S}[2,-2,2,n] - 2\,\text{S}[2,-1,3,n] - 
  2\,\text{S}[3,-1,2,n]}\\
\end{mma}  
\end{fmma}

\subsection{Application of Differentiation}
\label{appdiff}
As shown in \cite{Bluemlein2008,Bluemlein2009,Bluemlein2009a} new relations between multiple harmonic sums arise if we introduce differentiation on harmonic sums. We can now use these relations to reduce the number of basic sums we computed in the first strategy of Section \ref{Application of the Relations} by allowing relations due to differentiation. Unfortunately we cannot adopt the second strategy since we have to know the actual values of the indices of a harmonic sum in order to differentiate it. Let us look at an example.
\begin{example}
Using only algebraic relations we have to take $S_{1}(n)$ and $S_{2}(n)$ into the basis, but using in addition the differential operator $\frac{d}{d n}$ we get the following relation:
$$\frac{d}{d n}S_1(n)=\zeta_2-S_2(n).$$
So we can express $S_2(n)$ with the help of $S_1(n)$ and we can eliminate $S_2(n)$ from the basis. In general we can express all sums of depth one by $\frac{d^i}{dn^i}S_{\pm 1}(n)$ with $i\geq 0;$ see \cite{Bluemlein2008}.
\label{diffrelexp}
\end{example}
As indicated in the previous example for $S_2(n)$ we proceed as follows. First we fix an algebraic basis in which we can represent our expression (see the first strategy in Section~\ref{Application of the Relations}).
For each multiple harmonic sum $S_{\ve a}(n)$ in the algebraic basis we try to find another multiple harmonic sum $S_{\ve b}(n)$ such that $S_{\ve b}(n)$ is less complicated than $S_{\ve a}(n)$ and such that $S_{\ve a}(n)$ is a \itshape most complicated \upshape harmonic sum in $\frac{d}{d n}S_{\ve b}(n)$. We can do this by computing the inverse Mellin transform of $S_{\ve a}(n)=S_{a_1,a_2,\ldots,a_k}(n)$. There we look for the \itshape most complicated \upshape harmonic polylogarithm. 
We call a harmonic polylogarithm \itshape most complicated \upshape if it is the polylogarithm with the largest weight or in the case of identical weights it has the largest number of nonzero indices \cite{Remiddi2000}. 
If at least one of the indices of the \itshape most complicated \upshape harmonic polylogarithm $\H{\ve m}{x}$ is zero (so if $\H{\ve m}{x}=\H{\ve{m}_1,0,\ve{m}_2}{x}$), then we can find a harmonic sum $S_{\ve b}(n)$ with the desired property, since we differentiate by multiplying $\H{0}x$. The harmonic sum $S_{\ve b}(n)$ is just the \itshape most complicated \upshape one in the Mellin transform of $\H{\ve{m}_1,\ve{m}_2}x$ weighted by the same factor as $\H{\ve m}x$ in the inverse Mellin transform of $S_{\ve a}(n).$ Hence, due to the consideration in Section \ref{Inverse Mellin Transform} it is clear that the found harmonic sum satisfies the desired property. In fact, if there is an index $a_l$ of $S_{\ve a}(n)$ which is not equal to $\pm 1,$ we will find $S_{\ve b}(n)$ with $\ve b=(a_{1},\ldots,,a_{l-1},\sign{a_l}(\abs{a_l}-1),a_{l+1},\ldots,a_k).$
\begin{example}
We consider the sum $S_{4,2}(n)$. We can use both $S_{3,2}(n)$ or $S_{4,1}(n)$ to get relations for $S_{4,2}(n)$:
\begin{eqnarray*}
\frac{d}{d n}S_{3,2}(n)&=&2\zeta_3^2+2\zeta_3S_{3}(n)-2S_{3,1}(n)-3S_{4,2}(n)\\
\frac{d}{d n}S_{4,1}(n)&=&\frac{118\zeta_2^3}{105}-\zeta_3^2+\zeta_2S_{4}(n)-S_{4,2}(n)-4S_{5,1}(n)
\end{eqnarray*}
\end{example}

\begin{table}
\begin{tabular}{|| r || r | r | r ||}
\hline	
&  \multicolumn{3}{|c||}{Number of} \\
\cline{2-4}
Weight& Sums& a-basic sums& d-basic sums\\
\hline	
  1 &    2 &   2 &   2 \\
  2 &    6 &   3 &   1 \\
  3 &   18 &   8 &   5 \\
  4 &   54 &  18 &  10 \\
  5 &  162 &  48 &  30 \\ 
  6 &  486 & 116 &  68 \\
\hline
\end{tabular}
\caption{Number of harmonic sums, and the respective numbers of basic sums by which all sums can be expressed using the algebraic (a-basic) or algebraic and differential relations (d-basic) in dependence on their weight; see \cite{Bluemlein2004,Bluemlein2008,Bluemlein2009,Bluemlein2009a}.}
\label{tad}
\end{table}

We use all these new relations to reduce the number of sums in the basis as far as possible. Summarizing, $S_{\ve b}(n)$ is less complicated than $S_{\ve a}(n)$ (which we want to eliminate) and all additionally introduced sums are not more complicated than $S_{\ve a}(n).$ In this way we can successively eliminate the original given algebraic basis to a basis whose sums are algebraically independent and where as much sums as possible can be eliminated by the differential operator. In Table \ref{tad} we can see how the number of basic sums reduces in comparison to the use of algebraic relations only; compare \cite{Bluemlein2009,Bluemlein2009a}. For the package \ttfamily HarmonicSums \rmfamily tables up to weight 6 are available.

\begin{fmma}
\In \text{\bf ReduceToBasis[S[1, 2, -2, n], UseDifferentiation $\rightarrow$ True]}\\
\Out {\frac{7\,\text{z2}\,\text{z3}}{12} - \frac{5\,\text{z5}}{24} - \frac{2\,\text{z2}\,\text{Diff}[\text{S}[-1,n],n,2]}{3} - \frac{\text{Diff}[\text{S}[-1,n],n,4]}{72} + 
  \frac{\text{Diff}[\text{S}[-1,n],n,2]\,\text{Diff}[\text{S}[1,n],n,1]}{3} - \frac{\text{z2}\,\text{Diff}[\text{S}[1,n],n,2]}{12} - 
  \frac{\text{Diff}[\text{S}[-1,n],n,1]\,\text{Diff}[\text{S}[1,n],n,2]}{6}  - \frac{\text{Diff}[\text{S}[-2,1,n],n,2]}{6} - \frac{\text{Diff}[\text{S}[2,-2,n],n,1]}{3} + \frac{4\,\text{S}[-3,2,n]}{3} + 
  \text{S}[1,n]\,\text{S}[2,-2,n] + \frac{\text{z2}\,\text{S}[2,1,n]}{2} + \text{Diff}[\text{S}[-1,n],n,1]\,\text{S}[2,1,n] + \text{S}[-2,2,1,n]}\\
\end{fmma}  
Here $\text{Diff}[f(x),x,i]:=\frac{d^if(x)}{dx^i}.$

%% file: halfint.tex
\chapter{Half-Integer Relations}
\label{Half-Integer Relations}
\setcounter{section}{1}
\setcounter{thm}{0}
In this chapter we allow besides the upper index $n$, the index $\frac{n}{2}$ (or equivalently the index $2n$) in the multiple harmonic sums. As it turns out, this leads to additional relations by using the operator $\textnormal{Half}[f(x),x]=f(\frac{x}{2}).$ Such relations have been determined in \cite{Bluemlein2008,Bluemlein2009,Bluemlein2009a,Bluemlein1999} by exploiting properties of the Mellin transform. Here we follow a different approach. We start with the following lemma, it has already been used in \cite{Bluemlein1999}. 
\begin{lemma}
Let $n,m,a\in\N$. Then the following relation holds:
\begin{equation}
S_{a}(2n)-S_{-a}(2n)=\frac{1}{2^{a-1}}S_a(n).
\end{equation}
\label{halfintlem}   
\end{lemma} 
\begin{proof}For $n\in \N,$
\begin{eqnarray}
S_{-a}(2n)&=&\sum_{i=1}^{2n}{\frac{(-1)^i}{i^a}}=\sum_{i=1}^{n}{\left(\frac{1}{(2i)^a}-\frac{1}{(2i-1)^a}\right)} 																																	=\frac{1}{2^a}S_a(n)-\sum_{i=1}^n{\frac{1}{(2i-1)^a}}\nonumber\\
					&=&\frac{1}{2^a}S_a(n)-\sum_{i=1}^n{\left(\frac{1}{(2i)^a}+\frac{1}{(2i-1)^a}-\frac{1}{(2i)^a}\right)}\nonumber\\ 																											&=&\frac{1}{2^a}S_a(n)-S_a(2n)+\frac{1}{2^k}S_a(n)\nonumber\\
					&=&\frac{1}{2^{a-1}}S_a(n)-S_a(2n).\nonumber
\end{eqnarray}
\end{proof}

\begin{thm}(see \cite{Vermaseren1998})
Let $n,m\in\N$ and $a_i \in \N$ for $i \in \N.$ Then we have the following relation:
\begin{eqnarray}
\sum{S_{\pm a_m, \pm a_{m-1},\ldots,\pm a_1}(2n)}=\frac{1}{2^{\sum_{i=1}^m a_i-m}}S_{a_m,a_{m-1},\ldots,a_1}(n)
\label{halfint}
\end{eqnarray}
where we sum on the left hand side over the $2^m$ possible combinations concerning $\pm$.
\end{thm}

\begin{proof}
We proceed by induction on the depth $m$ of the multiple harmonic sums. In Lemma \ref{halfintlem} we have already proven the case $m=1.$ Let's assume (\ref{halfint})  holds for $m$:
\begin{eqnarray}
\sum{S_{\pm a_{m+1}, \pm a_m, \ldots,\pm a_1}(2n)}
&=&
\sum_{i=1}^n\biggl(\frac{1}{(2i)^{a_{n+1}}}\sum{S{\pm a_m, \pm a_{m-1},\ldots,\pm a_1}(2n)}\biggr.\nonumber\\
&&+\biggl.\frac{1}{^(2i-1)^{a_{n+1}}}\sum{S_{\pm a_m, \pm a_{m-1},\ldots,\pm a_1}(2n)}\biggr)\nonumber\\
&&+\sum_{i=1}^n\biggl(\frac{(-1)^{2i}}{(2i)^{a_{n+1}}}\sum{S{\pm a_m, \pm a_{m-1},\ldots,\pm a_1}(2n)}\biggr.\nonumber\\
&&+\biggl.\frac{(-1)^{2i+1}}{^(2i-1)^{a_{n+1}}}\sum{S_{\pm a_m, \pm a_{m-1},\ldots,\pm a_1}(2n)}\biggr)\nonumber\\
&=&
2\sum_{i=1}^n\frac{1}{(2i)^{a_{n+1}}}\sum{S{\pm a_m, \pm a_{m-1},\ldots,\pm a_1}(2n)}\nonumber\\
&=&
2\sum_{i=1}^n\frac{1}{(2i)^{a_{n+1}}}\frac{1}{2^{\sum_{i=1}^m a_i-m}}S_{a_m,a_{m-1},\ldots,a_1}(n)\nonumber\\
&=&
\frac{1}{2^{\sum_{i=1}^m a_i-m}}\frac{1}{2^{a_{m+1}-1}}S_{a_{m+1},a_m,\ldots,a_1}(n)\nonumber\\
&=&
\frac{1}{2^{\sum_{i=1}^m a_i-(m+1)}}S_{a_{m+1},a_m,\ldots,a_1}(n).\nonumber
\end{eqnarray}
Hence (\ref{halfint}) holds for $m+1.$ 
\end{proof}

\begin{example} For $n\in \N,$
$$
\frac{1}{8}S_{2,3}(n)=S_{2,3}(2n)+S_{2,-3}(2n)+S_{-2,3}(2n)+S_{-2,-3}(2n).
$$
\end{example}

Similar to Section \ref{appdiff} we use all these new relations to reduce the number of sums in the basis as far as possible. In the Tables \ref{tadh1} and \ref{tadh2} we can see how the number of basic sums reduces if we use algebraic relations and half-integer relations in comparison to the use of algebraic relations only. If we use all three types of relations, namely algebraic, differential and half-integer relations, the number of basic sums reduces further as we can see in the Tables \ref{tadh1} and \ref{tadh2}.

\begin{example}
From Example \ref{diffrelexp} we know that all harmonic sums of depth one can be expressed by the two sums $S_{1}(n)$ and $S_{-1}(n)$ by using differentiation. From Theorem \ref{halfint} we get:
$$
S_{-1}(2n) = S_{1}(2n) - S_1(n).
$$
Hence we can express $S_{-1}(n)$ by $S_1(n)$ and $S_{1}(n/2),$ and we can represent all harmonic sums of depth one by $S_1(n)$ using differentiation and half-integer relations.  
\end{example}

\begin{table}
\begin{tabular}{|| r || r | r | r | r | r ||}
\hline	
&  \multicolumn{5}{|c||}{Number of} \\
\cline{2-6}
Weight& Sums& a-basic sums& d-basic sums& h-basic sums& dh-basic sums\\
\hline	
  1 &    2 &   2 &   2 &   1  &  1 \\
  2 &    6 &   3 &   1 &   2  &  1 \\
  3 &   18 &   8 &   5 &   6  &  4 \\
  4 &   54 &  18 &  10 &  15  &  9 \\
  5 &  162 &  48 &  30 &  42  & 27 \\ 
  6 &  486 & 116 &  68 & 107  & 65 \\
\hline
\end{tabular}
\caption{Number of harmonic sums, and the respective numbers of basic sums by which all sums can be expressed using algebraic (a-basic sums), differential (d-basic sums) and/or half-integer relations (h-basic/ dh-sums); see \cite{Bluemlein2004,Bluemlein2008,Bluemlein2009,Bluemlein2009a}.}
\label{tadh1}
\end{table}

\begin{table}
\begin{tabular}{|| r || r | r | r | r ||}
\hline	
&  \multicolumn{4}{|c||}{Number of} \\
\cline{2-5}
Weight& Sums $\neg \left\{-1\right\}$ & a-basic sums& d-basic sums&dh-basic sums\\
\hline	
  1 &    1 &   1 &   1 &      1 \\
  2 &    3 &   2 &   1 &      0 \\
  3 &    7 &   4 &   2 &      2 \\
  4 &   17 &   7 &   3 &      3 \\
  5 &   41 &  16 &  10 &      9 \\ 
  6 &   99 &  30 &  17 &     17 \\
\hline
\end{tabular}
\caption{Number of sums, which do not contain the index $\left\{-1\right\}$ and the respective numbers of basic sums by which all sums can be expressed; see \cite{Bluemlein2004,Bluemlein2008,Bluemlein2009,Bluemlein2009a}.}
\label{tadh2}
\end{table}

For the package \ttfamily HarmonicSums \rmfamily tables using algebraic and half-integer relations up to weight 6 are available. We remark that all sums (d-, h-, dh-basic sums) in the Tables \ref{tadh1} and \ref{tadh2} are algebraically independent by construction. In \cite{Bluemlein2009} the representations of the dh-basic sums (up to weight 5) have been established with a different approach.

\begin{fmma}
\In \text{\bf ReduceToBasis[S[1, 2, -2, n], UseHalfInteger $\rightarrow$ True]}\\
\Out {\text{S}[1,n]\,\left( \frac{\text{Half}[\text{S}[2,n],n]}{2} - \text{S}[2,n] \right) \,\text{S}[2,n] + 
  \text{S}[1,n]\,\left( \frac{\text{Half}[\text{S}[4,n],n]}{8} - \text{S}[4,n] \right)  - \text{S}[-4,1,n] - 
  \text{S}[1,n]\,\left( \frac{\text{Half}[\text{S}[4,n],n]}{8} + \left( \frac{\text{Half}[\text{S}[2,n],n]}{2} - \text{S}[2,n] \right) \,\text{S}[2,n] - \text{S}[4,n] - 
     \text{S}[2,-2,n] \right)  - \left( \frac{\text{Half}[\text{S}[2,n],n]}{2} - \text{S}[2,n] \right) \,\text{S}[2,1,n] + \text{S}[3,-2,n] + \text{S}[-2,2,1,n]}\\
\end{fmma} 

Moreover tables using all three types of relations up to weight 6 are available.
\begin{fmma}
\In \text{\bf ReduceToBasis[S[1, 2, -2, n], UseDifferentiation $\rightarrow$ True, UseHalfInteger $\rightarrow$ True]}\\
\Out {\frac{7\,\text{z2}\,\text{z3}}{12} - \frac{5\,\text{z5}}{24} - 
  \frac{2\,\text{z2}\,\text{Diff}[\text{Half}[\text{S}[1,n],n] - \text{S}[1,n],n,2]}{3} - 
  \frac{\text{Diff}[\text{Half}[\text{S}[1,n],n] - \text{S}[1,n],n,4)}{72} + 
  \frac{\text{Diff}[\text{Half}[\text{S}[1,n],n] - \text{S}[1,n],n,2]\,\text{Diff}[\text{S}[1,n],n,1]}{3} - 
  \frac{\text{z2}\,\text{Diff}[\text{S}[1,n],n,2]}{12} - \\
  \frac{\text{Diff}[\text{Half}[\text{S}[1,n],n] - \text{S}[1,n],n,1]\,\text{Diff}[\text{S}[1,n],n,2]}{6} - 
  \frac{\text{Diff}[\text{S}[-2,1,n],n,2]}{6} - \frac{\text{Diff}[\text{S}[2,-2,n],n,1]}{3} + \frac{4\,\text{S}[-3,2,n]}{3} + \text{S}[1,n]\,\text{S}[2,-2,n] + 
  \frac{\text{z2}\,\text{S}[2,1,n]}{2} + \text{Diff}[\text{Half}[\text{S}[1,n],n] - \text{S}[1,n],n,1]\,\text{S}[2,1,n] + \text{S}[-2,2,1,n]}\\
\end{fmma}   
Further investigations in this direction in connection with Chapter \ref{An Example from Particle Physics} are in progress~\cite{Ablinger2009}.

%% file: summation.tex
\chapter{Summation of Multiple Harmonic Sums}
\label{Summation of Multiple Harmonic Sums}

With the package \ttfamily Sigma \rmfamily \cite{Schneider2007} we can simplify nested sums such that the nested depth is optimal and the degree of the denominator is minimal. This is possible due to a refined summation theory \cite{Schneider2007a,Schneider2008a} of $\Pi \Sigma$-fields \cite{Karr1981}. In the following we want to find sum representations of such simplified nested sums in terms of harmonic sums as much as it is possible. Inspired by \cite{Savio} for harmonic numbers and \cite{Moch2002,Vermaseren1998} we consider sums of the form
\begin{equation}
	\sum_{i=1}^n{s\,r(i)}
	\label{sumsum0}
\end{equation}
where $s\in \mathcal{S}(i)$ or $s\in \mathcal{S}(i)[(-1)^i]$ and $r(i)$ is a rational function in $i.$ We can use the quasi-shuffle algebra property of the multiple harmonic sums to split such sums into sums of the form
\begin{equation}
	\sum_{i=1}^n{r(i)S_{\ve a}(i)}
	\label{sumsum1}
\end{equation}
or
\begin{equation}
	\sum_{i=1}^n{(-1)^ir(i)S_{\ve a}(i)}
	\label{sumsum2}
\end{equation}
where $r(i)$ is a rational function in $i.$\\

In this chapter we will present new formulas to rewrite sums of the form (\ref{sumsum1}) and (\ref{sumsum2}) and we will show how we can use these formulas in combination with \ttfamily Sigma\rmfamily. In the following we consider the problem:\\
\bfseries Given: \normalfont a sum $\sigma$ of the form (\ref{sumsum0}).\\
\bfseries Find: \normalfont as much as possible a representation of $\sigma$ in terms of harmonic sums.

\section{Polynomials in the Summand}
\label{polyinthesum}
First we consider the case that $r(i)$ is a polynomial in $\Z[i],$ \ie $r(i)=p_mi^m+\ldots+p_1i+p_0$ with $p_k \in \Z.$ Here we generalize the harmonic sum case from depth 1 \cite{Savio} to arbitrary depth. If we are able to work out the sum for any power of $i$ times a multiple harmonic sum (\ie $\sum_{i=1}^n i^m S_\ve a(i)$ or $\sum_{i=1}^n (-1)^ii^m S_\ve a(i)$, $m \in \N$), we can work out the sum for each polynomial $r(i)$. Let us start with $m=0$ where the factor $(-1)^i$ is present.
 
\begin{thm}
\label{sumtheo1}
Let $n \in \N$ and $a_k \in \Z/\{0\}.$ Then
\begin{equation}
\sum_{i=1}^n{(-1)^iS_{a_1,a_2,\ldots}(i)}=\frac{1}{2}\left(S_{-a_1,a_2,\ldots}(n)+(-1)^nS_{a_1,a_2,\ldots}(n)\right).
\end{equation}
\end{thm}

\begin{proof}For $n \in \N$
\begin{eqnarray*}
\sum_{i=1}^n{(-1)^iS_{a_1,a_2,\ldots}(i)}&=&\sum_{i=1}^n{(-1)^i\sum_{j=1}^i{\frac{\sign{a_1}^j}{j^{\abs{a_1}}} S_{a_2,a_3\ldots}(j)}}\\
&=&\sum_{j=1}^n{\frac{\sign{a_1}^j}{j^{\abs{a_1}}} S_{a_2,a_3\ldots}(j) \sum_{i=j}^n{(-1)^i}}\\
&=&\sum_{j=1}^n{\frac{\sign{a_1}^j}{j^{\abs{a_1}}} S_{a_2,a_3\ldots}(j) \frac{(-1)^j+(-1)^n}{2}}\\
&=&\frac{1}{2}\left(S_{-a_1,a_2,\ldots}(n)+(-1)^nS_{a_1,a_2,\ldots}(n)\right).
\end{eqnarray*}
\end{proof}

When $m=0$ and the factor $(-1)^i$ is not present we have to distinguish between several cases. In order to find this identities C. Schneider's \ttfamily Sigma \rmfamily package played a decisive role.

\begin{thm}
\label{sumtheo2}
Let $k, n, l \in \N$ and $a_j \in \Z/\{0\}.$ Then $\sum_{i=1}^n{S_{a_1,a_2,\ldots,a_l}(i)}$ can be simplified as follows.
\begin{description}
	\item [$\abs{a_1}\geq 2:$]
				\begin{equation}
			\sum_{i=1}^n{S_{a_1,a_2,\ldots,a_l}(i)}=(n+1)S_{a_1,a_2,\ldots,a_l}(n)-S_{\sign{a_1}(\abs{a_1}-1),a_2,\ldots,a_l}(n);
				\nonumber
				\end{equation}
	\item [$a_1=-1:$]
				\begin{equation}
			\sum_{i=1}^n{S_{-1,a_2,a_3,\ldots,a_l}(i)}= 																																																(n+1)S_{-1,a_2,a_3,\ldots,a_l}(n)-\frac{1}{2}\left(S_{-a_2,a_3,\ldots,a_l}(n)+(-1)^nS_{a_2,a_3,\ldots,a_l}(n)\right);\nonumber
				\end{equation}
		\item [$a_j=1, \ 1\leq j\leq l:$]
				\begin{eqnarray}
					\sum_{i=1}^n{S_{\scriptsize{{\underbrace{1,1,\ldots,1}_{l \textnormal{ times}}}}}(i)}&=& 																															(n+1)\Bigl(S_{\scriptsize{\underbrace{1,1,\ldots,1}_{l \textnormal{ times}}}}(n) 																												-S_{\scriptsize{\underbrace{1,1,\ldots,1}_{l-1\textnormal{ times}}}}(n) 																																+\cdots+(-1)^{l-1}S_{1}(n)+(-1)^{l}\Bigr)\nonumber\\
												&&+(-1)^{l+1};\nonumber
				\end{eqnarray}
		\item [$a_j=1, \ 1\leq j\leq l, a_{l+1}=-1:$]
				\begin{eqnarray}
					\sum_{i=1}^n{S_{\scriptsize{\underbrace{1,1,\ldots,1}_{l \textnormal{ times}}},-1}(i)}&=& 																														(n+1)\Bigl(S_{\scriptsize{\underbrace{1,1,\ldots,1}_{l \textnormal{ times}}},-1}(n) 																										-S_{\scriptsize{\underbrace{1,1,\ldots,1}_{l-1\textnormal{ times}}},-1}(n) 																															+\cdots+(-1)^{l}S_{-1}(n)\Bigr)\nonumber\\
											&&+\frac{(-1)^{l}}{2}\left(1-(-1)^n\right);\nonumber
				\end{eqnarray}				
		\item [$a_j=1, \ 1\leq j\leq l, a_{l+1}=-1:$]
				\begin{eqnarray}
					\sum_{i=1}^n{S_{\scriptsize{\underbrace{1,1,\ldots,1}_{l \textnormal{ times}}},-1,a_{l+2},\ldots}(i)}&=& 																						(n+1)\Bigl(S_{\scriptsize{\underbrace{1,1,\ldots,1}_{l \textnormal{ times}}},-1,a_{l+2},\ldots}(n) 																			-S_{\scriptsize{\underbrace{1,1,\ldots,1}_{l-1\textnormal{ times}}},-1,a_{l+2},\ldots}(n)\Bigr. \nonumber\\															&&\Bigl.+\cdots+(-1)^{l}S_{-1,a_{l+2},\ldots}(n)\Bigr)\nonumber\\
					&&+\frac{(-1)^{l+1}}{2}\Bigl(S_{-a_{l+2},a_{l+3},\ldots}(n)+(-1)^n S_{a_{l+2},a_{l+3},\ldots}(n)\Bigr);\nonumber
				\end{eqnarray}
		\item[$a_j=1, \ 1\leq j\leq l, \abs{a_{l+1}}\geq 2:$]
				\begin{eqnarray}
					\sum_{i=1}^n{S_{\scriptsize{\underbrace{1,1,\ldots,1}_{l \textnormal{ times}}},a_{l+1},a_{l+2},\ldots}(i)}&=& 																			(n+1)\Bigl(S_{\scriptsize{\underbrace{1,1,\ldots,1}_{l \textnormal{ times}}},,a_{l+1},a_{l+2},\ldots}(n) 																-S_{\scriptsize{\underbrace{1,1,\ldots,1}_{l-1\textnormal{ times}}},a_{l+1},a_{l+2},\ldots}(n)\Bigr. \nonumber\\	 																						&&\Bigl.+\cdots+(-1)^{l}S_{,a_{l+1},a_{l+2},\ldots}(n)\Bigr)\nonumber\\
					&&+(-1)^{l+1}S_{\sign{a_{l+1}}(\abs{a_{l+1}}-1),a_{l+2},\ldots}(n).\nonumber
				\end{eqnarray}													
\end{description}
\end{thm}

\begin{proof}
We give a proof of the first and the third identity all the other identities follow similarly. First let $\abs{a_1}\geq 2$:
\begin{eqnarray*}
\sum_{i=1}^n{S_{a_1,a_2,\ldots}(i)}&=&\sum_{i=1}^n{\sum_{j=1}^i{\frac{\sign{a_1}^j}{j^{\abs{a_1}}} S_{a_2,a_3\ldots}(j)}}\\
&=&\sum_{j=1}^n{\frac{\sign{a_1}^j}{j^{\abs{a_1}}} S_{a_2,a_3\ldots}(j) \sum_{i=j}^n{1}}\\
&=&\sum_{j=1}^n{\frac{\sign{a_1}^j}{j^{\abs{a_1}}} S_{a_2,a_3\ldots}(j) (n+1-j)}\\
&=&(n+1)S_{a_1,a_2,\ldots}(n)-S_{\sign{a_1}(\abs{a_1}-1),a_2,\ldots}(n).
\end{eqnarray*}
Consider the third identity, \ie $a_j=1, \ 1\leq j\leq l$:
We proceed by induction on the depth $l.$ For $l=1$, the right hand side is $(n+1)S_{1}(n)-n$ and the left hand side gives:
\begin{eqnarray*}
\sum_{i=1}^n{S_{1}(i)}&=&\sum_{i=1}^n{\sum_{j=1}^i{\frac{1}{j}}}=\sum_{j=1}^n{\frac{1}{j}\sum_{i=j}^n{1}}=\sum_{j=1}^n{\frac{1}{j}(n+1-j)}=(n+1)S_{1}(n)-\sum_{j=1}^n{1}\\
&=&(n+1)S_{1}(n)-n.
\end{eqnarray*}
Now assume that the third identity holds for depth $\leq l-1.$ Then 
\begin{eqnarray*}
\sum_{i=1}^n{S_{\scriptsize{{\underbrace{1,1,\ldots,1}_{l \textnormal{ times}}}}}(i)}&=& \sum_{i=1}^n{\sum_{j=1}^i{\frac{1}{j} S_{\scriptsize{{\underbrace{1,1,\ldots,1}_{l-1 \textnormal{ times}}}}}(j)}}
=\sum_{j=1}^n{\frac{1}{j} S_{\scriptsize{{\underbrace{1,1,\ldots,1}_{l-1 \textnormal{ times}}}}}(j)\sum_{i=j}^n{1}}\\
&=&\sum_{j=1}^n{\frac{1}{j} S_{\scriptsize{{\underbrace{1,1,\ldots,1}_{l-1 \textnormal{ times}}}}}(j)(n+1-j)}
=(n+1)S_{\scriptsize{{\underbrace{1,1,\ldots,1}_{l \textnormal{ times}}}}}(n)-\sum_{i=1}^n{S_{\scriptsize{{\underbrace{1,1,\ldots,1}_{l-1 \textnormal{ times}}}}}(i)}\\
&=&(n+1)S_{\scriptsize{{\underbrace{1,1,\ldots,1}_{l \textnormal{ times}}}}}(n)-(n+1)\Bigl(S_{\scriptsize{\underbrace{1,1,\ldots,1}_{l-1 \textnormal{ times}}}}(n) 																												-S_{\scriptsize{\underbrace{1,1,\ldots,1}_{l-2\textnormal{ times}}}}(n) 																																+\cdots \Bigr.\\
&&\Bigl.+(-1)^{l-2}S_{1}(n)+(-1)^{l-1}\Bigr)-(-1)^{l}\\
&=&(n+1)\Bigl(S_{\scriptsize{\underbrace{1,1,\ldots,1}_{l \textnormal{ times}}}}(n) 																												-S_{\scriptsize{\underbrace{1,1,\ldots,1}_{l-1\textnormal{ times}}}}(n) 																																+\cdots+(-1)^{l-1}S_{1}(n)+(-1)^{l}\Bigr)\\
&&+(-1)^{l+1}.
\end{eqnarray*}
Therefore the identity holds for $l$ and we finished the proof of the third identity.
\end{proof}

In general, for $m \geq 1$ we can use the following theorem. 

\begin{thm}
\label{sumtheo3}
Let $k, \ n \in \N$ and $a_i \in \Z/\{0\}.$ Then
\begin{eqnarray}
\sum_{i=1}^n{i^kS_{a_1,a_2,\ldots}(i)}&=&
S_{a_1,a_2,\ldots}(n)\sum_{i=1}^n{i^k}-\sum_{j=1}^n{\frac{\sign{a_1}^jS_{a_2,a_3,\ldots}(i)}{j^{\abs{a_1}}}\sum_{i=1}^{j-1}{i^k}},\nonumber \\
\sum_{i=1}^n{(-1)^ii^kS_{a_1,a_2,\ldots}(i)}&=&
S_{a_1,a_2,\ldots}(n)\sum_{i=1}^n{(-1)^i i^k}-\sum_{j=1}^n{\frac{S_{a_2,a_3,\ldots}(i)}{j^{a_1}}\sum_{i=1}^{j-1}{(-1)^i i^k}}.\nonumber
\end{eqnarray}
\end{thm}

\begin{proof}
We just give a proof for the first equality. The second follows analogously.
\begin{eqnarray}
\sum_{i=1}^n{i^kS_{a_1,a_2,\ldots}(i)}
																	&=&\sum_{i=1}^n{i^k \sum_{j=1}^i{\frac{\sign{a_1}^jS_{a_2,a_3,\ldots}(i)}{j^{\abs{a_1}}}}}\nonumber\\
																	&=&\sum_{j=1}^n{\frac{\sign{a_1}^jS_{a_2,a_3,\ldots}(i)}{j^{\abs{a_1}}}\sum_{i=j}^n{i^k}}\nonumber
\end{eqnarray}
\begin{eqnarray}																	
\ \ \ \ \ \												&=&\sum_{j=1}^n{\frac{\sign{a_1}^jS_{a_2,a_3,\ldots}(i)}{j^{\abs{a_1}}}
																			\left(\sum_{i=1}^n{i^k}-\sum_{i=1}^{j-1}{i^k}\right)}\nonumber\\
																	&=&S_{a_1,a_2,\ldots}(n)\sum_{i=1}^n{i^k}-\sum_{j=1}^n{\frac{\sign{a_1}^j 																																	S_{a_2,a_3,\ldots}(i)}{j^{\abs{a_1}}}\sum_{i=1}^{j-1}{i^k}}.\nonumber		  
\end{eqnarray}
\end{proof}
Since the sums $\sum_{i=1}^{n}{(-1)^i i^k}$ and $\sum_{i=1}^{n}{i^k}$ can be expressed as a polynomial in $n$, together with a factor $(-1)^n$, we can apply Theorem \ref{sumtheo3} recursively and we end up at sums over harmonic sums that can be handled in Theorem \ref{sumtheo1} or Theorem \ref{sumtheo2}. Summarizing, we can always work out the sums $\sum_{i=1}^n{r(i)S_{\ve a}(i)}$ and $\sum_{i=1}^n{(-1)^ir(i)S_{\ve a}(i)}$ where $r(i)$ is a polynomial; the result will be a combination of multiple harmonic sums in the upper summation index $n$, rational functions in $n$ and the factor $(-1)^n.$

\begin{example} For $n \in \N$
\begin{eqnarray*}
\sum_{i = 1}^{n}i^2\,\text{S}_{2,1}(i)&=&\frac{-7\,n + n^2}{12} + 
  \frac{\left( 3 + 2\,n - n^2 \right) \, \text{S}_{1}(n)}{6} - \frac{\text{S}_{1,1}(n)}{6} \\
  &&+ \frac{n\,\left( 1 + n \right) \, \left( 1 + 2\,n \right) \, \text{S}_{2,1}(n)}{6}.
\end{eqnarray*}     
\end{example}
Our method is built in in\ttfamily HarmonicSums \rmfamily with th function call TransformToSSums. 
\begin{fmma}
\In {\text{\bf TransformToSSums[}\text{\bf $\sum_{\textnormal{i}=1}^{\textnormal{n}}$}{{\textnormal{i}}^4\,\text{S}[2,1,\textnormal{i}]}\text{\bf ]}}\\
\Out {\frac{1}{720}(-44 n + 81 n^2 - 46 n^3 + 9 n^4) + \frac{1}{60}(-5 n + 2 n^3 + 4 n^3 - 3 n^4) \text{S}[1, n] + 														\frac{1}{30}\text{S}[1, 1, n] + \frac{1}{30}n (1 + n) (1 + 2 n) (-1 + 3 n + 3 n^2) \text{S}[2, 1, n]}\\
\end{fmma} 

\section{Special Rational Functions in the Summand}

Next we consider the case $r(i)=1/{(i+c)}^p$, where $c \in \N$ is a fixed number. First we look at sums of the form $S_{\ve a}(n+c),$ compare \cite{Moch2002}.
\begin{lemma}
Let $c, n, k \in \N$, $a_j \in \Z/\{0\}$ for $j \in \left\{1,2,\ldots,k\right\}.$ Then for $n\geq 0$
\begin{eqnarray}	S_{a_1,a_2,\ldots,a_k}(n+c)&=&S_{a_1,a_2,\ldots,a_k}(n)+\sum_{j=1}^c{\frac{{\sign{a_1}}^{j+n}}{(j+n)^{\abs{a_1}}}S_{a_2,\ldots,a_k}(n+j)},\label{sumlem1}
\end{eqnarray}
and $n\geq c,$
\begin{eqnarray}
S_{a_1,a_2,\ldots,a_k}(n-c)&=&S_{a_1,a_2,\ldots,a_k}(n)+\sum_{j=1}^c{\frac{{\sign{a_1}}^{j+n-c}}{(j+n-c)^{\abs{a_1}}}S_{a_2,\ldots,a_k}(n-c+j)}.
\label{sumlem2}	
\end{eqnarray}
\end{lemma}
\begin{proof}
We only proof the first identity:
\begin{eqnarray*}
  S_{a_1,a_2,\ldots,a_k}(n+c)
  					&=&\sum_{j=1}^{n+c}{\frac{{\sign{a_1}}^{j}}{j^{\abs{a_1}}}S_{a_2,\ldots,a_k}(j)}\\
  					&=&S_{a_1,a_2,\ldots,a_k}(n)+\sum_{j=n+1}^{n+c}{\frac{{\sign{a_1}}^{j}}{j^{\abs{a_1}}}S_{a_2,\ldots,a_k}(j)}\\
  					&=&S_{a_1,a_2,\ldots,a_k}(n)+\sum_{j=1}^c{\frac{{\sign{a_1}}^{j+n}}{(j+n)^{\abs{a_1}}}S_{a_2,\ldots,a_k}(n+j)}.	
\end{eqnarray*}
\end{proof}
Since the multiple harmonic sum on the right sides of (\ref{sumlem1}) and (\ref{sumlem2}) have reduced depth, we can relate a sum of the form $S_{\ve a}(n\pm c)$ to sums of the form $S_{\ve a}(n)$ by recursive application of (\ref{sumlem1}) or (\ref{sumlem2}). Namely, we can use this lemma to work out sums of the form 
$$
\sum_{i=1}^n{\frac{S_{\ve a}(i)}{(i+c)^m}} \ \textnormal{or} \ \sum_{i=1}^n{(-1)^i\frac{S_{\ve a}(i)}{(i+c)^m}}.
$$
We shift the summation index:
$$
\sum_{i=c+1}^{n+c}{\frac{S_{\ve a}(i-c)}{i^m}} \ \textnormal{or} \ \sum_{i=c+1}^{n+c}{(-1)^i\frac{S_{\ve a}(i-c)}{i^m}}.
$$
and afterwards we use (\ref{sumlem2}) to relate the sum $S_{\ve a}(i-c)$ to $S_{\ve a}(i)$. Similarly we can work out sums with $`-`$ in the denominator \ie 
$$
\sum_{i=1+c}^n{\frac{S_{\ve a}(i)}{(i-c)^m}} \ \textnormal{or} \ \sum_{i=1+c}^n{(-1)^i\frac{S_{\ve a}(i)}{(i-c)^m}}.
$$

\begin{example} For $\in in \N,$
\begin{eqnarray*}
\sum_{i=1}^n{(-1)^i\frac{S_{2,3}(i)}{i+1}}&=&
 		\sum_{i=2}^{n+1}{(-1)^{i+1}\frac{S_{2,3}(i-1)}{i}}=-\sum_{i=2}^{n+1}{(-1)^{i}\frac{-\frac{S_{3}(i)}{i^2}+S_{2,3}(i)}{i}}\\
&=&\sum_{i=2}^{n+1}{\left(\frac{{\left( -1 \right) }^i\,S_3(i)}{i^3}-\frac{{\left( -1 \right) }^i\,S_{2,3}(i)}{i}\right)}
\end{eqnarray*}
\begin{eqnarray*}
\ \ \ \ \ \ \ \ \ \ \ \ \ &=&\sum_{i=2}^{n+1}{\frac{{\left( -1 \right) }^i\,S_3(i)}{i^3}}-\sum_{i=2}^{n+1}{\frac{{\left( -1 \right) }^i\,S_{2,3}(i)}{i}}\\
&=&S_{-3,3}(n+1)+1-S_{-3,2,3}(n+1)-1\\
&=&S_{-3,3}(n) + \frac{{\left( -1 \right) }^n\,S_{2,3}(n)}{1 + n} - S_{-1,2,3}(n).
\end{eqnarray*} 
\end{example}

\section{General Rational Functions in the Summand}
If we want to work out sums of the form (\ref{sumsum1}) or (\ref{sumsum2}) for a general rational function $r(n)=\frac{p(n)}{q(n)},$ where $p(n)$ and $q(n)$ are polynomials, we can combine these strategies. We start as follows:\\

If the degree of $p$ is greater than the degree of $q$ we compute polynomials $\overline{r}(n)$ and $\overline{p}(n)$ such 							that $r(n)=\overline{r}(n)+\frac{\overline{p}(n)}{q(n)}$ and the degree of $\overline{p}(n)$ is smaller than the degree of $q.$
We split the sum into two parts, $\ie$ into 
$$\sum_{i=1}^n{r(i)S_{\ve a}(i)}=\sum_{i=1}^n{\overline{r}(i)S_{\ve a}(i)}+\sum_{i=1}^n{\frac{\overline{p}(i)}{q(i)}S_{\ve a}(i)}$$
or
$$\sum_{i=1}^n{(-1)^ir(i)S_{\ve a}(i)}=\sum_{i=1}^n{(-1)^i\overline{r}(i)S_{\ve a}(i)}+\sum_{i=1}^n{(-1)^i\frac{\overline{p}(i)}{q(i)}S_{\ve a}(i)}.$$
The first sum can be done using Theorems \ref{sumtheo1}, \ref{sumtheo2} and 	\ref{sumtheo3}. For the second sum we proceed as follows:
\begin{enumerate}				
	\item Factorize the denominator $q(n)$ over $\Q$.
	\item Let $A$ be the product of all factors of the form $(i+c)^m$ with $m\in \N$ and $c \in \Z,$ and let $B$ be the product of all 						remaining factors such that $A(n)B(n)=q(n).$
	\item For example with the extended Euclidean algorithm (see Remark \ref{exteuc}) we compute polynomials $s(n)$ and $t(n)$ such that 					$s(n)A(n)+t(n)B(n)=\overline{p}(n)$. This is always possible since $A$ and $B$ are relatively prime. Hence we get 
				$$\frac{\overline{p}}{q}=\frac{t}{A}+\frac{s}{B}.$$
	\item Split the sum into two sums, each sum over one fraction, $\ie$			
				$$\sum_{i=1}^n{\frac{\overline{p}(i)}{q(i)}S_{\ve a}(i)}=\sum_{i=1}^n{\frac{t}{A}S_{\ve a}(i)}+\sum_{i=1}^n{\frac{s}{B}S_{\ve 					a}(i)}$$
				or
				$$\sum_{i=1}^n{(-1)^i\frac{\overline{p}(i)}{q(i)}S_{\ve a}(i)}=\sum_{i=1}^n{(-1)^i\frac{t}{A}S_{\ve 																		a}(i)}+\sum_{i=1}^n{(-1)^i\frac{s}{B}S_{\ve a}(i)}.$$											
	\item The sum with the denominator $B$ remains untouched (for details see Remark \ref{denomB}). We can now do a complete partial 							fraction decomposition to the first summand and split the sum such that we sum over each fraction separately.
	\item Each of these new sums can be expressed in terms of harmonic sums following Subsection \ref{polyinthesum}.
	\item We end up in a combination of rational functions in $n$, harmonic sums with upper index $n$, the factor $(-1)^n$ and perhaps a 					sum over the fraction with denominator $B.$
\end{enumerate}

\begin{remark}
In our implementation the sum $\sum_{i=1}^n{\frac{s}{B}S_{\ve a}(i)}$ is passed further to Schneider's \ttfamily Sigma \rmfamily package. The underlying difference field and difference ring algorithms \cite{Schneider2007a,Schneider2008,Schneider2008a} can simplify those sums further to sum expressions where the denominator has minimal degree. The result of \ttfamily Sigma \rmfamily is again passed to the package \ttfamily HarmonicSums\rmfamily. If possible, it finds a closed form in terms of harmonic sums.
\label{denomB}
\end{remark}

\begin{remark}
In our implementation we do not use the extended Euclidean algorithm since it turned out that it is faster to construct a partial fraction decomposition with unknown coefficients. After clearing the denominators we have to solve the linear system of equations which we get by coefficient comparison.  
\label{exteuc}
\end{remark}

\begin{example}
Let us give several examples using the package \ttfamily HarmonicSums\rmfamily; the first sum is without a `bad` part:
\end{example}
\begin{fmma}
\begin{mma}
\In {\text{\bf TransformToSSums[}\text{\bf $\sum_{\textnormal{i}=1}^{\textnormal{n}}$}\frac{\textnormal{i}\,\text{S}[2,3,\textnormal{i}]}{6 + 5\,\textnormal{i} + \textnormal{i}^2}\text{\bf ]}}\\
\Out {\frac{3\,n}{4\,\left( 1 + n \right)}  - 
  \frac{3\,\text{S}[2,n]}{4} + 
  \frac{\left( -11 - n + 3\,n^2 \right) \,
     \text{S}[3,n]}{4\,
     \left( 1 + n \right) \,\left( 2 + n \right) } + 
  \frac{7\,\text{S}[4,n]}{4} + 
  \frac{\left( -19\,n - 20\,n^2 - 5\,n^3 \right) \,
     \text{S}[2,3,n]}{2\,
     \left( 1 + n \right) \,\left( 2 + n \right) \,
     \left( 3 + n \right) } - 
  \text{S}[3,3,n] + 
  \text{S}[1,2,3,n]}\\
\end{mma} 
\begin{mma}
\In {\text{\bf TransformToSSums[}\text{\bf $\sum_{\textnormal{i}=1}^{\textnormal{n}}$}\frac{\textnormal{i}\,\text{S}[2,3,\textnormal{i}]}{3 + 5\,\textnormal{i} + 2 \textnormal{i}^2}\text{\bf]}}\\
\Out {-3\, \sum_{i_1 = 1}^{n}\frac{\left( \sum_{i_2 = 1}^{i_1}\frac{1}{3 + 2\,i_2} \right) \,\text{S}[3,i_1]}{{i_1}^2}   + 
  \frac{\left( 4 + 4\,n \right) \,\left( \sum_{i_1 = 1}^{n}\frac{\text{S}[3,i_1]}{3 + 2\,i_1} \right)  - 
     2\,\text{S}[1,3,n] - 2\,n\,\text{S}[1,3,n]}{3\,
     \left( 1 + n \right)}+\frac{3\,n\,\text{S}[2,3,n] + 
     \left( \sum_{i_1 = 1}^{n}\frac{1}{3 + 2\,i_1} \right) \,\left( 9\,\text{S}[2,3,n] + 9\,n\,\text{S}[2,3,n] \right)  + 
     3\,\text{S}[3,3,n] + 3\,n\,\text{S}[3,3,n]}{3\,
     \left( 1 + n \right)}-\frac{3\,\text{S}[1,2,3,n]+3\,n\,\text{S}[1,2,3,n]}{3\,
     \left( 1 + n \right)}}\\
\end{mma}
\noindent Of course we can also work out more difficult nested sums:
\begin{mma}
\In {\text{\bf TransformToSSums[}\text{\bf{$\sum_{i=1}^n{\frac{\sum_{j=1}^i{j\frac{\sum_{k=1}^j{\frac{\sum_{l=1}^k{\frac{l}{2+l}}}{k}}}{j+j^2}}}{i(i+1)^3}}$}}\text{\bf]}}\\
\Out {\frac{45\,n + 67\,n^2 + 66\,n^3 + 48\,n^4 + 19\,n^5 + 3\,n^6}{4\,{\left( 1 + n \right) }^5\,\left( 2 + n \right) } + 
  \frac{3\,\left( 1 + 6\,n + 4\,n^2 + n^3 \right) \,\text{S}[1,n]}{2\,{\left( 1 + n \right) }^4} - 2\,\text{S}[2,n] + 
  \frac{\left( -1 - 6\,n - 4\,n^2 - n^3 \right) \,\text{S}[1,1,n]}{{\left( 1 + n \right) }^4} + 
  \frac{\left( -7 + 3\,n + 11\,n^2 + 5\,n^3 \right) \,\text{S}[2,1,n]}{2\,{\left( 1 + n \right) }^3} + 
  \frac{9\,\text{S}[3,1,n]}{2} + \frac{2\,\left( 3 + 3\,n + n^2 \right) \,\text{S}[1,1,1,n]}{{\left( 1 + n \right) }^3} - 
  3\,\text{S}[2,1,1,n] - 2\,\text{S}[2,2,1,n] - 3\,\text{S}[3,1,1,n] - 2\,\text{S}[3,2,1,n] + 
  2\,\text{S}[2,1,1,1,n] + 2\,\text{S}[3,1,1,1,n]}\\
\end{mma}
\noindent We can also work out S-sums:
\begin{mma}
\In {\text{\bf TransformToSSums[}\text{\bf{$\sum_{i=1}^n{\frac{3^{i+2}\sum_{j=1}^i{\frac{2^j}{j}}}{(i+2)^2}}$}}\text{\bf]}}\\
\Out {\frac{-3\,\left( 2 - 2^{1 + n}\,3^n + 4\,n - 6^n\,n + 2\,n^2 \right) }{{\left( 1 + n \right) }^2} + 
  \frac{3^{1 + n}\,\left( 7 + 10\,n + 4\,n^2 \right) \,\text{S}[1,\{ 2\} ,n]}{{\left( 2 + 3\,n + n^2 \right) }^2} - \frac{5\,\text{S}[1,\{ 6\} ,n]}{2} + \frac{\text{S}[2,\{ 6\} ,n]}{2} - 
  \text{S}[3,\{ 6\} ,n] + \text{S}[2,1,\{ 3,2\} ,n]}\\
\end{mma}
\noindent Of course not all sums can be rewritten in terms of harmonic sums:
\begin{mma}
\In {\text{\bf TransformToSSums[}\text{\bf{$\sum_{i=1}^n{\frac{\sum_{j=1}^i{\frac{1}{j+1}}}{i\,i!}}$}}\text{\bf]}}\\
\Out {\sum_{i=1}^n\left(-\frac{1}{i!(i+1)}+\frac{\text{S}[1,i]}{i!\,i}\right)}\\
\end{mma}
\end{fmma}

%% file: example.tex
\chapter{An Example from Particle Physics}
\label{An Example from Particle Physics}
Single scale quantities as anomalous dimensions and hard scattering cross section in renormalizable quantum field theories are found to obey difference equations of finite order in Mellin space. In \cite{Bluemlein,Bluemlein2009b} a formalism is established to construct general solutions for these quantities from a sufficiently large number of fixed moments. From the given moments a recurrence relation is established and solved in terms of the Mellin parameter $N$, finite harmonic sums and their generalizations at finite values of $N$ and in the limit $N\rightarrow\infty.$
In these computation our package \ttfamily HarmonicSums \rmfamily contributed as follows. E.g., for the $C_AC_FN_F$-term of the 3-loop non-singlet splitting function $P^{-}_{NS,2}$, see \cite{Bluemlein}, \ttfamily Sigma \rmfamily solved a recurrence of order 7 and found the following closed form evaluation \cite[Exp.~5]{Bluemlein}.
\begin{eqnarray}
&&P:=P_{NS,2,C_A C_FN_F}^{-}=\nonumber \\ 
&&\ -\frac{2 \left(1086 N^7+3258 N^6+2129 N^5-288 N^4-67
   N^3-206 N^2-156 N+144\right)}{27 N^4 (N+1)^3}\nonumber \\
&&\ \frac{32 \left(8 N^4+33 N^3+53 N^2+25 N+3\right)}{9 N
   (N+1)^4}(-1)^N
   +\frac{16}{3}
   \sum_{i=1}^N\frac{1}{i^4}+\frac{32}{3}
   \sum_{i=1}^N\frac{(-1)^i}{i^4}\nonumber \\
&&\   -\frac{16\left(10 N^2+10 N+3\right)}{9 N(N+1)}\sum_{i=1}^N\frac{(-1)^i}{i^3}+\frac{1336}{27}
   \sum_{i=1}^N\frac{1}{i^2}-\frac{64
   \left(8 N^2+8 N+3\right)}{9 N(N+1)}\sum_{i=1}^N\frac{(-1)^i}{i}\nonumber \\
&&\   +\frac{16 \left(4
   N^6+88 N^5+314 N^4+412 N^3+201 N^2+16 N-12\right)}{9 N^2
   (N+1)^2 (N+2)^2}\sum_{i=1}^N\frac{(-1)^i}{i^2}\nonumber \\
&&\   +
   \left(-\frac{8 \left(14 N^2+14 N+3\right)}{3 N
   (N+1)}-\frac{16}{3}\sum_{i=1}^N\frac{1}{i}\right)\sum_{i=1}^N\frac{1}{i^3}+\frac{64}{3}
   \sum_{i=1}^N\frac{\sum_{j=1}^i\frac{1}{j^3}}{i}
   +32\sum_{i=1}^N\frac{\sum_{j=1}^i\frac{(-1)^j}{j^2}}{(i+2)^2}\nonumber \\
&&\   -\frac{32 \left(22 N^2+22
   N-3\right)}{9 N
   (N+1)}\sum_{i=1}^N\frac{\sum_{j=1}^i\frac{(-1)^j}{j^2}}{i+2}+\left(\sum_{i=1}^N\frac{1}{i}\right)\Bigg(\frac{32\left(2 N^2+4 N+1\right)}{3
   (N+1)^3}(-1)^N\nonumber
\end{eqnarray}
\begin{eqnarray}     
&&\   -\frac{4 \left(65 N^6+195 N^5+195 N^4+137 N^3+36 N^2+36
   N+18\right)}{27 N^3 (N+1)^3}+\frac{32}{3}\sum_{i=1}^N
   \frac{(-1)^i}{i^3}\nonumber \\
&&\ +\frac{128}{3}
   \sum_{i=1}^N\frac{(-1)^i}{i}
   +\frac{32
   \left(2 N^3+2 N^2-3 N-2\right)}{3 N (N+1)(N+2)}\sum_{i=1}^N\frac{(-1)^i}{i^2}-\frac{64}{3}
   \sum_{i=1}^N\frac{\sum_{j=1}^i\frac{(-1)^j}{j^2}}{i+2}\Bigg)\nonumber \\
&&\   -\frac{256}{3}
   \sum_{i=1}^N\frac{(-1)^i
   \sum_{j=1}^i\frac{1}{j}}{i}
   +\frac{128}{3}
   \sum_{i=1}^N\frac{\left(\sum_{j=1}^i\frac{(-1)^j}{j^2}\right)\left(
   \sum_{i=1}^N\frac{1}{j}\right)}{i+2}.
\label{P}   
\end{eqnarray}

Given such sum representations, the package \ttfamily HarmonicSums \rmfamily provides several sophisticated functions to simplify such expressions. First we rewrite (\ref{P}) into an expression in terms of harmonic sums (harmonic sums are defined in Chapter \ref{Algebraic Relations between Multiple Harmonic Sums}) by using the function call $S:=$TransformToSSums[$P$] (for details on this function see Chapter \ref{Summation of Multiple Harmonic Sums}):
\begin{eqnarray}
&&S=P_{NS,2,C_A C_FN_F}^{-}=\nonumber \\ 
&&\ \frac{-2\,\left( 144 - 12\,n - 362\,n^2 - 417\,n^3 - 1603\,n^4 - 96\,{\left( -1 \right) }^n\,n^4 - 1855\,n^5 - 384\,{\left( -1 \right) }^n\,n^5 \right) }{27\,n^4\,{\left( 1 + n \right) }^4}\nonumber \\
&&\ \frac{-2\,\left( 347\,n^6 + 1080\,n^7 + 270\,n^8 \right) }{27\,n^4\,{\left( 1 + n \right) }^4}+ \frac{64\,S_{-4}(n)}{3}\nonumber \\
&&\  + \frac{16\,S_{-3}(n)\,\left( -18 - 13\,n + 54\,n^2 + 34\,n^3 + 36\,n\,S_{1}(n) + 54\,n^2\,S_{1}(n) + 18\,n^3\,S_{1}(n) \right) }
   {9\,n\,\left( 1 + n \right) \,\left( 2 + n \right) }\nonumber \\
&&\  + \frac{16\,S_{-2}(n)\,
     \left( -6 + 17\,n + 42\,n^2 + 16\,n^3 - 12\,n\,S_{1}(n) - 42\,n^2\,S_{1}(n) - 6\,n^3\,S_{1}(n) + 4\right) }{9\,
     n^2\,{\left( 1 + n \right) }^2\,\left( 2 + n \right) }\nonumber \\
&&\  + \frac{16\,S_{-2}(n)\,
     \left(48\,n^4\,S_{1}(n) + 24\,n^5\,S_{1}(n) \right) }{9\,n^2\,{\left( 1 + n \right) }^2\,\left( 2 + n \right) }     
 + \frac{1336\,S_{2}(n)}{27} - 
  \frac{8\,\left( 3 + 14\,n + 14\,n^2 \right) \,S_{3}(n)}{3\,n\,\left( 1 + n \right) }\nonumber  \\
&&\  + \frac{16\,S_{4}(n)}{3} - \frac{128\,S_{-3,1}(n)}{3} -   \frac{128\,\left( -1 + n + n^2 \right) \,S_{-2,1}(n)}{3\,\left( 1 + n \right) \,\left( 2 + n \right) } - \frac{128\,S_{1,-3}(n)}{3} \nonumber  \\
&&\ -  \frac{32\,\left( -6 + 5\,n + 42\,n^2 + 22\,n^3 \right) \,S_{1,-2}(n)}{9\,n\,\left( 1 + n \right) \,\left( 2 + n \right) }- 
  \frac{4\,S_{1}(n)\,\left( 18 + 36\,n + 36\,n^2 + 281\,n^3  \right) }{27\,n^3\,{\left( 1 + n \right) }^3}\nonumber \\
&&\ -  
  \frac{4\,S_{1}(n)\,\left( 72\,{\left( -1 \right) }^n\,n^3 + 627\,n^4 + 627\,n^5 + 209\,n^6 + 36\,n^3\,S_{3}(n) + 
       108\,n^4\,S_{3}(n) \right) }{27\,n^3\,{\left( 1 + n \right) }^3}\nonumber \\
&&\ -  
  \frac{4\,S_{1}(n)\,\left( 108\,n^5\,S_{3}(n) + 36\,n^6\,S_{3}(n) + 144\,n^3\,S_{1,-2}(n) + 432\,n^4\,S_{1,-2}(n)  \right) }{27\,n^3\,{\left( 1 + n \right) }^3}\nonumber \\
&&\ -  
  \frac{4\,S_{1}(n)\,\left(  432\,n^5\,S_{1,-2}(n) + 
       144\,n^6\,S_{1,-2}(n) \right) }{27\,n^3\,{\left( 1 + n \right) }^3} + \frac{64\,S_{1,3}(n)}{3} - \frac{32\,S_{2,-2}(n)}{3}\nonumber  \\
&&\ + 
  \frac{128\,S_{1,-2,1}(n)}{3} + \frac{128\,S_{1,1,-2}(n)}{3}.
\label{P1}   
\end{eqnarray}

In (\ref{P1}) 15 different harmonic sums appear:
\begin{eqnarray*}
&&S_{-4}(n), S_{-3}(n), S_{-2}(n), S_{1}(n), S_{2}(n), S_{3}(n), S_{4}(n), S_{-3, 1}(n), S_{-2, 1}(n), S_{1, -3}(n),\\
&&S_{1, -2}(n), S_{1, 3}(n), S_{2, -2}(n),S_{1, -2, 1}(n), S_{1, 1, -2}(n).
\end{eqnarray*}
We can use the algebraic relations between harmonic sums derived in Chapter \ref{Algebraic Relations between Multiple Harmonic Sums} to reduce the number of different harmonic sums appearing in (\ref{P1}). Namely, activating the call ReduceToBasis[$S$] we get:

\begin{eqnarray}
&&P_{NS,2,C_A C_FN_F}^{-}=\nonumber \nonumber \\
&&\		\frac{2}{27\,n^4\,{\left( 1 + n \right) }^4} 
			\left( -144 + 12\,n + 362\,n^2 + 417\,n^3 + 1603\,n^4 + 96\,{\left( -1 \right) }^n\,n^4 + 1855\,n^5 \right. \nonumber \\
&&\	 + 384\,{\left( -1 \right) }^n\,n^5 - 347\,n^6 - 1080\,n^7 - 270\,n^8 + \left( 288\,n^4 + 1152\,n^5 + 1728\,n^6 + 1152\,n^7 \right.\nonumber \\
&&\  \left.+ 288\,n^8 \right) \,S_{-4}(n) + S_{-3}(n)\,\left( -72\,n^3 - 456\,n^4 - 1176\,n^5 - 1512\,n^6 - 960\,n^7 - 240\,n^8\right.\nonumber \\
&&\  \left.+\left( 144\,n^4 + 576\,n^5 + 864\,n^6 + 576\,n^7 + 144\,n^8 \right) \,S_{1}(n) \right)  + \left( 668\,n^4 + 2672\,n^5\right.\nonumber \\
&&\  \left.+ 4008\,n^6 + 2672\,n^7 + 668\,n^8 \right) \,S_{2}(n) + S_{-2}(n)\, \left( -72\,n^2 + 96\,n^3 + 792\,n^4 + 1008\,n^5\right.\nonumber \\
&&\  + 384\,n^6 + \left( -480\,n^4 - 1920\,n^5 - 2880\,n^6 - 1920\,n^7 - 480\,n^8 \right) \,S_{1}(n) +  \left( 288\,n^4 + 1152\,n^5\right.\nonumber \\
&&\  \left. \left.+ 1728\,n^6 + 1152\,n^7 + 288\,n^8 \right) \,S_{2}(n) \right)  + \left( -108\,n^3 - 828\,n^4 - 2340\,n^5 - 3132\,n^6 - 2016\,n^7\right.\nonumber \\
&&\  \left.- 504\,n^8 \right) \,S_{3}(n) + \left( 360\,n^4 + 1440\,n^5 + 2160\,n^6 + 1440\,n^7 + 360\,n^8 \right) \,S_{4}(n)\nonumber \\
&&\  + \left( -144\,n^3 + 48\,n^4 + 1488\,n^5 + 2736\,n^6 + 1920\,n^7 + 480\,n^8 \right) \,S_{-2,1}(n) + S_{1}(n)\,\left( -36\,n\right.\nonumber \\
&&\  - 108\,n^2 - 144\,n^3 - 634\,n^4 - 144\,{\left( -1 \right) }^n\,n^4 - 1816\,n^5 - 144\,{\left( -1 \right) }^n\,n^5 - 2508\,n^6\nonumber \\
&&\  - 1672\,n^7 - 418\,n^8 + \left( 216\,n^4 + 864\,n^5 + 1296\,n^6 + 864\,n^7 + 216\,n^8 \right) \,S_{3}(n) + \left( 288\,n^4\right.\nonumber \\
&&\  \left. \left. + 1152\,n^5 + 1728\,n^6 + 1152\,n^7 + 288\,n^8 \right) \,S_{-2,1}(n) \right)  + \left( -144\,n^4 - 576\,n^5 - 864\,n^6\right.\nonumber \\
&&\  \left.- 576\,n^7 - 144\,n^8 \right) \,S_{2,-2}(n) + \left( -288\,n^4 - 1152\,n^5 - 1728\,n^6 - 1152\,n^7 - 288\,n^8 \right) \,S_{3,1}(n)\nonumber \\
&&\  \left.+ \left( -576\,n^4 - 2304\,n^5 - 3456\,n^6 - 2304\,n^7 - 576\,n^8 \right) \,S_{-2,1,1}(n) \right)      
\label{P2}   
\end{eqnarray}

In (\ref{P2}) only the following 11 different harmonic sums appear:
$$
S_{-4}(n), S_{-3}(n), S_{-2}(n), S_{1}(n), S_{2}(n), S_{3}(n), S_{4}(n), S_{-2, 1}(n), S_{2, -2}(n), S_{3, 1}(n), S_{-2, 1, 1}(n).
$$
We want to emphasize that these sums are algebraic independent. Exactly these representations were used in \cite{Bluemlein} to establish compactified representations originally computed by \cite{Moch2004} and \cite{Vermaseren2005}. Using relations originating from differentiation (for details see Chapter~\ref{Harmonic Polylogarithms}) we can go further. We use the function call ReduceToBasis[$S$, UseDifferentiation $\rightarrow$ True]: 

\begin{eqnarray}
&&P_{NS,2,C_A C_FN_F}^{-}=\nonumber \\
&&\			\frac{2}{135\,n^4\,{\left( 1 + n \right) }^4} \left( 720 - 60\,n - 1810\,n^2 - 2085\,n^3 - 8015\,n^4 - 480\,{\left( -1 \right) }^n\,n^4 - 9275\,n^5\right.\nonumber \\
&&\	 - 1920\,{\left( -1 \right) }^n\,n^5 + 1735\,n^6 + 5400\,n^7 + 1350\,n^8 - 180\,n^2\,\zeta_2 + 240\,n^3\,\zeta_2 - 1360\,n^4\,\zeta_2 \nonumber 
\end{eqnarray}
\begin{eqnarray}
&&\	 - 10840\,n^5\,\zeta_2 - 19080\,n^6\,\zeta_2 - 13360\,n^7\,\zeta_2 - 3340\,n^8\,\zeta_2 + 1224\,n^4\,{\zeta_2}^2 + 4896\,n^5\,{\zeta_2}^2 \nonumber\\
&&\	 + 7344\,n^6\,{\zeta_2}^2 + 4896\,n^7\,{\zeta_2}^2 + 1224\,n^8\,{\zeta_2}^2 + 270\,n^3\,\zeta_3 + 2430\,n^4\,\zeta_3 + 7290\,n^5\,\zeta_3\nonumber \\
&&\	 + 9990\,n^6\,\zeta_3 + 6480\,n^7\,\zeta_3 + 1620\,n^8\,\zeta_3 +  \left( 180\,n^3 + 1140\,n^4 + 2940\,n^5 + 3780\,n^6 + 2400\,n^7\right.\nonumber \\ 
&&\   + \left( 3340\,n^4 + 13360\,n^5 + 20040\,n^6 + 13360\,n^7 + 3340\,n^8 - 720\,n^4\,\zeta_2 - 2880\,n^5\,\zeta_2 - 4320\,n^6\,\zeta_2\right.\nonumber \\
&&\  \left. - 2880\,n^7\,\zeta_2 - 720\,n^8\,\zeta_2 \right) \,\text{D}[S_{1}(n),n,1] + \left( -360\,n^4 - 1440\,n^5 - 2160\,n^6\right.\nonumber\\  
&&\  \left.- 1440\,n^7 - 360\,n^8 \right) \,{\text{D}[S_{1}(n),n,1]}^2 + \text{D}[S_{-1}(n),n,1]\,\left( -360\,n^2 + 480\,n^3 + 3960\,n^4 \right.\nonumber \\
&&\  + 5040\,n^5 + 1920\,n^6 + 1440\,n^4\,\zeta_2 + 5760\,n^5\,\zeta_2 + 8640\,n^6\,\zeta_2 + 5760\,n^7\,\zeta_2 + 1440\,n^8\,\zeta_2 \nonumber \\
&&\  \left.+ \left( -1440\,n^4 - 5760\,n^5 - 8640\,n^6 - 5760\,n^7 - 1440\,n^8 \right) \,\text{D}[S_{1}(n),n,1] \right)+ \left( 270\,n^3\right.\nonumber \\
&&\  \left. + 2070\,n^4 + 5850\,n^5 + 7830\,n^6 + 5040\,n^7 + 1260\,n^8 \right) \,\text{D}[S_{1}(n),n,2] + \left( 360\,n^4 + 1440\,n^5\right.\nonumber \\
&&\  \left. + 2160\,n^6 + 1440\,n^7 + 360\,n^8 \right) \,\text{D}[S_{1}(n),n,3] + \left( -720\,n^4 - 2880\,n^5 - 4320\,n^6 - 2880\,n^7 \right.\nonumber \\
&&\  \left.- 720\,n^8 \right) \,\text{D}[S_{2,1}(n),n,1]+ \left( 720\,n^3 - 240\,n^4 - 7440\,n^5 - 13680\,n^6 - 9600\,n^7 - 2400\,n^8 \right)\nonumber \\
&&\  S_{-2,1}(n) + S_{1}(n)\,\left( 180\,n + 540\,n^2 + 720\,n^3 + 3170\,n^4 + 720\,{\left( -1 \right) }^n\,n^4 + 9080\,n^5 \right.\nonumber \\
&&\  + 720\,{\left( -1 \right) }^n\,n^5 + 12540\,n^6 + 8360\,n^7 + 2090\,n^8 - 1200\,n^4\,\zeta_2 - 4800\,n^5\,\zeta_2 - 7200\,n^6\,\zeta_2\nonumber \\
&&\  - 4800\,n^7\,\zeta_2 - 1200\,n^8\,\zeta_2 - 540\,n^4\,\zeta_3 - 2160\,n^5\,\zeta_3 - 3240\,n^6\,\zeta_3 - 2160\,n^7\,\zeta_3\nonumber \\
&&\  - 540\,n^8\,\zeta_3 + \left( -2400\,n^4 - 9600\,n^5 - 14400\,n^6 - 9600\,n^7 - 2400\,n^8 \right) \,\text{D}[S_{-1}(n),n,1]\nonumber \\
&&\	 + \left( -360\,n^4 - 1440\,n^5 - 2160\,n^6 - 1440\,n^7 - 360\,n^8 \right) \,\text{D}[S_{-1}(n),n,2] + \left( -540\,n^4 - 2160\,n^5\right.\nonumber \\
&&\  - 3240\,n^6 - 2160\,n^7 \left.- 540\,n^8 \right) \,\text{D}[S_{1}(n),n,2] + \left( -1440\,n^4 - 5760\,n^5 - 8640\,n^6 - 5760\,n^7 \right.\nonumber \\
&&\	 \left. \left.- 1440\,n^8 \right) \,S_{-2,1}(n) \right) + \left( 720\,n^4 + 2880\,n^5 + 4320\,n^6 + 2880\,n^7 + 720\,n^8 \right) \,S_{2,-2}(n)\nonumber \\
&&\  \left.+ \left( 2880\,n^4 + 11520\,n^5 + 17280\,n^6 + 11520\,n^7  + 2880\,n^8 \right) \,S_{-2,1,1}(n) \right) 
\label{P3}   
\end{eqnarray}

Note that we reduced the number of different harmonic sums to 6:
$$
S_{-1}(n), S_{1}(n), S_{-2, 1}(n), S_{2, -2}(n), S_{2, 1}(n), S_{-2, 1, 1}(n)
$$
by using in addition the differential operator $\text{D}=\frac{d}{dn}.$
Finally we can apply half-integer relation (see Chapter \ref{Half-Integer Relations}) with the function call ReduceToBasis[$S$, UseDifferentiation $\rightarrow$ True, UseHalfInteger $\rightarrow$ True]:
\begin{eqnarray}
&&P_{NS,2,C_A C_FN_F}^{-}=\nonumber \\
&&\ \frac{-2}{135\,n^4\,{\left( 1 + n \right) }^4} \left( 720 - 60\,n - 1810\,n^2 - 2085\,n^3 - 8015\,n^4 - 480\,{\left( -1 \right) }^n\,n^4 - 9275\,n^5\right.\nonumber \\
&&\ - 1920\,{\left( -1 \right) }^n\,n^5 + 1735\,n^6 + 5400\,n^7 + 1350\,n^8 - 180\,n^2\,\zeta_2 +  240\,n^3\,\zeta_2 - 1360\,n^4\,\zeta_2\nonumber \\
&&\ - 10840\,n^5\,\zeta_2 - 19080\,n^6\,\zeta_2 - 13360\,n^7\,\zeta_2 - 3340\,n^8\,\zeta_2 + 1224\,n^4\,{\zeta_2}^2 +  4896\,n^5\,{\zeta_2}^2\nonumber \\
&&\ + 7344\,n^6\,{\zeta_2}^2 + 4896\,n^7\,{\zeta_2}^2 + 1224\,n^8\,{\zeta_2}^2 + 270\,n^3\,\zeta_3 + 2430\,n^4\,\zeta_3 +  7290\,n^5\,\zeta_3\nonumber \\
&&\ + 9990\,n^6\,\zeta_3 + 6480\,n^7\,\zeta_3 + 1620\,n^8\,\zeta_3 + \left( 180\,n^3 + 1140\,n^4 + 2940\,n^5 + 3780\,n^6 + 2400\,n^7\right.\nonumber \\
&&\ \left. + 600\,n^8 \right) \,\text{D}[\text{Half}[S_{1}(n),n] - S_{1}(n),n,2] + \left( 240\,n^4 + 960\,n^5 + 1440\,n^6 + 960\,n^7\right.\nonumber \\
&&\ \left. + 240\,n^8 \right) \,\text{D}[\text{Half}[S_{1}(n),n] - S_{1}(n),n,3] +  \left( 3340\,n^4 + 13360\,n^5 + 20040\,n^6\right.\nonumber
\end{eqnarray}
\begin{eqnarray}
&&\ \left.+ 13360\,n^7 + 3340\,n^8 - 720\,n^4\,\zeta_2 - 2880\,n^5\,\zeta_2 - 4320\,n^6\,\zeta_2 - 2880\,n^7\,\zeta_2 - 720\,n^8\,\zeta_2 \right)\nonumber\\
&&\ \,\text{D}[S_{1}(n),n,1] + \left( -360\,n^4 - 1440\,n^5 - 2160\,n^6 - 1440\,n^7 - 360\,n^8 \right) \,{\text{D}[S_{1}(n),n,1]}^2\nonumber\\
&&\ +  \text{D}[\text{Half}[S_{1}(n),n] - S_{1}(n),n,1]\,\left( -360\,n^2 + 480\,n^3 + 3960\,n^4 + 5040\,n^5 + 1920\,n^6 + 1440\,n^4\,\zeta_2\right.\nonumber \\
&&\ + 5760\,n^5\,\zeta_2 +  8640\,n^6\,\zeta_2 + 5760\,n^7\,\zeta_2 + 1440\,n^8\,\zeta_2 + \left( -1440\,n^4 - 5760\,n^5 - 8640\,n^6\right.\nonumber \\
&&\ \left. \left. - 5760\,n^7 - 1440\,n^8 \right) \,\text{D}[S_{1}(n),n,1] \right)  + \left( 270\,n^3 + 2070\,n^4 + 5850\,n^5 + 7830\,n^6\right.\nonumber \\
&&\ \left.+ 5040\,n^7 + 1260\,n^8 \right) \,\text{D}[S_{1}(n),n,2] +  \left( 360\,n^4 + 1440\,n^5 + 2160\,n^6 + 1440\,n^7 + 360\,n^8 \right)\nonumber \\
&&\ \,\text{D}[S_{1}(n),n,3] + \left( -720\,n^4 - 2880\,n^5 - 4320\,n^6 - 2880\,n^7 - 720\,n^8 \right) \,\text{D}[S_{2,1}(n),n,1]\nonumber \\
&&\ + \left( 720\,n^3 - 240\,n^4 - 7440\,n^5 - 13680\,n^6 - 9600\,n^7 - 2400\,n^8 \right) \,S_{-2,1}(n) + S_{1}(n)\,\left( 180\,n\right.\nonumber \\
&&\ + 540\,n^2 + 720\,n^3 + 3170\,n^4 + 720\,{\left( -1 \right) }^n\,n^4 + 9080\,n^5 + 720\,{\left( -1 \right) }^n\,n^5 + 12540\,n^6\nonumber \\
&&\ + 8360\,n^7 + 2090\,n^8 - 1200\,n^4\,\zeta_2 - 4800\,n^5\,\zeta_2 - 7200\,n^6\,\zeta_2 - 4800\,n^7\,\zeta_2 - 1200\,n^8\,\zeta_2\nonumber \\
&&\ - 540\,n^4\,\zeta_3 - 2160\,n^5\,\zeta_3 - 3240\,n^6\,\zeta_3 -  2160\,n^7\,\zeta_3 - 540\,n^8\,\zeta_3 + \left( -2400\,n^4 - 9600\,n^5\right.\nonumber \\
&&\ \left. - 14400\,n^6 - 9600\,n^7 - 2400\,n^8 \right) \,\text{D}[\text{Half}[S_{1}(n),n] - S_{1}(n),n,1] + \left( -360\,n^4 - 1440\,n^5\right.\nonumber \\
&&\ \left. - 2160\,n^6 - 1440\,n^7 - 360\,n^8 \right) \,\text{D}[\text{Half}[S_{1}(n),n] - S_{1}(n),n,2] +  \left( -540\,n^4 - 2160\,n^5\right.\nonumber \\
&&\ \left. - 3240\,n^6 - 2160\,n^7 - 540\,n^8 \right) \,\text{D}[S_{1}(n),n,2] + \left( -1440\,n^4 - 5760\,n^5 - 8640\,n^6 - 5760\,n^7\right.\nonumber \\
&&\ \left.\left. - 1440\,n^8 \right) \,S_{-2,1}(n) \right)  + \left( 720\,n^4 + 2880\,n^5 + 4320\,n^6 + 2880\,n^7 + 720\,n^8 \right) \,S_{2,-2}(n)\nonumber \\
&&\ \left.+ \left( 2880\,n^4 + 11520\,n^5 + 17280\,n^6 + 11520\,n^7 + 2880\,n^8 \right) \,S_{-2,1,1}(n) \right)
\label{P4}   
\end{eqnarray}

In this way only the 5 multiple harmonic sums
$$
S_{1}(n), S_{-2, 1}(n), S_{2, -2}(n), S_{2, 1}(n) \textnormal{ and } S_{-2, 1, 1}(n)
$$
together with D and the half-integer function Half remain. Hence we were able to reduce the number of contributing sums from 15 in (\ref{P1}) to 5 in (\ref{P4}).

%% file: da.bbl
\begin{thebibliography}{10}

\bibitem{Ablinger2009}
J.~Ablinger, J.~Blümlein, and C.~Schneider.
\newblock In preparation.
\newblock 2009.

\bibitem{Berndt1985}
B.~Berndt.
\newblock {\em Ramanujan's Notebooks}.
\newblock Springer, 1985.

\bibitem{Bigotte2002}
M.~Bigotte, G.~Jacob, and M.~Petitot N.E.~Oussous.
\newblock Lyndon words and shuffle algebras for generating the coloured
  multiple zeta values relations tables.
\newblock {\em Theoretical Computer Science}, 273:271--282, 2002.

\bibitem{Bluemlein2000}
J.~Blümlein.
\newblock Analytic continuation of mellin transforms up to two-loop order.
\newblock {\em Comput. Phys. Commun.}, 133:76 [arXiv:hep--ph/0003100], 2000.

\bibitem{Bluemlein2004}
J.~Blümlein.
\newblock Algebraic relations between harmonic sums and associated quantities.
\newblock {\em Comput. Phys. Commun.}, 159:19 [arXiv:hep--ph/0311046], 2004.

\bibitem{Bluemlein2008}
J.~Blümlein.
\newblock Structural relations between nested harmonic sums.
\newblock {\em Nuclear Physics B (Proc. Suppl.)}, 183:232--237, 2008.

\bibitem{Bluemlein2009a}
J.~Blümlein.
\newblock Structural relations of harmonic sums and mellin transforms at weight
  w=6.
\newblock {\em DESY 08-206}, [arXiv.org:0901.0837], 2009.

\bibitem{Bluemlein2009}
J.~Blümlein.
\newblock Structural relations of harmonic sums and mellin transforms up to
  weight w=5.
\newblock {\em DESY 08-206 07-042}, [arXiv.org:0901.3106], 2009.

\bibitem{Bluemlein}
J.~Blümlein, M.~Kauers, S.~Klein, and C.~Schneider.
\newblock Determining the closed forms of the $o(a_s^3)$ anomalous dimensions
  and wilson coefficients from mellin moments by means of computer algebra.
\newblock [arXiv:hep-ph/0902.4091], 2009.

\bibitem{Bluemlein2009b}
J.~Blümlein, M.~Kauers, S.~Klein, and C.~Schneider.
\newblock From moments to functions in quantum chromodynamics.
\newblock DESY-09-011, [arXiv:hep-ph/0902.4095], 2009.

\bibitem{Bluemlein1999}
J.~Blümlein and S.~Kurth.
\newblock Harmonic sums and mellin transforms up to two-loop order.
\newblock {\em Phys. Rev.}, D60 014018:arXiv:hep--ph/9810241v2, 1999.

\bibitem{Bluemlein2005}
J.~Blümlein and S.~Moch.
\newblock Analytic continuation of the harmonic sums for the 3- loop anomalous
  dimensions.
\newblock {\em Phys. Lett. B}, 614:53 [arXiv:hep--ph/0503188], 2005.

\bibitem{Buchberger1983}
B.~Buchberger, G.~E. Collins, R.~Loos, and R.~Albrecht, editors.
\newblock {\em Computer algebra symbolic and algebraic computation}.
\newblock Springer-Verlag, 1983.

\bibitem{Chen1958}
K.T. Chen, R.H. Fox, and R.C. Lyndon.
\newblock Free differential calculus {I}{V}. the quotient groups of the lower
  central series.
\newblock {\em Ann. Math.}, 58:81--95, 1958.

\bibitem{Bruno1881}
F.~Faà di~Bruno.
\newblock {\em Einleitung in die {T}heorie der {B}inären {F}ormen}.
\newblock Teubner, 1881.

\bibitem{Dilcher1995}
K.~Dilcher.
\newblock Some q-series identities related to divisor functions.
\newblock {\em Discr. Math.}, 145:83--93, 1995.

\bibitem{Euler1775}
L.~Euler.
\newblock {\em Novi Comm. Acad. Dci. Petropol}, 20:140, 1775.

\bibitem{Heuser2003}
H.~Heuser.
\newblock {\em Lehrbuch der {A}nalysis, {T}eil 1}.
\newblock Teubner, 2003.

\bibitem{Hoffman1992}
M.~Hoffman.
\newblock Multiple harmonic series.
\newblock {\em Pacific J. Math.}, 152:275--290, 1992.

\bibitem{Hoffman1997}
M.~Hoffman.
\newblock The algebra of multiple harmonic series.
\newblock {\em J. Algebra}, 194:477--495, 1997.

\bibitem{Hoffman}
M.~Hoffman.
\newblock Quasi-shuffle products.
\newblock {\em J. Algebraic Combin.}, 11:49--68, 2000.

\bibitem{Karr1981}
M.~Karr.
\newblock Summation in finite terms.
\newblock {\em J. ACM}, 28:305--350, 1981.

\bibitem{Kirschenhofer1996}
P.~Kirschenhofer.
\newblock A note on alternating sums.
\newblock {\em Electron. J. Combin}, 3, 1996.

\bibitem{Lappo-Danielevsky1953}
J.A. Lappo-Danielevsky.
\newblock {\em Mémoirs sur la Théorie des Systèms Différentielles Linéaires}.
\newblock Chelsea Publishing Company, 1953.

\bibitem{Lothaire1983}
M.~Lothaire.
\newblock Combinatorics on words.
\newblock {\em Encyclopedia of Mathematics and its Applications}, 17:63--99,
  1983.

\bibitem{Maitre2006}
D.~Maître.
\newblock Hpl, a {M}athematica implementation of the harmonic polylogarithms.
\newblock {\em Comput. Phys. Commun.}, 174:222, 2006.

\bibitem{Minh2000}
Hoang~Ngoc Minh and J.~van der~Hoeven M.~Petitot.
\newblock Structure and asymptotic expansion of multiple harmonic sums.
\newblock {\em Discr. Math.}, page 100, 2000.

\bibitem{Moch2002}
S.~Moch, P.~Uwer, and S.~Weinzierl.
\newblock Nested sums, expansion of transcendental functions, and multiscale
  multiloop integrals.
\newblock {\em Journal of Math. Physics}, 43:3363, 2002.

\bibitem{Moch2004}
S.~Moch, J.A.M. Vermaseren, and A.~Vogt.
\newblock The three-loop splitting functions in qcd: The non-singlet case.
\newblock {\em Nucl. Phys. B}, 688:101 [arXiv:hep--ph/0403192], 2004.

\bibitem{Paris2001}
R.~B. Paris and D.~Kaminski.
\newblock {\em Asymptotics and Mallin-Barnes Integrals}.
\newblock Cambridge University Press, 2001.

\bibitem{Poincare1884}
H.~Poincaré.
\newblock {\em Acta Math 4}, 4:201, 1884.

\bibitem{Radford1979}
D.~Radford.
\newblock A natural ring basis for the shuffle algebra and an application to
  group schemes.
\newblock {\em J. Algebra}, 58:432, 1979.

\bibitem{Remiddi2000}
E.~Remiddi and J.A.M. Vermaseren.
\newblock Harmonic polylogarithms.
\newblock {\em Int. J. Mod. Phys.}, A 15:725 [arXiv:hep--ph/9905237], 2000.

\bibitem{Reutenauer1969}
C.~Reutenauer.
\newblock {\em Free Lie Algebras}.
\newblock Oxford University Press, 1969.

\bibitem{Savio}
D.~Y. Savio, E.~A. Lamagna, and Shing-Min Liu.
\newblock Summation of harmonic numbers.

\bibitem{Schneider2007}
C.~Schneider.
\newblock Symbolic summation assists combinatorics.
\newblock {\em Sém. Lothar. Combin.}, 56:1--36, 2007.

\bibitem{Schneider2008a}
C.~Schneider.
\newblock A refined difference field theory for symbolic summation.
\newblock {\em J. Symbolic Comput.}, 43(9):611--644, 2008.

\bibitem{Schneider2008}
C.~Schneider.
\newblock Parameterized telescoping proves algebraic independence of sums.
\newblock {\em Ann. Comb.}, To appear, 2009.

\bibitem{Schneider2007a}
C.~Schneider.
\newblock A symbolic summation approach to find optimal nested sum
  representation.
\newblock {\em Clay Math. Proc.}, To appear, 2009.

\bibitem{Temme1996}
N.~Temme.
\newblock {\em Special Functions. An Introduction to the Classical Functions of
  Mathematical Physics}.
\newblock John Wiley \& Sons, 1996.

\bibitem{Vermaseren1998}
J.A.M. Vermaseren.
\newblock Harmonic sums, {M}ellin transforms and integrals.
\newblock {\em Int.J.Mod.Phys.}, A14:2037--2076, 1999.

\bibitem{Vermaseren2005}
J.A.M. Vermaseren, A.~Vogt, and S.~Moch.
\newblock The third-order {QCD} corrections to deep-inelastic scattering by
  photon exchange.
\newblock {\em Nucl. Phys. B}, 724:3 [arXiv:hep--ph/0504242], 2005.

\bibitem{Witt1937}
E.~Witt.
\newblock Treue {D}arstellung {L}iescher {R}inge.
\newblock {\em J. Reine Angew. Math}, 177:152, 1937.

\bibitem{Witt1956}
E.~Witt.
\newblock Die {U}nterringe der freien {L}ieschen {R}inge.
\newblock {\em Math. Z.}, 64:195, 1956.

\bibitem{Zagier1994}
D.~Zagier.
\newblock {\em Values of zeta functions and their applications}, page 497.
\newblock First European Congress of Mathematics, {V}ol II, 1994.

\end{thebibliography}
